\documentclass[a4paper,UKenglish,cleveref, autoref,numberwithinsect,thm-restate]{lipics-v2021}

\nolinenumbers 

\usepackage{stmaryrd}
\usepackage{amssymb}
\usepackage{amsmath}
\usepackage{todonotes}
\usepackage[linesnumbered]{algorithm2e}
\usepackage{amsmath}
\usepackage{subcaption}
\usepackage{hyperref}

\usepackage{mathtools}

\usepackage{amsmath,amsfonts,amssymb}

\usepackage{thmtools}

\EventEditors{Siddharth Barman and S{\l}awomir Lasota}
\EventNoEds{2}
\EventLongTitle{44th IARCS Annual Conference on Foundations of Software Technology and Theoretical Computer Science (FSTTCS 2024)}
\EventShortTitle{FSTTCS 2024}
\EventAcronym{FSTTCS}
\EventYear{2024}
\EventDate{December 16--18, 2024}
\EventLocation{Gandhinagar, Gujarat, India}
\EventLogo{}
\SeriesVolume{323}
\ArticleNo{15}

\newcommand{\KWD}[1]{\ensuremath{\textit{#1}}\xspace}
\newcommand{\PI}{\KWD{PI}}
\newcommand{\IN}{\KWD{in}}

\newcommand{\LogicName}[1]{\ensuremath{\text{\textup{\sffamily #1}}}\xspace}

\newcommand{\MSO}{\LogicName{S1S}}
\newcommand{\MSOE}{\LogicName{S1S[E]}}
\newcommand{\HU}{\LogicName{H$_\mu$}}
\newcommand{\KLTL}{\LogicName{KLTL}}
\newcommand{\HPDL}{\LogicName{HyperPDL$-\Delta$}}
\newcommand{\HQPTL}{\LogicName{HyperQPTL}}
\newcommand{\PDL}{\LogicName{PDL}}
\newcommand{\AHLTL}{\LogicName{A-HyperLTL}}

\newcommand{\HAA}{\LogicName{HAA}}
\newcommand{\SNBA}{\LogicName{SNBA}}
\newcommand{\NBA}{\LogicName{NBA}}
\newcommand{\FOE}{\LogicName{FO[$<$,E]}}

\newcommand{\HCTLStar}{\LogicName{HyperCTL$^{*}$}}
\newcommand{\HLTL}{\LogicName{HyperLTL}}
\newcommand{\CHLTL}{\LogicName{HyperLTL$_{\textsf{C}}$}}
\newcommand{\SHLTL}{\LogicName{HyperLTL$_{\textsf{S}}$}}
\newcommand{\GHLTL}{\LogicName{GHyperLTL$_{\textsf{S}+\textsf{C}}$}}
\newcommand{\SGHLTL}[1]{\LogicName{SHyperLTL$_{\textsf{S}+\textsf{C}}^{#1}$}}
\newcommand{\LTL}{\LogicName{LTL}}
\newcommand{\CTL}{\LogicName{CTL}}
\newcommand{\CTLStar}{\LogicName{CTL$^{*}$}}
\newcommand{\PLTL}{\LogicName{PLTL}}
\newcommand{\QPTL}{\LogicName{QPTL}}

\newcommand{\Impl}{\rightarrow}
\newcommand{\tempfont}[1]{\textbf{#1}}
\newcommand{\tempbin}[1]{\mathbin{\mathbf{#1}}}

\newcommand{\unt}{\tempbin{U}}
\newcommand{\Until}{\unt}

\newcommand{\nxt}{\tempfont{X}}
\newcommand{\Next}{\nxt}

\newcommand{\future}{\tempbin{F}}
\newcommand{\Future}{\future}
\newcommand{\always}{\tempbin{G}}
\newcommand{\Always}{\always}
\newcommand{\Release}{\tempbin{R}}

\newcommand{\since}{\tempbin{S}}
\newcommand{\Since}{\since}
\newcommand{\PastRelease}{\tempbin{P}}

\newcommand{\once}{\tempbin{O}}
\newcommand{\Once}{\once}
\newcommand{\historically}{\tempbin{H}}
\newcommand{\Historically}{\historically}
\newcommand{\yesterday}{\tempfont{Y}}
\newcommand{\Yesterday}{\yesterday}

\newcommand{\balways}[1]{\tempbin{G}}

\newcommand{\know}{\tempfont{K}}

\newcommand{\trace}{\sigma}
\newcommand{\vartrace}{x}
\newcommand{\vartraceAux}{y}

\newcommand{\Traces}{\Pi}
\newcommand{\TracesMap}{\Traces}

\newcommand{\ActTraces}{\mathcal{L}}
\newcommand{\Lang}{\ActTraces}

\newcommand{\ctx}[1]{\langle #1 \rangle}

\newcommand{\Ctx}{C}
\newcommand{\Undef}{\texttt{und}}

\newcommand{\mktuple}[1]{\langle #1 \rangle}

\newcommand{\kstr}{\mathcal{K}}
\newcommand{\KS}{\kstr}
\newcommand{\States}{\ensuremath{S}}

\newcommand{\state}{\ensuremath{s}}
\newcommand{\Path}{\ensuremath{\pi}}
\newcommand{\Lab}{\ensuremath{Lab}}
\newcommand{\FStates}{\ensuremath{F}}
\newcommand{\Trans}{\ensuremath{E}}

\newcommand{\AP}{\textsf{AP}}

\newcommand{\nat}{\mathbb{N}}

\newcommand{\Rel}[2]{\ensuremath{#1[#2]}}

 \newcommand{\Var}{\textsf{VAR}}
 \newcommand{\Halt}{\textsf{Halt}}
\newcommand{\Dom}{{\textit{Dom}}}
\newcommand{\SUCC}{\textit{succ}}
\newcommand{\PRED}{\textit{pred}}
\newcommand{\Pt}{\textsf{P}}
\newcommand{\ie}{i.e.\xspace}

 \newcommand{\stfr}{\textit{stfr}}
  \newcommand{\pad}{\#}
  \newcommand{\TMap}{\textsf{T}}
  \newcommand{\Qf}{\textsf{Q}}
    \newcommand{\Frag}{\mathcal{F}}
  \newcommand{\Tag}{\textit{tag}}
\newcommand{\en}{\textit{en}}
\newcommand{\ini}{\textit{in}}
\newcommand{\dir}{\textit{dir}}

\newcommand{\details}[1]{}

\def\EXPSPACE{{\sc EXPSPACE}}
\def\NLOGSPACE{{\sffamily NLOGSPACE}}

\newcommand{\Au}{\mathcal{A}}
\newcommand{\Family}{\mathcal{F}}
\newcommand{\PosBool}{\ensuremath{\mathcal{B}^{+}}}
\newcommand{\true}{\texttt{true}}
\newcommand{\false}{\texttt{false}}
\newcommand{\FamStratum}{\ensuremath{\mathsf{F}}}
\newcommand{\Bu}{\ensuremath{\textsf{B}}}
\newcommand{\Co}{\ensuremath{\textsf{C}}}
\newcommand{\tran}{\texttt{t}}
\newcommand{\Acc}{\textit{Acc}}
\newcommand{\cl}{{\textit{cl}}}
\newcommand{\at}{{\textit{at}}}
\newcommand{\TransR}{\ensuremath{R}}
\newcommand{\Op}{\ensuremath{\mathcal{O}}}
\newcommand{\Plant}{\ensuremath{\mathcal{P}}}
\newcommand{\D}{\ensuremath{\mathcal{D}}}
\newcommand{\TransA}{\delta}
\newcommand{\acc}{\textit{acc}}
\newcommand{\Inf}{\textit{Inf}}
\newcommand{\dual}[1]{\ensuremath{\widetilde{#1}}} 
\newcommand{\sad}{\textit{sad}}
\newcommand{\Sub}{\textit{Sub}}
\newcommand{\Eq}{\textsf{E}}
\newcommand{\Ag}{\textsf{Agts}}
\newcommand{\Obs}{\textsf{Obs}}
\newcommand{\ObsPt}{\textsf{ObsPt}}
\newcommand{\Int}{\textsf{Int}}
\newcommand{\HMap}{\textsf{H}}

\newcommand{\Tower}{\mathsf{Tower}}

\newcommand{\DefinedAs}{\ensuremath{\,:=\,}}

\title{Unifying Asynchronous Logics for Hyperproperties}

\def\orcidID#1{\smash{\href{http://orcid.org/#1}{\protect\raisebox{-1.25pt}{\protect\includegraphics{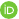}}}}}

\author{Alberto Bombardelli}{FBK, Povo TN, Italy}{abombardelli@fbk.eu}{http://orcid.org/0000-0003-3385-3205}{}
\author{Laura Bozzelli}{University of Napoli ``Federico II'', Napoli, Italy}{}{http://orcid.org/0000-0003-0963-8169}{}
\author{C\'esar S\'anchez}{IMDEA Software Institute, Madrid, Spain}{cesar.sanchez@imdea.org}{http://orcid.org/0000-0003-3927-4773}{}
\author{Stefano Tonetta}{FBK, Povo TN, Italy}{tonettas@fbk.eu}{http://orcid.org/0000-0001-9091-7899}{}

\Copyright{Alberto Bombardelli, Laura Bozzelli, C\'esar S\'anchez and Stefano Tonetta}

\authorrunning{A. Bombadelli et al.}



\begin{document}
\maketitle
\begin{abstract}
  We introduce and investigate a powerful hyper logical framework in
  the linear-time setting that we call \emph{generalized $\HLTL$ with
    stuttering and contexts} ($\GHLTL$ for short).
  $\GHLTL$ unifies the asynchronous extensions of $\HLTL$ called
  $\SHLTL$ and $\CHLTL$, and the well-known extension $\KLTL$ of
  $\LTL$ with knowledge modalities under both the synchronous and
  asynchronous perfect recall semantics.
  As a main contribution, we identify a meaningful
  fragment of $\GHLTL$, that we call \emph{simple} $\GHLTL$, with a
  decidable model-checking problem, which is more expressive than
  $\HLTL$ and known fragments of asynchronous extensions of $\HLTL$
  with a decidable model-checking problem.  Simple $\GHLTL$ subsumes
  $\KLTL$ under the synchronous semantics and the one-agent fragment
  of $\KLTL$ under the asynchronous semantics and to the best of our
  knowledge, it represents the unique hyper logic with a decidable
  model-checking problem which can express powerful non-regular trace
  properties when interpreted on singleton sets of traces. We justify
  the relevance of simple $\GHLTL$ by showing that it can express
  diagnosability properties, interesting classes of information-flow
  security policies, both in the synchronous and asynchronous
  settings, and bounded termination (more in general, global
  promptness in the style of Prompt $\LTL$).
  %
  %
  \end{abstract}

  \section{Introduction}

Temporal logics~\cite{MP92} play a fundamental role in the  formal
verification  of the dynamic behaviour of complex reactive systems.
Classic \emph{regular} temporal logics such as $\LTL$, $\CTL$, and
$\CTLStar$~\cite{Pnueli77,EmersonH86} are suited for the specification of
 \emph{trace properties} which describe  the ordering of events along individual execution
traces of a system. In the last 15 years, a novel specification paradigm has been introduced
that generalizes traditional regular trace properties by properties of
sets of traces, the so called
\emph{hyperproperties}~\cite{ClarksonS10}.
Hyperproperties relate distinct traces and are useful to formalize a wide range of properties
of prime interest which go, in general,
 beyond regular properties and cannot be expressed in standard regular temporal logics.
 A relevant example  concerns information-flow security policies like
noninterference~\cite{goguen1982security,McLean96} and observational
determinism~\cite{ZdancewicM03}  which compare observations made by an
external low-security agent along traces resulting from different
values of not directly observable inputs. Other examples include bounded termination of programs,
diagnosability of critical systems (which amounts to checking whether the available sensor  information is sufficient
to infer the presence of faults on the hidden behaviour of the system)~\cite{SampathSLST95,BozzanoCGT15,BittnerBCGTV22}, and epistemic properties describing the
 knowledge  of agents in distributed systems \cite{HalpernV86,MeydenS99,HalpernO08}.

 In the context of model checking of finite-state reactive systems, many temporal logics for hyperproperties have been proposed~\cite{DimitrovaFKRS12,ClarksonFKMRS14,BozzelliMP15,Rabe2016,FinkbeinerH16,CoenenFHH19,GutsfeldMO20}
for which model checking is decidable, including
$\HLTL$~\cite{ClarksonFKMRS14}, $\HCTLStar$~\cite{ClarksonFKMRS14},
$\HQPTL$~\cite{Rabe2016,CoenenFHH19}, and $\HPDL$~\cite{GutsfeldMO20}
which extend $\LTL$, $\CTLStar$, $\QPTL$~\cite{SistlaVW87}, and
$\PDL$~\cite{FischerL79}, respectively, by explicit first-order
quantification over traces and trace variables to refer to multiple
traces at the same time. The semantics of all these logics is \emph{synchronous}:
the temporal modalities are evaluated by a
lockstepwise traversal of all the traces assigned to the quantified
trace variables. Other  approaches for the formalization of synchronous hyper logics are either based on hyper variants
of monadic second-order logic over traces or trees~\cite{CoenenFHH19}, or the
adoption of a \emph{team semantics} for standard temporal logics, in particular, $\LTL$~\cite{KrebsMV018,Luck20,VirtemaHFK021}.
For the first approach in the linear-time setting,
we recall the logic $\MSOE$~\cite{CoenenFHH19} (and its first-order fragment $\FOE$~\cite{Finkbeiner017})
which syntactically  extends monadic second-order logic of one successor $\MSO$ with the \emph{equal-level predicate}
$\Eq$, which relates the same time point on different traces. More recently, an extension of $\HLTL$ with second-order quantification over traces
has been introduced~\cite{BeutnerFFM23} which allows to express common knowledge in multi-agent distributed systems. Like $\MSOE$, model checking of this extension
  of $\HLTL$ is highly undecidable~\cite{BeutnerFFM23}. 

Hyper logics supporting asynchronous features have been introduced recently~\cite{GutsfeldOO21,BaumeisterCBFS21,BozzelliPS21}.
These logics allow to relate traces at distinct time points which can be arbitrarily far
from each other. Asynchronicity is ubiquitous in many real-world systems, for example, in multithreaded environments in which threads are not scheduled lockstepwise,
and traces associated with distinct threads progress with different speed.
Asynchronous hyperproperties  are also useful in information-flow security and diagnosability settings where an
observer cannot distinguish
consecutive time points along an execution having the same
observations. This requires to match asynchronously  sequences of observations along
distinct execution traces.
 The first systematic study of asynchronous hyperproperties was done by Gutsfeld et al.~\cite{GutsfeldOO21}, who introduced
 the  temporal fixpoint calculus $\HU$ and its automata-theoretic counterpart for expressing such properties in the linear-time
 setting.

 More recently, three temporal logics~\cite{BaumeisterCBFS21,BozzelliPS21} which syntactically extend \HLTL\ have been introduced
for expressing  asynchronous hyperproperties:  \emph{Asynchronous $\HLTL$}
($\AHLTL$)~\cite{BaumeisterCBFS21} and  \emph{Stuttering $\HLTL$} (\SHLTL)~\cite{BozzelliPS21}, both
useful for asynchronous security analysis, and \emph{Context $\HLTL$} (\CHLTL)~\cite{BozzelliPS21}, useful for expressing hyper-bounded-time response
requirements. The logic  $\AHLTL$, which is expressively incomparable with both $\HLTL$ and $\SHLTL$~\cite{BozzelliPS22}, models asynchronicity   by means of an additional quantification layer over the so called \emph{trajectories} which control  the relative speed at which traces progress by choosing
 at each instant which traces move and which traces stutter. On the other hand, the logic $\SHLTL$
 exploits relativized versions of
the temporal modalities with respect to finite sets $\Gamma$ of $\LTL$
formulas: these modalities are evaluated by a lockstepwise traversal of the sub-traces of the
given traces which are obtained by removing ``redundant'' positions
with respect to the pointwise evaluation of the $\LTL$ formulas in
$\Gamma$. Finally, the logic $\CHLTL$ is more expressive than $\HLTL$ and is  not expressively subsumed  by either $\AHLTL$ or $\SHLTL$~\cite{BozzelliPS22}.  $\CHLTL$ extends $\HLTL$ by unary modalities $\ctx{C}$ parameterized by
a non-empty subset $C$ of trace variables---called the
\emph{context}---which restrict the evaluation of the temporal
modalities to the traces associated with the variables in $C$.
Note that the temporal modalities in $\CHLTL$ are evaluated by a lockstepwise
 traversal of the  traces assigned to the variables in the current context, and unlike $\HLTL$, the current time points of these traces from which
   the evaluation starts are in general different.
  It is known that these three syntactical extensions of $\HLTL$ are less expressive than
$\HU$~\cite{BozzelliPS22} and like $\HU$, model checking  the respective quantifier  alternation-free
fragments are already undecidable~\cite{BaumeisterCBFS21,BozzelliPS21}. The works~\cite{BaumeisterCBFS21,BozzelliPS21} identify practical fragments
of the logics $\AHLTL$ and $\SHLTL$ with a decidable model checking problem. In particular, we recall the so called  \emph{simple fragment} of $\SHLTL$~\cite{BozzelliPS21}, which is more expressive than $\HLTL$~\cite{BozzelliPS22} and can specify interesting security policies in both the asynchronous and synchronous settings.

Formalization of asynchronous hyperproperties in the \emph{team semantics setting} following an approach similar to the \emph{trajectory construct} of $\AHLTL$ has been investigated in~\cite{GutsfeldMOV22}. It is worth noting that unlike other hyper logics (including logics with team semantics) which only capture regular trace properties when interpreted on singleton sets of traces, the logics $\CHLTL$, $\AHLTL$, and $\HU$ can express non-regular trace properties~\cite{BozzelliPS22}.

\vspace{0.1cm}

\noindent \textbf{Our contribution.} Specifications in $\HLTL$ and in the known asynchronous extensions of $\HLTL$, whose most expressive representative is 
 $\HU$~\cite{GutsfeldOO21}, consist  of a prefix of trace quantifiers followed by a quantifier-free formula which expresses temporal requirements on a fixed number of traces. Thus, these hyper logics lack mechanisms to relate directly an unbounded number of traces, which are required for example to express bounded termination or diagnosability properties~\cite{SampathSLST95,BozzanoCGT15,BittnerBCGTV22}. This ability is partially supported by temporal logics with team semantics~\cite{KrebsMV018,Luck20,VirtemaHFK021} and extensions of temporal logics
with the knowledge modalities of epistemic logic~\cite{fagin1995reasoning}, which
relate computations whose histories are observationally equivalent for a
  given agent. In this paper, we introduce and investigate a hyper logical framework in the linear-time setting which unifies two known asynchronous extensions of $\HLTL$ and  the well-known extension $\KLTL$~\cite{HalpernV86} of $\LTL$ with  knowledge modalities under both the
  synchronous and asynchronous perfect recall semantics (where an agent remembers the
  whole sequence of its observations). The novel logic, that we call \emph{generalized $\HLTL$ with stuttering and contexts} ($\GHLTL$ for short),
merges $\SHLTL$ and $\CHLTL$ and adds two
new natural modeling facilities: past temporal modalities for asynchronous hyperproperties and general trace quantification where trace quantifiers can occur in
the scope of temporal modalities. Past temporal modalities used in combination with context modalities provide a powerful mechanism to compare histories of computations at distinct time points. Moreover, unrestricted trace quantification allows to relate an unbounded number of traces.

As a main contribution, we identify a meaningful fragment of $\GHLTL$ with a decidable model-checking problem, that we call  \emph{simple} $\GHLTL$.
This fragment is obtained from $\GHLTL$ by carefully imposing restrictions on the use of the stuttering and context modalities. Simple $\GHLTL$  allows quantification over arbitrary \emph{pointed} traces (i.e., traces plus time points) in the style of $\FOE$~\cite{Finkbeiner017}, it is more expressive than the simple fragment of $\SHLTL$~\cite{BozzelliPS21}, and it is expressively incomparable with full $\SHLTL$ and $\MSOE$. Moreover, this fragment subsumes both $\KLTL$ under the synchronous semantics and the one-agent fragment of $\KLTL$ under the asynchronous semantics.
In fact, simple $\GHLTL$ can be seen as a very large fragment of
$\GHLTL$ with a decidable model checking problem which (1) strictly
subsumes $\HLTL$ and the simple fragment of $\SHLTL$, (2) is closed
under Boolean connectives, and (3) allows an unrestricted nesting of
temporal modalities.
We justify the relevance of simple $\GHLTL$ by showing
that it can express diagnosability properties, interesting classes of
information-flow security policies, both in the synchronous and
asynchronous settings, and bounded termination (more in general,
global promptness in the style of Prompt
$\LTL$~\cite{KupfermanPV09}).
To the best of our knowledge, simple $\GHLTL$ represents the unique
hyper logic with a decidable model-checking problem which can express
powerful non-regular trace properties when interpreted over singleton
sets of
traces.
Due to lack of space, some proofs are omitted and are given in the Appendix.

  \section{Background}\label{sec:Background}

We denote by $\nat$ the set of natural numbers.
Given $i,j\in\nat$, we write $[i,j]$ for the set of natural numbers
$h$ such that $i\leq h\leq j$, we use $[i,j)$ for the set
$[i,j]\setminus\{j\}$, we use $(i,j]$ for the set
$[i,j]\setminus\{i\}$, and $[i,\infty)$ for the set of natural numbers
$h$ such that $h\geq i$.
Given a word $w$ over some alphabet $\Sigma$, $|w|$ is the length of
$w$ ($|w|=\infty$ if $w$ is infinite).
For each $0\leq i<|w|$, $w(i)$ is the $(i+1)^{th}$ symbol of $w$,
$w^{i}$ is the suffix of $w$ from position $i$, that is, the word
$w(i)w(i+1)\ldots$, and $w[0,i]$ is the prefix of $w$ that ends at
position $i$.

We fix a finite set $\AP$ of atomic propositions.
A \emph{trace} is an infinite word over $2^{\AP}$, while a \emph{finite trace} is a
nonempty finite word over $2^{\AP}$.
A \emph{pointed trace} is a pair $(\trace,i)$ consisting of a trace
$\trace$ and a position (timestamp) $i\in\nat$ along
$\trace$.

\vspace{0.5em}
\noindent \textbf{Kripke structures.}
We define the dynamic behaviour of reactive systems by 
\emph{Kripke structures} $\KS=\mktuple{\States,\States_0,\Trans,\Lab}$
over a finite set $\AP$ of atomic propositions, where
$\States$ is a set of states, $\States_0\subseteq \States$ is the set of initial
states, $\Trans\subseteq \States\times \States$ is a transition
relation   which is total in the first argument (i.e., for each
$\state\in \States$ there is  $\state'\in \States$ with
$(\state,\state')\in \Trans$), and $\Lab:\States \rightarrow 2^{\AP}$ is a labeling map
assigning to each state $\state$ the set of
propositions holding at $\state$.
The Kripke structure $\KS$ is finite if $\States$ is finite.
A \emph{path} $\Path$ of $\KS$ is an infinite word
$\Path= \state_0,\state_1,\ldots$ over $\States$ such that $\state_0\in \States_0$  and for all
$i\geq 0$, $(\state_{i},\state_{i+1})\in \Trans$.
The path $\Path= \state_0,\state_1,\ldots$ induces the trace
$\Lab(\state_0)\Lab(\state_1)\ldots$.
A \emph{trace of $\KS$}  is a trace induced by some path of $\KS$.
We denote by $\Lang(\KS)$ the set of traces of $\KS$.
A \emph{finite path} of $\KS$ is a non-empty infix of some path of $\KS$.
We also consider \emph{fair finite Kripke structures}
$(\KS,\FStates)$, that is, finite Kripke structures
$\KS$ equipped with a subset $\FStates$ of $\KS$-states.
A path $\Path$ of $\KS$ is \emph{$\FStates$-fair} if $\Path$ visits
infinitely many times some state in $\FStates$.
We denote by $\Lang(\KS,\FStates)$ the set of traces of
$\KS$ associated with the $\FStates$-fair paths of
$\KS$.\vspace{0.5em}

\noindent  \textbf{Standard $\LTL$ with past ($\PLTL$ for
short)~\cite{Pnueli77}.}
%
%
 Formulas $\psi$ of $\PLTL$ over the given finite set $\AP$ of atomic propositions
 are defined by the following grammar:
\[
\psi ::= \top \mid  p  \mid \neg \psi \mid  \psi \vee \psi \mid \Next \psi   \mid  \Yesterday \psi   \mid \psi \Until \psi \mid \psi \Since \psi
\]
\noindent where $p\in \AP$, $\Next$ and $\Until$ are the \emph{next} and
\emph{until} temporal modalities respectively, and $\Yesterday$
(\emph{previous} or \emph{yesterday}) and $\Since$ (\emph{since}) are
their past counterparts.
$\LTL$ is the fragment of $\PLTL$ that does not contain the past
temporal modalities $\Yesterday$ and $\Since$.
We also use the following abbreviations:
$\Future\psi:=\top \Until \psi$ (\emph{eventually}),
$\Once\psi:=\top \Since \psi$ (\emph{past eventually} or
\emph{once}), and their duals
$\Always\psi:=\neg \Future \neg\psi$ (\emph{always}) and
$\Historically\psi:=\neg \Once \neg\psi$ (\emph{past always} or
\emph{historically}).

The semantics of $\PLTL$ is defined over pointed traces $(\trace,i)$.
The satisfaction relation $(\trace,i)\models \psi$, that defines
whether formula $\psi$ holds at position $i$ along $\trace$, is
inductively defined as follows (we omit the semantics for the Boolean
connectives which is standard):
\[
\begin{array}{ll}
     (\trace, i) \models  p  &  \Leftrightarrow  p \in \trace(i)\\
     (\trace,i)\models  \Next\psi &  \Leftrightarrow (\trace,i+1)\models  \psi\\
     (\trace,i)\models  \Yesterday\psi &  \Leftrightarrow i>0 \text{ and } (\trace,i-1)\models  \psi\\
     (\trace, i) \models  \psi_1\Until \psi_2  &
  \Leftrightarrow  \text{for some $j\geq i$}: (\trace, j)
  \models  \psi_2
  \text{ and }  (\trace, k) \models   \psi_1 \text{ for all }i\leq k<j \\
     (\trace, i) \models  \psi_1\Since \psi_2  &
  \Leftrightarrow  \text{for some $j\leq i$}: (\trace, j)
  \models  \psi_2
  \text{ and }  (\trace, k) \models   \psi_1 \text{ for all }j< k\leq i \\
\end{array}
\]
\noindent A trace $\trace$ is a \emph{model} of $\psi$, written
$\trace\models \psi$, whenever $(\trace,0)\models \psi$.
\vspace{0.5em}

\noindent  \textbf{The logic $\HLTL$~\cite{ClarksonFKMRS14}.}
The syntax of $\HLTL$ formulas $\varphi$ over the given finite set $\AP$ of atomic propositions
and a finite set $\Var$ of trace variables is as follows:
\[
    \varphi := \exists \vartrace.\,\varphi     \mid \forall \vartrace.\,\varphi \mid \psi \quad \quad
     \psi :=   \top \mid \Rel{p}{\vartrace}  \mid \neg \psi \mid \psi \vee \psi   \mid \Next \psi \mid
     \psi \Until  \psi
 \]
where $p\in\AP$, $\vartrace\in\Var$, and
$\exists \vartrace$ and $\forall \vartrace$ are the
\emph{hyper} existential and universal trace quantifiers for variable $\vartrace$, respectively,
which allow relating different traces of the given set of traces. Note that a
$\HLTL$ formula consists of a prefix of traces quantifiers followed by a quantifier-free formula, where the latter corresponds to an $\LTL$
formula whose atomic propositions $p$ are replaced with $\vartrace$-relativized versions $\Rel{p}{\vartrace}$.
Intuitively, $\Rel{p}{\vartrace}$ asserts that $p$ holds at the pointed trace assigned to variable $x$.
 A \emph{sentence} is a formula
where each relativized proposition $\Rel{p}{\vartrace}$ occurs in the scope of trace quantifier
$\exists \vartrace$ or $\forall \vartrace$.

In order to define the semantics of $\HLTL$, we need additional definitions. The \emph{successor} $\SUCC(\trace,i)$ of a pointed trace $(\trace,i)$
is the pointed trace $(\trace,i+1)$, which captures the standard local successor of a position along a trace.

Given a set of traces $\ActTraces$, a (\emph{pointed}) \emph{trace assignment} over $\ActTraces$
is a  partial mapping $\TracesMap : \Var \rightarrow \ActTraces \times \nat$ assigning to each trace variable
$\vartrace$---where $\TracesMap$ is defined---a pointed trace $(\trace,i)$ such that $\trace\in \ActTraces$.
We use $\Dom(\TracesMap)$ to refer to the trace variables for which $\TracesMap$ is defined.
The \emph{successor $\SUCC(\TracesMap)$ of $\TracesMap$} is
the trace assignment over $\ActTraces$ having domain $\Dom(\TracesMap)$ such that $\SUCC(\TracesMap)(\vartrace)=\SUCC(\TracesMap(\vartrace))$
for each $\vartrace\in \Dom(\TracesMap)$. For each $i\geq 0$, we use $\SUCC^i$
for the function obtained by $i$ applications of the function $\SUCC$: $\SUCC^0(\TracesMap)\DefinedAs \TracesMap$ and
$\SUCC^{i+1}(\TracesMap)\DefinedAs \SUCC(\SUCC^i(\TracesMap))$.

Given 
$\vartrace\in \Var$ and a pointed trace $(\trace,i)$ with $\trace\in \ActTraces$,
we denote by $\TracesMap[\vartrace \mapsto (\trace,i)]$ the trace assignment that is identical
to $\TracesMap$ besides for $\vartrace$, which is mapped to $(\trace, i)$.

Given a formula $\varphi$, a set of traces $\ActTraces$, and
a trace assignment $\TracesMap$ over $\ActTraces$ such that $\Dom(\TracesMap)$ contains all the trace variables occurring free in $\varphi$,
the satisfaction relation $\TracesMap \models_{\ActTraces} \varphi$
  is inductively defined as follows (we again omit the semantics of the Boolean connectives):
\[
 \hspace{-0.2cm}\begin{array}{ll}
     \TracesMap  \models_{\ActTraces} \Rel{p}{\vartrace} & \Leftrightarrow   \Traces(\vartrace) = (\trace,i) \text{ and }
    p\in \trace(i)\\
    \TracesMap  \models_{\ActTraces} \exists \vartrace.\, \varphi & \Leftrightarrow  \text{ for some } \trace \in\ActTraces,\,
    \TracesMap[\vartrace \mapsto (\trace, 0)]  \models_{\ActTraces}\varphi\\
     \TracesMap \models_{\ActTraces} \forall \vartrace.\, \varphi & \Leftrightarrow  \text{ for all } \trace\in\ActTraces,\,
     \TracesMap[\vartrace \mapsto (\trace, 0)]  \models_{\ActTraces}\varphi\\
     \TracesMap  \models_{\ActTraces} \Next \psi & \Leftrightarrow   \SUCC (\TracesMap) \models \psi \\
      \TracesMap \models_{\ActTraces} \psi_1 \unt  \psi_2 &\Leftrightarrow
      \text{for some } i\geq 0:\,  \SUCC^i (\TracesMap)  \models_{\ActTraces} \psi_2 \text{ and }  \SUCC^j (\TracesMap) \models_{\ActTraces} \psi_1
       \text{ for all } 0 \le j < i
      \end{array}
\]
\noindent Note that trace quantification ranges over \emph{initial} pointed traces $(\trace,0)$ over $\ActTraces$ (the timestamp is $0$).
As an example, the sentence
$\forall \vartrace_1.\,\forall \vartrace_2.\, \bigwedge_{p\in \AP}\Always (\Rel{p}{\vartrace_1} \leftrightarrow \Rel{p}{\vartrace_2})$ captures the sets of traces which are singletons.

 For a sentence $\varphi$ and a set of traces $\ActTraces$, $\ActTraces$ is a \emph{model} of $\varphi$, written $\ActTraces\models \varphi$, if
$\TracesMap_\emptyset\models_{\ActTraces} \varphi$ where $\TracesMap_\emptyset$ is the trace assignment with empty domain.

  \section{Unifying Framework for Asynchronous Extensions of
HyperLTL}\label{sec:NovelLogic}

 In this section, we introduce a novel logical framework for specifying both asynchronous and synchronous linear-time hyperproperties which
 unifies two known more expressive extensions of $\HLTL$~\cite{ClarksonFKMRS14}, namely
 \emph{Stuttering $\HLTL$} ($\SHLTL$ for short)~\cite{BozzelliPS21} and
\emph{Context $\HLTL$} ($\CHLTL$ for short)~\cite{BozzelliPS21}.
%
%
%
The proposed hyper logic, that we call \emph{generalized $\HLTL$ with stuttering and contexts} ($\GHLTL$ for short),
merges $\SHLTL$ and $\CHLTL$ and adds two
new features: past temporal modalities for asynchronous/synchronous hyperproperties and general trace quantification where trace quantifiers can occur in
the scope of temporal modalities. Since model checking of the logics $\SHLTL$ and $\CHLTL$ is already undecidable~\cite{BozzelliPS21}, we also
consider a meaningful fragment of $\GHLTL$ which is strictly more expressive than the known \emph{simple fragment} of $\SHLTL$~\cite{BozzelliPS21}. Our fragment is
able to express relevant classes of hyperproperties and, as we show in Section~\ref{sec:modelCheckingFragment},  its model checking problem  is decidable.



\subsection{PLTL-Relativized Stuttering and Context Modalities}\label{subSec:RelativizedStuttering}

Classically, a trace is stutter-free if there are no consecutive
positions having the same propositional valuation unless the valuation
is repeated ad-infinitum.
We can associate to each trace a unique stutter-free trace by removing
``redundant'' positions.
The logic $\SHLTL$~\cite{BozzelliPS21}  generalizes these notions with respect to the
pointwise evaluation of a finite set of $\LTL$ formulas. Here, we
consider $\LTL$ with past ($\PLTL$).

\begin{definition}[$\PLTL$ stutter factorization~\cite{BozzelliPS21}]
  Let $\Gamma$ be a finite set of $\PLTL$ formulas and $\trace$ a
  trace.
  The \emph{$\Gamma$-stutter factorization of $\trace$} is the unique
  increasing sequence of positions $\{i_k\}_{k\in[0,m_{\infty}]}$ for
  some $m_{\infty}\in \nat\cup\{\infty\}$ such that the following holds
  for all $j <m_{\infty}$:
  \begin{itemize}
  \item $i_0=0$ and $i_j<i_{j+1}$;
  \item for each $\theta\in \Gamma$, the truth value of $\theta$ along
    the segment $[i_j,i_{j+1})$ does not change, that is, for all
    $h,k\in [i_j,i_{j+1})$, $(\trace,h)\models \theta$ iff
    $(\trace,k)\models \theta$, and the same holds for the infinite segment
    $[m_{\infty}, \infty]$ in case $m_{\infty}\neq \infty$;
  \item the truth value of some formula in $\Gamma$ changes along
    adjacent segments, that is, for some $\theta\in \Gamma$ (depending on
    $j$), $(\trace,i_j)\models \theta$ iff $(\trace,i_{j+1})\not\models \theta$.
  \end{itemize}
\end{definition}

Thus, the $\Gamma$-stutter factorization
$\{i_k\}_{k\in[0,m_{\infty}]}$ of $\trace$ partitions the trace in
adjacent non-empty segments such that the valuation of formulas in
$\Gamma$ does not change within a segment, and changes in moving from
a segment to the adjacent ones.
This factorization induces in a natural way a trace obtained by
selecting the first positions of the finite segments and all the
positions of the unique tail infinite segment, if any. These positions
form an infinite increasing sequence $\{\ell_k\}_{k\in \nat}$ called \emph{$(\Gamma,\omega)$-stutter factorization} of $\trace$,
where:
\[
\ell_0,\ell_1, \ldots \DefinedAs
  \begin{cases}
    i_0,i_1,\ldots & \text{if } m_{\infty}=\infty \\
    i_0,  i_1, \ldots,  i_{m_{\infty}}, i_{m_{\infty}}+1,\ldots & \text{otherwise }
  \end{cases}
\]
The \emph{$\Gamma$-stutter trace $\stfr_{\Gamma}(\trace)$ of $\trace$} (see~\cite{BozzelliPS21})  is defined as follows: $\stfr_{\Gamma}(\trace)\DefinedAs \trace(\ell_0)\trace(\ell_1)\ldots $.
 Note that for $\Gamma=\emptyset$, $\stfr_{\Gamma}(\trace)=\trace$. A trace $\trace$ is \emph{$\Gamma$-stutter free} if it coincides with its $\Gamma$-stutter trace, \ie
$\stfr_{\Gamma}(\trace)=\trace$.

As an example, assume that $\AP =\{p,q,r\}$ and let
$\Gamma= \{p\Until q\}$.
Given $h,k\geq 1$, let $\trace_{h,k}$ be the trace
$\trace_{h,k} = p^{h} q^{k} r^{\omega}$.
These traces have the same $\Gamma$-stutter trace 
given by $p r^{\omega}$.
%

The semantics of the $\Gamma$-relativized temporal modalities
in $\SHLTL$ is based on the notion of \emph{$\Gamma$-successor} $\SUCC_\Gamma(\trace,i)$ of a pointed trace $(\trace,i)$~\cite{BozzelliPS21}:
$\SUCC_\Gamma(\trace,i)$ is the pointed trace $(\trace,\ell)$
where $\ell$ is the smallest position $\ell_j$ in the
$(\Gamma,\omega)$-stutter factorization $\{\ell_k\}_{k\in \nat}$ of $\trace$ which is greater than $i$.
Note that for $\Gamma=\emptyset$,  $\SUCC_\emptyset(\trace,i)=\SUCC(\trace,i)=(\trace,i+1)$. Hence,
$\emptyset$-relativized temporal modalities
in $\SHLTL$ correspond to the temporal modalities of $\HLTL$.

In this paper we extend $\SHLTL$ with past temporal modalities, so that  we introduce the past counterpart
of the successor function. The \emph{$\Gamma$-predecessor} $\PRED_\Gamma(\trace,i)$ of a pointed trace
$(\trace,i)$ is undefined if $i=0$ (written $\PRED_\Gamma(\trace,i)=\Undef$); otherwise,
$\PRED_\Gamma(\trace,i)$ is the pointed trace $(\trace,\ell)$
where $\ell$ is the greatest position $\ell_j$ in the
$(\Gamma,\omega)$-stutter factorization $\{\ell_k\}_{k\in \nat}$ of $\trace$ which is smaller than $i$ (since $\ell_0=0$ such an $\ell_j$ exists).
Note that for $\Gamma=\emptyset$, $\PRED_\emptyset(\trace,i)$ captures the standard local predecessor of a position along a trace.\vspace{0.1cm}

\noindent \textbf{Successors and predecessors of trace  assignments.}
We now define a generalization of the successor  $\SUCC(\TracesMap)$ of a trace assignment  $\TracesMap$ in $\HLTL$.
This generalization is based on the notion of $\Gamma$-successor $\SUCC_\Gamma(\trace,i)$ of a pointed trace $(\trace,i)$
and also takes into account the context modalities $\ctx{\Ctx}$ of  $\CHLTL$~\cite{BozzelliPS21}, where a \emph{context} $\Ctx$ is a
non-empty subset of $\Var$. Intuitively, modality  $\ctx{\Ctx}$ allows reasoning over a subset of the traces assigned to the variables in the formula,
by restricting the temporal progress to those traces.

Formally, let  $\TracesMap$ be a trace assignment over some set of traces $\ActTraces$,  $\Gamma$ be a finite set of $\PLTL$ formulas,
and $\Ctx$ be a context. The \emph{$(\Gamma,\Ctx)$-successor of
$\TracesMap$}, denoted by $\SUCC_{(\Gamma,\Ctx)}(\TracesMap)$,  is the trace assignment over $\ActTraces$  having domain
 $\Dom(\TracesMap)$, and defined as follows for each $\vartrace\in \Dom(\TracesMap)$:
\[
\SUCC_{(\Gamma,\Ctx)}(\TracesMap)(\vartrace) :=
  \begin{cases}
    \SUCC_{\Gamma}(\TracesMap(\vartrace)) & \text{if }    \vartrace\in \Ctx \\
    \TracesMap(\vartrace) & \text{otherwise }
  \end{cases}
\]
Note that $\SUCC_{(\emptyset,\Var)}(\TracesMap)= \SUCC(\TracesMap)$. Moreover, we define the \emph{$(\Gamma,\Ctx)$-predecessor $\PRED_{(\Gamma,\Ctx)}(\TracesMap)$ of
$\TracesMap$}  as follows:
 $\PRED_{(\Gamma,\Ctx)}(\TracesMap)$ is \emph{undefined}, written $\PRED_{(\Gamma,\Ctx)}(\TracesMap)= \Undef$, if
 there is $\vartrace\in \Dom(\TracesMap)$ such that $\PRED_\Gamma(\TracesMap(\vartrace))=\Undef$. Otherwise,
 $\PRED_{(\Gamma,\Ctx)}(\TracesMap)$  is the trace assignment over $\ActTraces$  having domain
 $\Dom(\TracesMap)$, and defined as follows for each $\vartrace\in \Dom(\TracesMap)$:
\[
\PRED_{(\Gamma,\Ctx)}(\TracesMap)(\vartrace) :=
  \begin{cases}
    \PRED_{\Gamma}(\TracesMap(\vartrace)) & \text{if }    \vartrace\in \Ctx \\
    \TracesMap(\vartrace) & \text{otherwise }
  \end{cases}
\]
 Finally, for each $i\geq 0$, we define the $i^{th}$ application $\SUCC^{i}_{(\Gamma,\Ctx)}$ of $\SUCC_{(\Gamma,\Ctx)}$ and
 the $i^{th}$ application $\PRED^{\,i}_{(\Gamma,\Ctx)}$ of $\PRED_{(\Gamma,\Ctx)}$ as follows, where $\PRED_{(\Gamma,\Ctx)}(\Undef)\DefinedAs \Undef$:
 \begin{itemize}
   \item $\SUCC^{0}_{(\Gamma,\Ctx)}(\TracesMap)\DefinedAs \TracesMap$ and  $\SUCC^{i+1}_{(\Gamma,\Ctx)}(\TracesMap)\DefinedAs \SUCC_{(\Gamma,\Ctx)}(\SUCC^{i}_{(\Gamma,\Ctx)}(\TracesMap))$.
   \item $\PRED^{\,0}_{(\Gamma,\Ctx)}(\TracesMap)\DefinedAs \TracesMap$ and  $\PRED^{\,i+1}_{(\Gamma,\Ctx)}(\TracesMap)\DefinedAs \PRED_{(\Gamma,\Ctx)}(\PRED^{\,i}_{(\Gamma,\Ctx)}(\TracesMap))$.
 \end{itemize}

\subsection{Generalized HyperLTL with Stuttering and Contexts}\label{subSec:GeneralizedHLTL}

We introduce now the novel logic $\GHLTL$. $\GHLTL$ formulas $\varphi$ over
$\AP$ and a finite set $\Var$ of trace variables are defined by the following syntax:
\[
    \varphi :=\top \mid \Rel{p}{\vartrace}  \mid \neg \varphi \mid \varphi \vee \varphi \mid \exists \vartrace.\,\varphi     \mid
    \ctx{\Ctx} \varphi  \mid \Next_{\Gamma} \varphi \mid \Yesterday_{\Gamma} \varphi \mid
     \varphi \Until_{\Gamma} \varphi \mid
      \varphi \Since_{\Gamma} \varphi
 \]
where $p\in\AP$, $\vartrace\in\Var$,
$\ctx{\Ctx}$  is the context modality
with $\emptyset \neq \Ctx\subseteq \Var$, $\Gamma$ is a finite set of $\PLTL$ formulas,
and $\Next_{\Gamma}$, $\Yesterday_{\Gamma}$,  $\Until_\Gamma$ and $\Since_{\Gamma}$ are the
stutter-relativized versions of the $\PLTL$ temporal modalities.
Intuitively, the context modality $\ctx{\Ctx}$
 restricts the evaluation of the temporal
modalities to the traces associated with the variables in $\Ctx$, while
the temporal modalities $\Next_{\Gamma}$, $\Yesterday_{\Gamma}$,  $\Until_\Gamma$ and $\Since_{\Gamma}$
are evaluated by a
lockstepwise traversal of the $\Gamma$-stutter traces associated to the traces assigned to the  variables in the current context $\Ctx$.
Note that  the hyper universal  quantifier $\forall\vartrace$ can be introduced as an abbreviation:
$\forall\vartrace.\,\varphi\equiv \neg \exists \vartrace.\,\neg\varphi$. For a  variable
$\vartrace$, we write $\ctx{\vartrace}$ instead of $\ctx{\{\vartrace\}}$.
Moreover, we write $\Next$, $\Yesterday$,  $\Until$ and $\Since$ instead of
$\Next_{\emptyset}$, $\Yesterday_{\emptyset}$,  $\Until_\emptyset$ and $\since_{\emptyset}$, respectively.
Furthermore, for a $\PLTL$ formula $\psi$ and a   variable $\vartrace$,
   $\Rel{\psi}{\vartrace}$ is the formula obtained from $\psi$ by replacing each  proposition $p$ with its $\vartrace$-version
$\Rel{p}{\vartrace}$.
 A \emph{sentence} is a formula
where each relativized proposition $\Rel{p}{\vartrace}$ occurs in the scope of trace quantifier
$\exists \vartrace$ or $\forall \vartrace$, and each temporal modality occurs in the scope of a trace quantifier.\vspace{0.2cm}

 \noindent \textbf{The known logics  $\SHLTL$,  $\CHLTL$, and simple $\SHLTL$.} A formula $\varphi$ of  $\GHLTL$ is in \emph{prenex} form
  if it is of the form $\Qf_1\vartrace_1.\,\ldots \Qf_n\vartrace_n.\,\psi$ where $\psi$ is quantifier-free and  $\Qf_i\in \{\exists,\forall\}$
  for all $i\in [1,n]$. The logics $\SHLTL$ and $\CHLTL$ introduced in~\cite{BozzelliPS21}  correspond to  syntactical fragments of
  $\GHLTL$ where the formulas are in prenex form and past temporal modalities are not used. Moreover, in $\SHLTL$, the context modalities are not allowed, while in $\CHLTL$, the subscript $\Gamma$ of every temporal modality must be the empty set. Note that in $\HLTL$~\cite{ClarksonFKMRS14},  both the context modalities and the temporal modalities where the subscript $\Gamma$ is not empty are disallowed. Finally, we recall the \emph{simple} fragment of $\SHLTL$~\cite{BozzelliPS21}, which
  is more expressive than $\HLTL$ and is parameterized by a finite set $\Gamma$ of $\LTL$ formulas. The  \emph{quantifier-free} formulas  of simple $\SHLTL$ for the parameter $\Gamma$ are defined
as  Boolean combinations of formulas of  the form $\Rel{\psi}{\vartrace}$, where $\psi$ is an $\LTL$ formula, and formulas $\psi_\Gamma$  defined by the following grammar:\vspace{0.1cm}

$
    \psi_\Gamma :=  \top \mid   \Rel{p}{\vartrace} \mid \neg \psi_\Gamma \mid \psi_\Gamma \vee \psi_\Gamma  \mid   \Next_{\Gamma} \psi_\Gamma    \mid \psi_\Gamma \Until_{\Gamma} \psi_\Gamma
 $\vspace{0.2cm}

\noindent  \textbf{Semantics of $\GHLTL$.} Given a formula $\varphi$, a set of traces $\ActTraces$,
a trace assignment $\TracesMap$ over $\ActTraces$ such that $\Dom(\TracesMap)$ contains all the trace variables occurring free in $\varphi$, and a context $\Ctx\subseteq \Var$,
the satisfaction relation $(\TracesMap, \Ctx) \models_{\ActTraces} \varphi$, meaning that the assignment $\TracesMap$ over $\ActTraces$ satisfies
$\varphi$ under the context $\Ctx$,
  is inductively defined as follows (we again omit the semantics of the Boolean connectives):
\[
 \hspace{-0.2cm}\begin{array}{ll}
     (\TracesMap, \Ctx) \models_{\ActTraces} \Rel{p}{\vartrace} & \Leftrightarrow   \Traces(\vartrace) = (\trace,i) \text{ and }
    p\in \trace(i)\\
    (\TracesMap, \Ctx) \models_{\ActTraces} \exists \vartrace.\, \varphi & \Leftrightarrow  \text{ for some } \trace \in\ActTraces,\,
    (\TracesMap[\vartrace \mapsto (\trace, 0)], \Ctx) \models_{\ActTraces}\varphi\\
      (\TracesMap, \Ctx) \models_{\ActTraces} \ctx{\Ctx'} \varphi & \Leftrightarrow   (\TracesMap, \Ctx') \models_{\ActTraces} \varphi\\
    (\TracesMap, \Ctx) \models_{\ActTraces} \Next_\Gamma \varphi &\Leftrightarrow  (\SUCC_{(\Gamma,\Ctx)}(\TracesMap),\Ctx) \models \varphi \\
    (\TracesMap, \Ctx) \models_{\ActTraces} \Yesterday_\Gamma \varphi &\Leftrightarrow
    \PRED_{(\Gamma,\Ctx)}(\TracesMap) \neq \Undef \text{ and } (\PRED_{(\Gamma,\Ctx)}(\TracesMap), \Ctx) \models_{\ActTraces} \varphi\\
      (\TracesMap, \Ctx) \models_{\ActTraces} \varphi_1 \unt_\Gamma \varphi_2 &\Leftrightarrow
      \text{for some $i\geq 0$: }  (\SUCC^i_{(\Gamma,\Ctx)}(\TracesMap),\Ctx) \models_{\ActTraces} \varphi_2 \text{ and }\\
    & \phantom{\Leftrightarrow} (\SUCC^j_{(\Gamma,\Ctx)}(\TracesMap), \Ctx) \models_{\ActTraces} \varphi_1 \text{ for all } 0 \le j < i,\,  \\
    (\TracesMap, \Ctx) \models_{\ActTraces} \varphi_1 \since_\Gamma \varphi_2 &\Leftrightarrow
      \text{for some  $i\geq 0$ such that  $\PRED^{\,i}_{(\Gamma,\Ctx)}(\TracesMap) \neq \Undef$: }
      (\PRED^{\,i}_{(\Gamma,\Ctx)}(\Traces),\Ctx) \models_{\ActTraces} \varphi_2   \\
      & \phantom{\Leftrightarrow} \text{and } (\PRED^{\,j}_{(\Gamma,\Ctx)}(\TracesMap), \Ctx) \models_{\ActTraces} \varphi_1 \text{ for all } 0 \le j < i
      \end{array}
\]
\noindent For a sentence $\varphi$ and a set of traces $\ActTraces$, $\ActTraces$ is a \emph{model} of $\varphi$, written $\ActTraces\models \varphi$, if
$(\TracesMap_\emptyset, \Var) \models_{\ActTraces} \varphi$ where $\TracesMap_\emptyset$ is the trace assignment with empty domain.\vspace{0.2cm}

\noindent \textbf{Fair model checking and standard model checking.}
For a fragment $\Frag$ of $\GHLTL$, 
the \emph{fair model checking problem} for $\Frag$ consists on deciding,
given a fair finite Kripke structure $(\KS,\FStates)$ and a sentence
$\varphi$ of $\Frag$, whether $\Lang(\KS,\FStates)\models \varphi$.
The previous problem is simply called \emph{model checking problem}
whenever $\FStates$ coincides with the set of $\KS$-states.
We consider fair model checking just for technical convenience.
For the decidable fragment of $\GHLTL$ introduced in
Section~\ref{subSec:MeaningfulNovelFragments}, we will obtain the same
complexity bounds for both fair model checking and standard model
checking (see Section~\ref{sec:modelCheckingFragment}).

\subsection{The Simple Fragment of $\GHLTL$}\label{subSec:MeaningfulNovelFragments}

We introduce now a fragment of $\GHLTL$, that we call \emph{simple $\GHLTL$}, which syntactically subsumes the simple fragment of $\SHLTL$~\cite{BozzelliPS21}.

In order to define the syntax of simple $\GHLTL$, we first consider some shorthands, obtained by a restricted use of the context modalities. The \emph{pointed existential quantifier}  $\exists^{\Pt}\vartrace$ and the \emph{pointed universal
   quantifier} $\forall^{\Pt}\vartrace$ are defined
 as follows: $\exists^{\Pt}\vartrace.\,\varphi \DefinedAs \exists\vartrace.\, \ctx{\vartrace}  \future \ctx{\Var} \varphi$ and $\forall^{\Pt}\vartrace.\,\varphi::= \neg \exists^{\Pt}\vartrace.\,\neg\varphi$.
   Thus the pointed quantifiers quantify on arbitrary pointed traces over the given set of traces and set the global context for the given operand.
   Formally, $(\TracesMap, \Ctx) \models_{\ActTraces} \exists^{\Pt} \vartrace.\, \varphi$ if   for some pointed trace $(\trace,i)$ with  $\trace \in\ActTraces$,
    $(\TracesMap[\vartrace \mapsto (\trace, i)], \Var) \models_{\ActTraces}\varphi$.

   \noindent For example, the sentence $\exists \vartrace_1.\,\exists^{\Pt} \vartrace_2.\, \bigl( \bigwedge_{p\in \AP}\Always(p[x_1] \leftrightarrow p[x_2])\wedge \ctx{\vartrace_2}\yesterday\top\bigr)$ asserts that  there are two traces $\trace_1$ and $\trace_2$ in the given model s.t.~some \emph{proper} suffix of
    $\trace_2$ coincides with $\trace_1$.

Simple $\GHLTL$  is parameterized by a finite set $\Gamma$ of $\PLTL$ formulas. The set of formulas $\varphi_\Gamma$ in the $\Gamma$-fragment is defined
as follows:
\[
    \varphi_\Gamma :=   \top \mid  \ctx{\vartrace}\Rel{\psi}{\vartrace} \mid \neg \varphi_\Gamma \mid \varphi_\Gamma \vee \varphi_\Gamma   \mid
    \exists^{\Pt} \vartrace.\,\varphi_\Gamma    \mid
     \Next_{\Gamma} \varphi_\Gamma \mid \Yesterday_{\Gamma} \varphi_\Gamma \mid
     \varphi_\Gamma \Until_{\Gamma} \varphi_\Gamma \mid
      \varphi_\Gamma \Since_{\Gamma} \varphi_\Gamma
 \]
\noindent where $\psi$ is a $\PLTL$ formula.  Note that  $\exists \vartrace.\,\varphi$ can be expressed as
$\exists^{\Pt} \vartrace.\,(\varphi\wedge \ctx{\vartrace}\neg\Yesterday\top)$.
$\SGHLTL{\Gamma}$ is the class of formulas obtained with the syntax above for a given $\Gamma$.
Simple $\GHLTL$ is the union $\SGHLTL{\Gamma}$ for all $\Gamma$.
We say that a formula $\varphi$ of simple $\GHLTL$ is \emph{singleton-free} if for each subformula  $\ctx{x}\Rel{\psi}{\vartrace}$ of $\varphi$, $\psi$ is an atomic proposition.
Evidently, for an atomic proposition $p$,
$\ctx{x}\Rel{p}{\vartrace}$ is equivalent to $\Rel{p}{\vartrace}$.

In Section~\ref{sec:modelCheckingFragment}, we will show that (fair)
model checking of simple $\GHLTL$ is decidable.
Simple $\GHLTL$ can be seen as a very large fragment of $\GHLTL$ with
a decidable model checking problem which subsumes the simple fragment
of $\SHLTL$, is closed under Boolean connectives, and allows an
unrestricted nesting of temporal modalities.
We conjecture (without proof) that this is the largest such sub-class
of $\GHLTL$ because:
\begin{enumerate}
\item Model checking of $\SHLTL$ is already
  undecidable~\cite{BozzelliPS21} for sentences whose relativized
  temporal modalities exploit two distinct sets of $\LTL$ formulas
  and, in particular, for two-variable quantifier alternation-free
  sentences of the form
  $\exists \vartrace_1.\,\exists \vartrace_2.\,(\varphi\wedge
  \Always_\Gamma \psi)$, where $\psi$ is a propositional formula,
  $\Gamma$ is a nonempty set of propositions, and $\varphi$ is a
  quantifier-free formula which use only the temporal modalities
  $\Future_\emptyset$ and $\Always_\emptyset$.
\item Model checking of $\CHLTL$ is undecidable~\cite{BozzelliPS22}
  even for the fragment consisting of two-variable quantifier
  alternation-free sentences of the form
  $\exists \vartrace_1.\exists \vartrace_2.\, \psi_0\wedge
  \ctx{\vartrace_2}\Future \ctx{\{\vartrace_1,\vartrace_2\}}\psi$,
  where $\psi_0$ and $\psi$ are quantifier-free $\HLTL$ formulas (note
  that $\psi_0$ is evaluated in the global context
  $\ctx{\{\vartrace_1,\vartrace_2\}}$).
\end{enumerate}

The second undecidability result suggests to consider the extension of
simple $\GHLTL$ where singleton-context subformulas of the form
$\ctx{\vartrace}\Rel{\psi}{\vartrace}$ are replaced with
\emph{quantifier-free} $\GHLTL$ formulas with multiple variables of
the form $\ctx{\vartrace}\xi$, where $\xi$ only uses singleton
contexts $\ctx{\vartraceAux}$ and temporal modalities with subscript
$\emptyset$. However, we can show that the resulting logic is not more
expressive than simple $\GHLTL$: a sentence in the considered
extension can be translated into an equivalent simple $\GHLTL$
sentence though with a non-elementary blowup (for details, see~Appendix~\ref{app:extensionSimpleGHLTL}).

\subsection{Examples of Specifications in Simple
  $\GHLTL$}\label{sec:SHLTLspecifications}

We consider some relevant properties which can be expressed in simple
$\GHLTL$. 
Simple $\GHLTL$ subsumes the simple fragment of $\SHLTL$, and this
fragment (as shown in~\cite{BozzelliPS21}) can express relevant
information-flow security properties for asynchronous frameworks such
as distributed systems or cryptographic protocols.
An example is the asynchronous variant of the \emph{noninterference}
property, as defined by Goguen and Meseguer~\cite{goguen1982security},
which asserts that the observations of low users (users accessing
only to public information) do not change when all inputs of high
users (users accessing secret information) are removed.

\paragraph*{Observational Determinism.} 
An important information-flow property is observational determinism,
which states that traces which have the same initial low inputs are
indistinguishable to a low user.
In an asynchronous setting, a user cannot infer that a transition
occurred if consecutive observations remain unchanged.
Thus, for instance, observational determinism with equivalence of
traces up to stuttering (as formulated in~\cite{ZdancewicM03}) can be
captured by the following simple $\SHLTL$ sentence (where $LI$ is the
set of propositions describing inputs of low users and $LO$ is set of
propositions describing outputs of low users):
\[
  \forall \vartrace.\,\forall \vartraceAux.\,   \bigwedge_{p\in LI}  (\Rel{p}{\vartrace} \leftrightarrow \Rel{p}{\vartraceAux}) \rightarrow \Always_{LO} \bigwedge_{p\in LO} (\Rel{p}{\vartrace} \leftrightarrow \Rel{p}{\vartraceAux})
\]

\paragraph*{After-initialization Properties.} 
Simple $\GHLTL$ also allows to specify complex combinations of
asynchronous and synchronous constraints.
As an example, we consider the property~\cite{GutsfeldOO21} that for
an $\HLTL$ sentence
$\Qf_1 \vartrace_1.\,\ldots \Qf_n\vartrace_n.\,
\psi(\vartrace_1,\ldots,\vartrace_n)$, the quantifier-free formula
$\psi(\vartrace_1,\ldots,\vartrace_n)$ holds along the traces bound by
variables $\vartrace_1,\ldots,\vartrace_n$ after an initialization
phase.
Note that this phase can take a different (and unbounded) number of
steps on each trace.
Let $\IN$ be a proposition characterizing the initialization phase.
The formula
$\PI(\vartrace)\DefinedAs \ctx{\vartrace}
(\neg\Rel{\IN}{\vartrace}\wedge(\neg\Yesterday\top \vee \Yesterday
\Rel{\IN}{\vartrace}))$
is a simple $\GHLTL$ formula that asserts that for the pointed trace
$(\trace,i)$ assigned to variable $\vartrace$, the position $i$ is the
first position of $\trace$ following the initialization phase.
In other words, $i$ is the first position at which $\neg\Rel{\IN}$
holds.
Then, the previous requirement can be expressed in simple $\GHLTL$ as
follows:
\[
  \Qf_1\vartrace_1.\ldots \Qf_n\vartrace_n.\big(\PI(x_1) \circ_1 \ldots \PI(x_n)\circ_n \psi(x_1,\ldots x_n)
\]
where $\circ_i$ is $\wedge$ if $\Qf_i=\exists$ and $\circ_i$ is
$\Impl$ if $\Qf_i=\forall$.

\paragraph*{Global Promptness.} 
As another meaningful example, we consider global promptness (in the
style of Prompt $\LTL$~\cite{KupfermanPV09}), where one need to check
that there is a uniform bound on the response time in all the traces
of the system, that is, ``\emph{there is $k$ such that for every
  trace, each request $q$ is followed by a response $p$ within $k$
  steps}''.
Global promptness is expressible in simple $\GHLTL$ as follows:
\[
\exists^{\Pt}\vartrace.\,\bigl(\Rel{q}{\vartrace}\,\wedge\,  \forall^{\Pt}\vartraceAux.\,  (\Rel{q}{\vartraceAux} \rightarrow  (\neg\Rel{p}{\vartrace}\,\Until\, \Rel{p}{\vartraceAux})) \bigr)
\]
The previous sentence asserts that there is request ($x$ in the
formula) that has the longest response.
Note that $y$ is quantified universally (so it can be instantiated to
the same trace as $x$), and that the use of until in
$(\neg\Rel{p}{\vartrace}\,\Until\, \Rel{p}{\vartraceAux})) \bigr)$
implies that the response $p[y]$ eventually happens.
Hence, all requests, including receive a response.
Now, the pointed trace $(\trace_{\vartrace},i)$ assigned to $x$ is
such that $\trace_{\vartrace}(i)$ is a request ($q[x]$) and for every
pointed trace $(\trace_{\vartraceAux},j)$ in the model such that
$\trace_{\vartraceAux}(j)$ is a request ($q[y]$), it holds that (i)
the request $\trace_{\vartraceAux}(j)$ is followed by a response
$\trace_{\vartraceAux}(j+k)$ for some $k\geq 0$, and (ii) no response
occurs in $\trace_{\vartrace}$ in the interval $[i,i+k)$.
Therefore, the response time $h$ for $x$ is the smallest $h$ such that
$\trace_{\vartrace}(i+h)$ is a response is a global bound on the
response time.

\paragraph*{Diagnosability.}
We now show that simple $\GHLTL$ is also able to express
\emph{diagnosability} of
systems~\cite{SampathSLST95,BozzanoCGT15,BittnerBCGTV22} in a general
asynchronous setting.
In the diagnosis process, faults of a critical system (referred as the
plant) are handled by a dedicated module (the \emph{diagnoser}) which
runs in parallel with the plant.
The diagnoser analyzes the observable information from the
plant---made available by predefined sensors---and triggers suitable
alarms in correspondence to (typically unobservable) conditions of
interest, called faults.
An alarm condition specifies the relation (delay) between a given
diagnosis condition and the raising of an alarm.
A plant $\Plant$ is \emph{diagnosable} with respect to a given alarm
condition $\alpha$, if there is a diagnoser $\D$ which satisfies
$\alpha$ when $\D$ runs in parallel with $\Plant$.

The given set of propositions $\AP$ is partitioned into a set of
observable propositions $\Obs$ and a set of unobservable propositions
$\Int$.
Two finite traces $\trace$ and $\trace'$ are \emph{observationally
equivalent} iff the projections of
$\stfr_{\Obs}(\trace\cdot P^{\omega})$ and
$\stfr_{\Obs}(\trace'\cdot (P')^{\omega})$ over $\Obs$ coincide, where
$P$ is the last symbol of $\trace$ and similarly $P'$ is the last
symbol of $\trace'$.
Given a pointed trace $(\trace,i)$, $i$ is an \emph{observation point}
of $\trace$ if either $i=0$, or $i>0$ and
$\trace(i-1)\cap \Obs \neq \trace(i)\cap \Obs$.
Then a plant $\Plant$ can be modeled as a finite Kripke structure
$\mktuple{\States,\States_0,\Trans,\Lab}$, where $\Trans$ is
partitioned into internal transitions $(\state,\state')$ where
$\Lab(\state)\cap \Obs = \Lab(\state')\cap \Obs$ and observable
transitions where $\Lab(\state)\cap \Obs \neq \Lab(\state')\cap \Obs$.
A diagnoser $\D$ is modelled as a finite deterministic Kripke
structure over $\AP' \supseteq \Obs$ (with $\AP'\cap \Int=\emptyset$).
In the behavioural composition of the plant $\Plant$ with $\D$, the
diagnoser only reacts to the observable transitions of the plant, that
is, every transition of the diagnoser is associated with an observable
transition of the plant.
Simple $\GHLTL$ can express diagnosability with \emph{finite delay},
\emph{bounded delay}, or \emph{exact
  delay} as defined in~\cite{BozzanoCGT15,BittnerBCGTV22}.
Here, we focus for simplicity on finite delay diagnosability.
Consider a diagnosis condition specified by a $\PLTL$ formula $\beta$.
A plant $\Plant$ is \emph{finite delay diagnosable} with respect to
$\beta$ whenever for every pointed trace $(\trace,i)$ of $\Plant$ such
that $(\trace,i)\models\beta$, there exists an observation point
$k\geq i$ of $\trace$ such that for all pointed traces $(\trace',k')$
of $\Plant$ so that $k'$ is an observation point of $\trace'$ and
$\trace[0,k]$ and $\trace'[0,k']$ are observationally equivalent, it
holds that $(\trace',i')\models \beta$ for some $i'\leq k'$.
Finite delay diagnosability w.r.t.~$\beta$ can be expressed in simple
$\GHLTL$ as follows:

\begin{align*}
&\forall^\Pt\vartrace.\,\Bigl(\ctx{x}\Rel{\beta}{\vartrace} \rightarrow \Future_{\Obs}\bigl(\ObsPt(\vartrace)\wedge \forall^\Pt\vartraceAux.\,\{(\ObsPt(\vartraceAux)\wedge
  \theta_{\Obs}(\vartrace,\vartraceAux)) \rightarrow \ctx{\vartraceAux}\Once\Rel{\beta}{\vartraceAux}\}\bigr)\Bigr)\\
  \intertext{where}
 \theta_{\Obs}(\vartrace,\vartraceAux)&\DefinedAs \displaystyle{\bigwedge_{p\in \Obs}}\Historically_{\Obs}(\Rel{p}{\vartrace} \leftrightarrow \Rel{p}{\vartraceAux})\,\wedge\,
  \Once_{\Obs}(\ctx{\vartrace}\neg \Yesterday\top \wedge \ctx{\vartraceAux}\neg \Yesterday\top)\\
  \ObsPt(\vartrace)&\DefinedAs \ctx{\vartrace} (\neg\Yesterday\top\wedge
\bigvee_{p\in\Obs}(\Rel{p}{\vartrace}\leftrightarrow \neg \Yesterday
\Rel{p}{\vartrace})
\end{align*}

Essentially $\ObsPt(\vartrace)$ determines the observation points and
$\theta_{\Obs}$ captures that both traces have the same history of
observations.
The main formula establishes that if $x$ detects a failure $\beta$
then there is future observation point in $x$ and for all other traces
that are observationally equivalent to $x$ have also detected $\beta$.

\subsection{Expressiveness Issues}
\label{subSec:Expressiveness}
In this section, we present some results and conjectures about the
expressiveness comparison among $\GHLTL$ (which subsumes $\SHLTL$ and
$\CHLTL$), its simple fragment and the logic $\SHLTL$.
We also consider the logics for linear-time hyperproperties based on
the equal-level predicate whose most powerful representative is
$\MSOE$.
Recall that the first-order fragment $\FOE$ of $\MSOE$ is already
strictly more expressive than $\HLTL$~\cite{Finkbeiner017} and, unlike
$\MSOE$, its model-checking problem is decidable~\cite{CoenenFHH19}.
Moreover, we show that $\GHLTL$ and its simple fragment represent a unifying
framework in the linear-time setting for specifying both
hyperproperties and the knowledge modalities of epistemic temporal
logics under both the synchronous and asynchronous perfect recall
semantics.

Our expressiveness results about linear-time hyper logics can be
summarized as follows.

\begin{theorem}\label{theorem:expressivenessResults}
  The following hold:
\begin{itemize}
\item $\GHLTL$ is more expressive than $\SHLTL$, simple $\GHLTL$, and
  $\FOE$.
\item Simple $\GHLTL$ is more expressive than simple $\SHLTL$.
\item Simple $\GHLTL$ and $\SHLTL$ are expressively incomparable.
\item Simple $\GHLTL$ and $\MSOE$ are expressively incomparable.
\end{itemize}
\end{theorem}
\begin{proof}
  We first show that there are hyperproperties expressible in simple
  $\GHLTL$ but not in $\SHLTL$ and in $\MSOE$.
  Given a sentence $\varphi$, the \emph{trace property denoted by
    $\varphi$} is the set of traces $\trace$ such that the singleton
  set of traces $\{\trace\}$ satisfies $\varphi$.
  It is known that $\SHLTL$ and $\MSOE$ capture only \emph{regular}
  trace properties~\cite{BozzelliPS22}.
  In contrast simple $\GHLTL$ can express powerful non-regular trace
  properties.
  For example, consider the so called \emph{suffix property} over
  $\AP=\{p\}$: a trace $\trace$ satisfies the suffix property if there
  exists a proper suffix $\trace^k$ of $\trace$ for some $k>0$ such
  that $\trace^k=\trace$.
  This non-regular trace property can be
  expressed in  $\SGHLTL{\emptyset}$ as follows:
\[
    \forall \vartrace_1.\,\forall \vartrace_2.\, \bigwedge_{p\in \AP}\Always(\Rel{p}{\vartrace_1} \leftrightarrow \Rel{p}{\vartrace_2})\,\wedge\,
    \forall \vartrace_1.\,\exists^{\Pt} \vartrace_2.\, \bigl( \bigwedge_{p\in \AP}\Always(p[x_1] \leftrightarrow p[x_2])\wedge \ctx{\vartrace_2}\yesterday\top\bigr)
\]
The first conjunct asserts that each model is a singleton, and the
second conjunct requires that for the unique trace $\trace$ in a
model, there is $k>0$ such that $\trace(i) = \trace(i +k)$ for all
$i\geq 0$.

Next, we observe that $\FOE$ can be easily translated into $\GHLTL$,
since the pointer quantifiers of $\GHLTL$ correspond to the
quantifiers of $\FOE$.
Moreover, the predicate $\vartrace\leq \vartrace'$ of $\FOE$,
expressing that for the pointed traces $(\trace,i)$ and $(\trace',i')$
bound to $\vartrace$ and $\vartrace'$, $\trace=\trace'$ and
$i\leq i'$, can be easily captured in $\GHLTL$.
This is also the case for the equal-level predicate
$\Eq(\vartrace,\vartrace')$, which can be expressed as
$\ctx{\{\vartrace,\vartrace'\}}\Once(\ctx{\vartrace}\neg\Yesterday\top
\wedge \ctx{\vartrace'}\neg\Yesterday\top)$.

  In Section~\ref{sec:modelCheckingFragment} we show that model checking
  of simple $\GHLTL$ is decidable.
  Thus, since model checking of both $\SHLTL$ and $\MSOE$ are
  undecidable~\cite{BozzelliPS21,CoenenFHH19} and $\GHLTL$ subsumes
  $\SHLTL$, by the previous argumentation, the theorem follows.
\end{proof}

It remains an open question whether $\FOE$ is subsumed by simple
$\GHLTL$.
We conjecture that neither $\CHLTL$ nor the fix-point calculus
$\HU$~\cite{GutsfeldOO21} (which captures both $\CHLTL$ and
$\SHLTL$~\cite{BozzelliPS22}) subsume simple $\GHLTL$.
The motivation for our conjecture is that $\HU$ sentences consist of a
prefix of quantifiers followed by a quantifier-free formula where
quantifiers range over \emph{initial} pointed traces
$(\trace,0)$.
Thus, unlike simple $\GHLTL$, $\HU$ cannot express requirements which
relate at some point in time an unbounded number of traces.
Diagnosability  (see Subsection~\ref{sec:SHLTLspecifications}) falls
in this class of requirements.
It is known that the following property, which can be easily expressed in simple $\GHLTL$, is not definable in
$\HLTL$~\cite{BozzelliMP15}: for some $i>0$, every trace in the given set
of traces does not satisfy proposition $p$ at position $i$.
We conjecture that similarly to $\HLTL$, such a property (and
diagnosability as well) cannot be expressed in $\HU$.
\vspace{0.2cm}

\noindent
\textbf{Epistemic Temporal Logic $\KLTL$ and its relation with
  $\GHLTL$.}
The logic $\KLTL$~\cite{HalpernV86} is a well-known extension of
$\LTL$ obtained by adding the unary knowledge modalities $\know_a$,
where $a$ ranges over a finite set $\Ag$ of agents, interpreted under
the synchronous or asynchronous (perfect recall) semantics.
The semantics is given with respect to an observation map
$\Obs:\Ag \mapsto 2^{\AP}$ that assigns to each agent $a$ the set of
propositions which agent $a$ can observe.
Given two finite traces $\trace$ and $\trace'$ and $a\in\Ag$, $\trace$
and $\trace'$ are \emph{synchronously equivalent for agent $a$},
written $\trace\sim_a^{sy}\trace'$, if the projections of $\trace$ and
$\trace'$ over $\Obs(a)$ coincide.
The finite traces $\trace$ and $\trace'$ are \emph{asynchronously
  equivalent for agent $a$}, written $\trace\sim_a^{as}\trace'$, if
the projections of $\stfr_{\Obs(a)}(\trace\cdot P^{\omega})$ and
$\stfr_{\Obs(a)}(\trace'\cdot (P')^{\omega})$ over $\Obs(a)$ coincide,
where $P$  is the last symbol of $\trace$ and $P'$ is the last symbol of
$\trace'$.
For a set of traces $\Lang$ and a pointed trace $(\trace,i)$ over
$\Lang$, the semantics of the knowledge modalities is as follows,
where $\sim_a$ is $\sim_a^{sy}$ under the synchronous semantics, and
$\sim_a^{as}$ otherwise:
$ (\trace,i) \models_{\Lang,\Obs} \know_a\, \varphi \ \Leftrightarrow
\ \text{for all pointed traces } (\trace',i') \text{ on } \Lang \text{
  such that } \trace[0,i] \sim_a \trace'[0,i'],\, (\trace',i')
\models_{\Lang,\Obs} \varphi $.

We say that \emph{$\Lang$ satisfies $\varphi$ w.r.t.~the observation
  map $\Obs$}, written $\Lang\models_\Obs\varphi$, if for all traces
$\trace\in \Lang$, $(\trace,0) \models_{\Lang,\Obs} \varphi$.
The logic $\KLTL$ can be easily embedded into $\GHLTL$.
In particular, the following holds (for details,
see~Appendix~\ref{app:EmbeddingKnowledge})

\begin{restatable}{theorem}{theoEmbeddingKnowledge}
  \label{theorem:EmbeddingKnowledge}
  Given an observation map $\Obs$ and a $\KLTL$ formula $\psi$ over
  $\AP$, one can construct in linear time a $\SGHLTL{\emptyset}$
  sentence $\varphi_\emptyset$ and a $\GHLTL$ sentence $\varphi$ such
  that $\varphi_\emptyset$ is equivalent to $\psi$ w.r.t.~$\Obs$ under
  the synchronous semantics and $\varphi$ is equivalent to $\psi$
  w.r.t.~$\Obs$ under the asynchronous semantics.
  Moreover, $\varphi$ is a simple $\GHLTL$ sentence if $\psi$ is in
  the single-agent fragment of $\KLTL$.
\end{restatable}

  \section{Decidability of Model Checking against Simple $\GHLTL$}\label{sec:modelCheckingFragment}

In this section, we show  that (fair) model checking against simple $\GHLTL$ is
decidable.
We first prove the result for the fragment $\SGHLTL{\emptyset}$ of
simple $\GHLTL$ by a linear-time reduction to satisfiability of
\emph{full} Quantified Propositional Temporal Logic ($\QPTL$, for
short)~\cite{SistlaVW87}, where the latter extends $\PLTL$ by
quantification over propositions.
Then, we show that (fair) model checking of simple $\GHLTL$ can be
reduced in time singly exponential in the size of the formula to
fair model checking of $\SGHLTL{\emptyset}$.
We also provide optimal complexity bounds for (fair) model checking the
fragment $\SGHLTL{\emptyset}$ in terms of a parameter of the given
formula called \emph{strong alternation depth}.
For this, we first give similar optimal complexity bounds for
satisfiability of $\QPTL$.

\noindent The syntax of $\QPTL$ formulas $\varphi$ over a finite set
$\AP$ of atomic propositions is as follows:
\[
  \varphi ::= \top \mid p \mid \neg \varphi \mid \varphi \vee \varphi
  \mid \Next \varphi \mid \Yesterday \varphi \mid \varphi \Until
  \varphi \mid \varphi \Since \varphi \mid \exists p\,.\varphi
\]
where $p\in \AP$ and $\exists p$ is the propositional existential quantifier.
A $\QPTL$ formula $\varphi$ is a \emph{sentence} if each proposition
$p$ occurs in the scope of a quantifier binding $p$ and each temporal
modality occurs in the scope of a quantifier.
By introducing $\wedge$ and the operators $\Release$ (\emph{release}, dual of
$\Until$), $\PastRelease$ (\emph{past release}, dual of $\Since$) and
$\forall p$ (propositional universal quantifier), every $\QPTL$
formula can be converted into an equivalent formula in \emph{negation normal
  form}, where negation only appears in front of atomic propositions.
$\QPTL$ formulas are interpreted over pointed traces $(\trace,i)$ over
$\AP$.
All $\QPTL$ temporal operators have the same semantics as in $\PLTL$.
The semantics of propositional quantification is as follows:
\[
(\trace,i) \models \exists p. \varphi \ \Leftrightarrow \  \text{ there is  a trace } \trace'  \text{ such that } \trace =_{\AP\setminus\{p\}} \trace' \text{ and } (\trace',i) \models \varphi
\]
where $\trace =_{\AP\setminus\{p\}} \trace'$ means that the
projections of $\trace$ and $\trace'$ over $\AP\setminus \{p\}$
coincide.
A formula $\varphi$ is satisfiable if $(\trace,0)\models \varphi$ for
some trace $\trace$.
We now give a generalization of the standard notion of alternation
depth between existential and universal quantifiers which corresponds
to the one given in~\cite{Rabe2016} for $\HCTLStar$.
Our notion takes into account also the occurrences of temporal
modalities between quantifier occurrences, but the nesting depth of
temporal modalities is not considered (intuitively, it is collapsed to
one).
Formally, the \emph{strong alternation depth} $\sad(\varphi)$ of a
$\QPTL$ formula $\varphi$ in negation normal form is inductively
defined as follows, where an existential formula is a formula of the
form $\exists p.\,\psi$, a universal formula is of the form
$\forall p.\,\psi$, and for a formula $\psi$, $\dual{\psi}$ denotes
the negation normal form of $\neg\psi$:

\begin{itemize}
\item For $\varphi=p$ and $\varphi=\neg p$ for a given $p\in\AP$:
  $\sad(\varphi)\DefinedAs 0$.
\item For $\varphi=\varphi_1\vee \varphi_2$ and for
  $\varphi=\varphi_1\wedge \varphi_2$:
  $\sad(\varphi)\DefinedAs \max(\{\sad(\varphi_1),\sad(\varphi_2)\})$.
\item For $\varphi = \exists p.\, \varphi_1$: if there is no universal
  sub-formula $\forall\psi$ of $\varphi_1$ such that
  $\sad(\forall\psi)=\sad(\varphi_1)$, then
  $\sad(\varphi)\DefinedAs\sad(\varphi_1)$. Otherwise,
  $\sad(\varphi)\DefinedAs\sad(\varphi_1)+1$.
\item For $\varphi = \forall p.\, \varphi_1$:
  $\sad(\varphi)\DefinedAs\sad(\exists p.\,\dual{\varphi_1})$.
\item For $\varphi = \Next\varphi_1$ or
  $\varphi = \Yesterday\varphi_1$:
  $\sad(\varphi)\DefinedAs \sad(\varphi_1)$.
\item For $\varphi = \varphi_1\Until \varphi_2$ or
  $\varphi = \varphi_1\Since \varphi_2$: let $h$ be the maximum over
  the strong alternation depths of the universal and existential
  sub-formulas of $\varphi_1$ and $\varphi_2$ (the maximum of the empty set is
  $0$).
  If the following conditions are met, then
  $\sad(\varphi)\DefinedAs h$; otherwise,
  $\sad(\varphi)\DefinedAs h+1$:
  \begin{itemize}
  \item there is no existential or universal sub-formula $\psi$ of
    $\varphi_1$ with $\sad(\psi)=h$;
  \item there is no universal sub-formula $\psi$ of $\varphi_2$ with
    $\sad(\psi)=h$;
  \item no existential formula $\psi$ with $\sad(\psi)=h$ occurs in
    the left operand (resp., right operand) of a sub-formula of
    $\varphi_2$ of the form $\psi_1\Op \psi_2$, where
    $\Op\in\{\Until,\Since\}$ (resp.,
    $\Op\in \{\Release,\PastRelease\}$).
  \end{itemize}
\item Finally, for $\varphi = \varphi_1\Release \varphi_2$ or
  $\varphi = \varphi_1\PastRelease \varphi_2$:
  $\sad(\varphi)\DefinedAs \sad(\dual{\varphi})$.
\end{itemize}

\noindent For example, $\sad(\exists p.(p\Until\exists q.q))=0$ and
$\sad(\exists p.(\exists p.p\Until q))=1$.
The strong alternation depth of an arbitrary $\QPTL$ formula
corresponds to the one of its negation normal form.
The strong alternation depth of a simple $\GHLTL$ formula is defined
similarly but we replace quantification over propositions with
quantification over trace variables.
For all $n,h\in\nat$, $\Tower(h,n)$ denotes a tower of exponentials of
height $h$ and argument $n$: $\Tower(0,n)=n$ and
$\Tower(h+1,n)=2^{\Tower(h,n)}$.
Essnetially, the strong alternation depth corresponds to the
(unavoidable) power set construction related to the alternation of
quantifiers to solve the model-checking problem.

The following result represents an improved version of Theorem~6
in~\cite{BozzelliMP15} where $h$-\EXPSPACE\ is the class of languages
decided by deterministic Turing machines bounded in space by functions
of $n$ in $O(\Tower(h,n^c))$ for some constant $c\geq 1$.
While the lower bound directly follows from~\cite{SistlaVW87}, the
upper bound improves the result in~\cite{BozzelliMP15}, since there,
occurrences of temporal modalities immediately preceding propositional
quantification always count as additional alternations (for details,
see Appendix~\ref{app:QPTLsatisfiability}).

\begin{theorem}\label{theorem:QPTLsatisfiability}
  For all $h\geq 0$, satisfiability of $\QPTL$ sentences with strong
  alternation depth at most $h$ is $h$-\EXPSPACE-complete.
\end{theorem}

\noindent \textbf{(Fair) Model checking against $\SGHLTL{\emptyset}$.}
We provide now linear-time reductions of (fair) model checking against
$\SGHLTL{\emptyset}$ to (and from) satisfiability of $\QPTL$ which
preserve the strong alternation depth.
We start with the reduction of (fair) model checking $\SGHLTL{\emptyset}$ to
$\QPTL$ satisfiability.

\begin{restatable}{theorem}{theoFromEmptyGammaFragmentToQPTL}
  \label{theorem:fromEmptyGammaFragmentToQPTL}
  Given a fair finite Kripke structure $(\KS,\FStates)$ and a
  $\SGHLTL{\emptyset}$ sentence $\varphi$, one can construct in linear
  time a $\QPTL$ sentence $\psi$  with the same strong alternation
  depth as $\varphi$  such that $\psi$ is satisfiable if and only if
  $\Lang(\KS,\FStates)\models\varphi$.
\end{restatable}
\begin{proof}[Sketched proof]
  Let $\KS= \mktuple{\States,\States_0,\Trans,\Lab}$.
  In the reduction of model checking $(\KS,\FStates)$ against
  $\varphi$ to $\QPTL$ satisfiability, we need to merge multiple traces
  into a unique trace where just one position is considered at any
  time.
  An issue is that the hyper quantifiers range over arbitrary pointed
  traces so that the positions of the different pointed traces in the
  current trace assignment do not necessarily coincide (intuitively, the different
  pointed traces are not aligned with respect to the relative current
  positions).
  However, we can solve this issue because the offsets between the
  positions of the pointed traces in the current trace assignment
  remain constant during the evaluation of the temporal modalities.
  In particular, assume that $(\trace,i)$ is the first pointed trace
  selected by a hyper quantifier during the evaluation along a path in
  the syntax tree of $\varphi$.
  We encode $\trace$ by keeping track also of the variable $\vartrace$
  to which $(\trace,i)$ is bound and the $\FStates$-fair path of $\KS$
  whose associated trace is $\trace$.
  Let $(\trace',i')$ be another pointed trace introduced by another
  hyper quantifier $y$ during the evaluation of $\varphi$.
  If $i'<i$, we consider an encoding of $\trace'$ which is similar to
  the previous encoding but we precede it with a \emph{padding prefix}
  of length $i-i'$ of the form
  $\{\pad_{\overrightarrow{\vartraceAux}}\}^{i-i'}$.
  The arrow $\rightarrow$ indicates that the encoding is along the
  \emph{forward direction}.
  Now, assume that $i'>i$.
  In this case, the encoding of $\trace'$ is the merging of two
  encodings over disjoint sets of propositions: one along the forward
  direction which encodes the suffix $(\trace')^{i'-i}$ and another
  one along the \emph{backward direction} which is of the form
  $\{\pad_{\overleftarrow{\vartraceAux}}\}\cdot \rho
  \cdot\{\pad_{\overleftarrow{\vartraceAux}}\}^{\omega}$, where $\rho$
  is a backward encoding of the \emph{reverse} of the prefix of
  $\trace'$ of length $i'-i$.
  In such a way, the encodings of the pointed traces later introduced
  in the evaluation of $\varphi$ are aligned with the reference
  pointed trace $(\trace,i)$.
  Since the positions in the backward direction overlap some positions
  in the forward direction, in the translation, we keep track of whether
  the current position refers to the forward or to the backward
  direction.
  The details of the reduction are to $\QPTL$ satisfiability can be found
  in Appendix~\ref{app:fromEmptyGammaFragmentToQPTL}.
\end{proof}

By an adaptation of the known reduction of satisfiability of $\QPTL$
without past to model checking of $\HCTLStar$~\cite{ClarksonFKMRS14},
we obtain the following result (for details, see Appendix~\ref{app:fromQPTLtoEmptyGammaFragment}).

\begin{restatable}{theorem}{theoFromQPTLtoEmptyGammaFragment}
  \label{theorem:fromQPTLtoEmptyGammaFragment}
  Given a $\QPTL$ sentence $\psi$ over $\AP$, one can build in
  linear time a finite Kripke structure $\KS_\AP$ (depending only on
  $\AP$) and a singleton-free $\SGHLTL{\emptyset}$ sentence $\varphi$
  having the same strong alternation depth as $\psi$ such that $\psi$
  is satisfiable \emph{iff} $ \Lang(\KS_\AP)\models\varphi$.
\end{restatable}

Hence, by
Theorems~\ref{theorem:QPTLsatisfiability}--\ref{theorem:fromQPTLtoEmptyGammaFragment},
we obtain the following result.

\begin{corollary}
  For all $h\geq 0$,   (fair) model checking against
  $\SGHLTL{\emptyset}$ sentences with strong alternation depth at most $h$ is
  $h$-\EXPSPACE-complete.
\end{corollary}

\noindent
\textbf{Reduction to fair model checking against $\SGHLTL{\emptyset}$.}
We solve the (fair) model checking problem for simple $\GHLTL$ by a reduction to
fair model checking against the fragment $\SGHLTL{\emptyset}$.
Our reduction is exponential in the size of the given sentence and is
an adaptation of the  reduction from model checking simple
$\SHLTL$ to model checking $\HLTL$ shown in~\cite{BozzelliPS21}.
As a preliminary step, we first show, by an easy adaptation of the
standard automata-theoretic approach for $\PLTL$~\cite{VardiW94}, that
the problem for a simple $\GHLTL$ sentence $\varphi$ can be reduced in
exponential time to the fair model checking problem against a
singleton-free sentence in the fragment $\SGHLTL{\Gamma}$ for some set
$\Gamma$ of \emph{atomic propositions} depending on $\varphi$.
For details, see Appendix~\ref{app:FormSimpleToProposition}.

\begin{restatable}{proposition}{propFormSimpleToProposition}
  \label{prop:FormSimpleToProposition}
  Given a simple $\GHLTL$ sentence $\varphi$ and a fair finite Kripke
  structure $(\KS,\FStates)$ over $\AP$, one can build in single
  exponential time in the size of $\varphi$, a fair finite Kripke
  structure $(\KS',\FStates')$ over an extension $\AP\,'$ of $\AP$ and
  a singleton-free $\SGHLTL{\Gamma'}$ sentence $\varphi'$ for some
  $\Gamma'\subseteq \AP\,'$ such that
  $\Lang(\KS',\FStates')\models \varphi'$ if and only if
  $\Lang(\KS,\FStates)\models \varphi$.
  Moreover, $\varphi'$ has the same strong alternation depth as
  $\varphi$, $|\varphi'|=O(|\varphi|)$, and
  $|\KS'|=O(|\KS|* 2^{O(|\varphi|)})$.
\end{restatable}

Let us fix a non-empty finite set $\Gamma\subseteq \AP$ of atomic
propositions.
We now show that fair model checking of the singleton-free fragment of
$\SGHLTL{\Gamma}$ can be reduced in polynomial time to fair model
checking of $\SGHLTL{\emptyset}$.
We observe that in the singleton-free fragment of $\SGHLTL{\Gamma}$,
when a pointed trace $(\trace,i)$ is selected by a pointed quantifier
$\exists^{\Pt}\vartrace$,  the positions of $\trace$ which are
visited during the evaluation of the temporal modalities are all in
the $(\Gamma,\omega)$-stutter factorization of $\trace$ with the
possible exception of the position $i$ chosen by
$\exists^{\Pt}\vartrace$.
Thus, given a set $\Lang$ of traces and a special proposition
$\pad\notin \AP$, we define an extension
$\stfr_{\Gamma}^{\pad}(\Lang)$ of the set
$\stfr_{\Gamma}(\Lang)=\{\stfr_{\Gamma}(\trace)\mid \trace\in
\Lang\}$ as follows.
Intuitively, we consider for each trace $\trace\in\Lang$, its
$\Gamma$-stutter trace $\stfr_{\Gamma}(\trace)$ and the extensions of
$\stfr_{\Gamma}(\trace)$ which are obtained by adding an extra
position marked by proposition $\#$ (this extra position does not
belong to the $(\Gamma,\omega)$-stutter factorization of $\trace$).
Formally, $\stfr_{\Gamma}^{\pad}(\Lang)$ is the smallest set
containing $\stfr_{\Gamma}(\Lang)$ and satisfying the following
condition:
\begin{itemize}
\item for each trace $\trace\in \Lang$ with $(\Gamma,\omega)$-stutter
  factorization $\{\ell_k\}_{k\geq 0}$ and position
  $i\in (\ell_k,\ell_{k+1})$ for some $k\geq 0$, the trace
  $\trace(\ell_0)\ldots
  \trace(\ell_k)\,(\trace(i)\cup\{\pad\})\,\trace(\ell_{k+1})\,\trace(\ell_{k+2})\ldots
  \in \stfr_{\Gamma}^{\pad}(\Lang)$.
\end{itemize}

Given a singleton-free formula $\varphi$ in $\SGHLTL{\Gamma}$, we
denote by $\TMap_{\pad}(\varphi)$ the formula in $\SGHLTL{\emptyset}$
obtained from $\varphi$ by applying inductively the following
transformations:
\begin{itemize}
\item the $\Gamma$-relativized temporal modalities are replaced with
  their $\emptyset$-relativized counterparts;
\item each formula $\exists^{\Pt} \vartrace.\,\phi$ is replaced with
  $\exists^{\Pt}\vartrace.\,\bigl(\TMap_{\pad}(\phi)\wedge
  \ctx{x}(\Next\Always\neg\Rel{\pad}{\vartrace} \wedge (\Yesterday\top
  \rightarrow \Yesterday\Historically
  \neg\Rel{\pad}{\vartrace}))\bigr)$.
\end{itemize}

Intuitively, the formula
$\TMap_{\pad}(\exists^{\Pt} \vartrace.\,\phi)$ states that for the
pointed trace $(\trace,i)$ selected by the pointed quantifier, at most
position $i$ may be marked by the special proposition $\#$.
By the semantics of the logics considered, the following holds.

\begin{remark}\label{remark:StutteringExtension}
  Given a singleton-free sentence $\varphi$ of $\SGHLTL{\Gamma}$ and a
  set $\Lang$ of traces, it holds that $\Lang\models \varphi$ if and
  only if $\stfr_{\Gamma}^{\pad}(\Lang)\models \TMap_{\#}(\varphi)$.
\end{remark}

Let us fix now a fair finite Kripke structure $(\KS,\FStates)$.
We first show that one can build in polynomial time a finite
Kripke structure $(\KS_\Gamma,\FStates_\Gamma)$ and a $\LTL$ formula
$\theta_\Gamma$ such that $\stfr_{\Gamma}^{\pad}(\Lang(\KS,\FStates))$
coincides with the traces of $\Lang(\KS_\Gamma,\FStates_\Gamma)$ satisfying
$\theta_\Gamma$ (details are  in~Appendix~\ref{app:ExtendedStutterTrace}).

\begin{restatable}{proposition}{propExtendedStutterTrace}
  \label{prop:ExtendedStutterTrace}
  Given $\emptyset\neq \Gamma\subseteq \AP$ and a fair finite Kripke
  structure $(\KS,\FStates)$ over $\AP$, one can construct in
  polynomial time a fair finite Kripke structure
  $(\KS_\Gamma,\FStates_\Gamma)$ and a $\LTL$ formula $\theta_\Gamma$
  such that $\stfr_{\Gamma}^{\pad}(\Lang(\KS,\FStates))$ is the set of
  traces $\trace\in\Lang(\KS_\Gamma,\FStates_\Gamma)$ so that
  $\trace\models \theta_\Gamma$.
\end{restatable}

Fix now a singleton-free sentence $\varphi$ of $\SGHLTL{\Gamma}$.
For the given fair finite Kripke structure $(\KS,\FStates)$ over
$\AP$, let $(\KS_\Gamma,\FStates_\Gamma)$ and $\theta_\Gamma$ as in
the statement of Proposition~\ref{prop:ExtendedStutterTrace}.
We consider the $\SGHLTL{\emptyset}$ sentence $\TMap(\varphi)$
obtained from $\TMap_{\pad}(\varphi)$ by inductively replacing 
each subformula $\exists^{\Pt} \vartrace.\,\phi$ of
$\TMap_{\pad}(\varphi)$ with
$\exists^{\Pt} \vartrace.\,(\TMap(\phi)\wedge
\ctx{\vartrace}\Once(\neg\Yesterday\top\wedge\Rel{\theta_\Gamma}{\vartrace}))$.
In other terms, we ensure that in $\TMap_{\pad}(\varphi)$ the hyper
quantification ranges over traces which satisfy the $\LTL$ formula
$\theta_{\Gamma}$.
By Remark~\ref{remark:StutteringExtension} and
Proposition~\ref{prop:ExtendedStutterTrace}, we obtain that
$\Lang(\KS,\FStates)\models \varphi$ if and only if
$\Lang(\KS_\Gamma,\FStates_\Gamma)\models \TMap(\varphi)$.
Thus, together with Proposition~\ref{prop:FormSimpleToProposition}, we obtain the
following result.

\begin{theorem}\label{theor:ReductionToEmptyGammaFragment}
  The (fair) model checking problem against simple $\GHLTL$ can be
  reduced in singly exponential time to fair model checking against
  $\SGHLTL{\emptyset}$.
\end{theorem}


  % ...
 \section{Conclusion}

We have introduced  a novel hyper logic $\GHLTL$  which merges two
known asynchronous temporal logics for hyperproperties, namely
stuttering $\HLTL$ and context $\HLTL$.
Even though model checking of the resulting logic $\GHLTL$ is undecidable, we have
identified a useful fragment, called \emph{simple $\GHLTL$}, that has a decidable model checking,
is strictly more expressive than $\HLTL$ and than previously proposed
 fragments of asynchronous temporal logics for
hyperproperties with a decidable model checking. For the fragment
$\SGHLTL{\emptyset}$ of simple $\GHLTL$, we have given optimal complexity bounds of (fair) model checking
in terms of the strong alternation depth of the given sentence. For arbitrary sentences in simple $\GHLTL$,
(fair) model checking is reduced in exponential time to fair model checking of $\SGHLTL{\emptyset}$.
It is worth noting that simple $\GHLTL$ can express non-regular trace properties
over singleton sets of traces which are not definable in $\MSOE$. An intriguing open
question is whether $\FOE$ can be embedded in simple $\GHLTL$.
In a companion paper, we study
asynchronous properties on finite traces by adapting simple $\GHLTL$ in prenex
form to
finite traces, 
and introduce practical
model-checking algorithms for useful fragments of this logic.


  \bibliographystyle{plainurl}
  \bibliography{biblio}

\begin{thebibliography}{10}

\bibitem{BaumeisterCBFS21}
Jan Baumeister, Norine Coenen, Borzoo Bonakdarpour, Bernd Finkbeiner, and
  C\'{e}sar S{\'{a}}nchez.
\newblock A {T}emporal {L}ogic for {A}synchronous {H}yperproperties.
\newblock In {\em Proc. 33rd {CAV}}, volume 12759 of {\em LNCS 12759}, pages
  694--717. Springer, 2021.
\newblock \href {https://doi.org/10.1007/978-3-030-81685-8\_33}
  {\path{doi:10.1007/978-3-030-81685-8\_33}}.

\bibitem{BeutnerFFM23}
Raven Beutner, Bernd Finkbeiner, Hadar Frenkel, and Niklas Metzger.
\newblock Second-{O}rder {H}yperproperties.
\newblock In {\em Proc. 35th {CAV}}, volume 13965 of {\em Lecture Notes in
  Computer Science}, pages 309--332. Springer, 2023.
\newblock \href {https://doi.org/10.1007/978-3-031-37703-7\_15}
  {\path{doi:10.1007/978-3-031-37703-7\_15}}.

\bibitem{BittnerBCGTV22}
Benjamin Bittner, Marco Bozzano, Alessandro Cimatti, Marco Gario, Stefano
  Tonetta, and Viktoria Voz{\'{a}}rov{\'{a}}.
\newblock Diagnosability of fair transition systems.
\newblock {\em Artif. Intell.}, 309:103725, 2022.
\newblock \href {https://doi.org/10.1016/J.ARTINT.2022.103725}
  {\path{doi:10.1016/J.ARTINT.2022.103725}}.

\bibitem{BozzanoCGT15}
Marco Bozzano, Alessandro Cimatti, Marco Gario, and Stefano Tonetta.
\newblock Formal {D}esign of {A}synchronous {F}ault {D}etection and
  {I}dentification {C}omponents using {T}emporal {E}pistemic {L}ogic.
\newblock {\em Log. Methods Comput. Sci.}, 11(4), 2015.
\newblock \href {https://doi.org/10.2168/LMCS-11(4:4)2015}
  {\path{doi:10.2168/LMCS-11(4:4)2015}}.

\bibitem{BozzelliMP15}
Laura Bozzelli, Bastien Maubert, and Spophie Pinchinat.
\newblock Unifying {H}yper and {E}pistemic {T}emporal {L}ogics.
\newblock In {\em Proc. 18th FoSSaCS}, LNCS 9034, pages 167--182. Springer,
  2015.
\newblock \href {https://doi.org/10.1007/978-3-662-46678-0\_11}
  {\path{doi:10.1007/978-3-662-46678-0\_11}}.

\bibitem{BozzelliPS21}
Laura Bozzelli, Adriano Peron, and C\'{e}sar S{\'{a}}nchez.
\newblock {A}synchronous {E}xtensions of {HyperLTL}.
\newblock In {\em Proc. 36th {LICS}}, pages 1--13. {IEEE}, 2021.
\newblock \href {https://doi.org/10.1109/LICS52264.2021.9470583}
  {\path{doi:10.1109/LICS52264.2021.9470583}}.

\bibitem{BozzelliPS22}
Laura Bozzelli, Adriano Peron, and C\'{e}sar S{\'{a}}nchez.
\newblock {E}xpressiveness and {D}ecidability of {T}emporal {L}ogics for
  {A}synchronous {H}yperproperties.
\newblock In {\em Proc. 33rd {CONCUR}}, volume 243 of {\em LIPIcs}, pages
  27:1--27:16. Schloss Dagstuhl - Leibniz-Zentrum f{\"{u}}r Informatik, 2022.
\newblock \href {https://doi.org/10.4230/LIPICS.CONCUR.2022.27}
  {\path{doi:10.4230/LIPICS.CONCUR.2022.27}}.

\bibitem{ClarksonFKMRS14}
Michael~R. Clarkson, Bernd Finkbeiner, Masoud Koleini, Kristopher~K. Micinski,
  Markus~N. Rabe, and C\'{e}sar S{\'a}nchez.
\newblock Temporal {L}ogics for {H}yperproperties.
\newblock In {\em Proc. 3rd POST}, LNCS 8414, pages 265--284. Springer, 2014.
\newblock \href {https://doi.org/10.1007/978-3-642-54792-8\_15}
  {\path{doi:10.1007/978-3-642-54792-8\_15}}.

\bibitem{ClarksonS10}
Michael~R. Clarkson and Fred~B. Schneider.
\newblock Hyperproperties.
\newblock {\em Journal of Computer Security}, 18(6):1157--1210, 2010.
\newblock \href {https://doi.org/10.3233/JCS-2009-0393}
  {\path{doi:10.3233/JCS-2009-0393}}.

\bibitem{CoenenFHH19}
Norine Coenen, Bernd Finkbeiner, Christopher Hahn, and Jana Hofmann.
\newblock The hierarchy of hyperlogics.
\newblock In {\em Proc. 34th {LICS}}, pages 1--13. {IEEE}, 2019.
\newblock \href {https://doi.org/10.1109/LICS.2019.8785713}
  {\path{doi:10.1109/LICS.2019.8785713}}.

\bibitem{DimitrovaFKRS12}
Rayna Dimitrova, Bernd Finkbeiner, M\'{a}t\'{e}~M Kov{\'a}cs, Markus~N. Rabe,
  and Helmut Seidl.
\newblock Model {C}hecking {I}nformation {F}low in {R}eactive {S}ystems.
\newblock In {\em Proc. 13th VMCAI}, LNCS 7148, pages 169--185. Springer, 2012.
\newblock \href {https://doi.org/10.1007/978-3-642-27940-9\_12}
  {\path{doi:10.1007/978-3-642-27940-9\_12}}.

\bibitem{EmersonH86}
E.~Allen Emerson and Joseph~Y. Halpern.
\newblock ``{S}ometimes'' and ``{N}ot {N}ever'' revisited: on branching versus
  linear time temporal logic.
\newblock {\em J. ACM}, 33(1):151--178, 1986.
\newblock \href {https://doi.org/10.1145/4904.4999}
  {\path{doi:10.1145/4904.4999}}.

\bibitem{fagin1995reasoning}
Ronald Fagin, Joseph~Y. Halpern, Yoram Moses, and Moshe~Y. Vardi.
\newblock {\em Reasoning about knowledge}, volume~4.
\newblock MIT Press Cambridge, 1995.
\newblock \href {https://doi.org/10.7551/mitpress/5803.001.0001}
  {\path{doi:10.7551/mitpress/5803.001.0001}}.

\bibitem{FinkbeinerH16}
Bernd Finkbeiner and Christopher Hahn.
\newblock Deciding {H}yperproperties.
\newblock In {\em Proc. 27th {CONCUR}}, LIPIcs 59, pages 13:1--13:14. Schloss
  Dagstuhl - Leibniz-Zentrum f{\"{u}}r Informatik, 2016.
\newblock \href {https://doi.org/10.4230/LIPIcs.CONCUR.2016.13}
  {\path{doi:10.4230/LIPIcs.CONCUR.2016.13}}.

\bibitem{Finkbeiner017}
Bernd Finkbeiner and Martin Zimmermann.
\newblock The first-order logic of hyperproperties.
\newblock In {\em Proc. 34th {STACS}}, LIPIcs 66, pages 30:1--30:14. Schloss
  Dagstuhl - Leibniz-Zentrum f{\"{u}}r Informatik, 2017.
\newblock \href {https://doi.org/10.4230/LIPIcs.STACS.2017.30}
  {\path{doi:10.4230/LIPIcs.STACS.2017.30}}.

\bibitem{FischerL79}
Michael~J. Fischer and Richard~E. Ladner.
\newblock {P}ropositional {D}ynamic {L}ogic of {R}egular {P}rograms.
\newblock {\em J. Comput. Syst. Sci.}, 18(2):194--211, 1979.
\newblock \href {https://doi.org/10.1016/0022-0000(79)90046-1}
  {\path{doi:10.1016/0022-0000(79)90046-1}}.

\bibitem{goguen1982security}
Joseph~A. Goguen and Jos\'e Meseguer.
\newblock Security {P}olicies and {S}ecurity {M}odels.
\newblock In {\em {IEEE} Symposium on Security and Privacy}, pages 11--20.
  {IEEE} Computer Society, 1982.
\newblock \href {https://doi.org/10.1109/SP.1982.10014}
  {\path{doi:10.1109/SP.1982.10014}}.

\bibitem{GutsfeldMOV22}
Jens~Oliver Gutsfeld, Arne Meier, Christoph Ohrem, and Jonni Virtema.
\newblock Temporal {T}eam {S}emantics {R}evisited.
\newblock In {\em Proc. 37th {LICS}}, pages 44:1--44:13. {ACM}, 2022.
\newblock \href {https://doi.org/10.1145/3531130.3533360}
  {\path{doi:10.1145/3531130.3533360}}.

\bibitem{GutsfeldMO20}
Jens~Oliver Gutsfeld, Markus M{\"{u}}ller{-}Olm, and Christoph Ohrem.
\newblock Propositional dynamic logic for hyperproperties.
\newblock In {\em Proc. 31st {CONCUR}}, LIPIcs 171, pages 50:1--50:22. Schloss
  Dagstuhl - Leibniz-Zentrum f{\"{u}}r Informatik, 2020.
\newblock \href {https://doi.org/10.4230/LIPIcs.CONCUR.2020.50}
  {\path{doi:10.4230/LIPIcs.CONCUR.2020.50}}.

\bibitem{GutsfeldOO21}
Jens~Oliver Gutsfeld, Markus M{\"{u}}ller{-}Olm, and Christoph Ohrem.
\newblock Automata and fixpoints for asynchronous hyperproperties.
\newblock {\em Proc. {ACM} Program. Lang.}, 4({POPL}), 2021.
\newblock \href {https://doi.org/10.1145/3434319} {\path{doi:10.1145/3434319}}.

\bibitem{HalpernO08}
Joseph~Y. Halpern and Kevin~R. O'Neill.
\newblock Secrecy in multiagent systems.
\newblock {\em ACM Trans. Inf. Syst. Secur.}, 12(1), 2008.

\bibitem{HalpernV86}
Joseph~Y. Halpern and Moshe~Y. Vardi.
\newblock The {C}omplexity of {R}easoning about {K}nowledge and {T}ime:
  {E}xtended {A}bstract.
\newblock In {\em Proc. 18th {STOC}}, pages 304--315. {ACM}, 1986.
\newblock \href {https://doi.org/10.1145/12130.12161}
  {\path{doi:10.1145/12130.12161}}.

\bibitem{KrebsMV018}
Andreas Krebs, Arne Meier, Jonni Virtema, and Martin Zimmermann.
\newblock Team {S}emantics for the {S}pecification and {V}erification of
  {H}yperproperties.
\newblock In {\em Proc. 43rd {MFCS}}, LIPIcs 117, pages 10:1--10:16. Schloss
  Dagstuhl - Leibniz-Zentrum f{\"{u}}r Informatik, 2018.
\newblock \href {https://doi.org/10.4230/LIPIcs.MFCS.2018.10}
  {\path{doi:10.4230/LIPIcs.MFCS.2018.10}}.

\bibitem{KupfermanPV09}
Orna Kupferman, Nir Piterman, and Moshe~Y. Vardi.
\newblock From liveness to promptness.
\newblock {\em Formal Methods Syst. Des.}, 34(2):83--103, 2009.
\newblock \href {https://doi.org/10.1007/S10703-009-0067-Z}
  {\path{doi:10.1007/S10703-009-0067-Z}}.

\bibitem{KupfermanV01}
Orna Kupferman and Moshe~Y. Vardi.
\newblock Weak alternating automata are not that weak.
\newblock {\em ACM Transactions on Computational Logic}, 2(3):408--429, 2001.
\newblock \href {https://doi.org/10.1145/377978.377993}
  {\path{doi:10.1145/377978.377993}}.

\bibitem{KupfermanVW00}
Orna Kupferman, Moshe~Y. Vardi, and Pierre Wolper.
\newblock An {A}utomata-{T}heoretic {A}pproach to {B}ranching-{T}ime {M}odel
  {C}hecking.
\newblock {\em J. ACM}, 47(2):312--360, 2000.
\newblock \href {https://doi.org/10.1145/333979.333987}
  {\path{doi:10.1145/333979.333987}}.

\bibitem{Luck20}
Martin L{\"{u}}ck.
\newblock On the complexity of linear temporal logic with team semantics.
\newblock {\em Theor. Comput. Sci.}, 837:1--25, 2020.
\newblock \href {https://doi.org/10.1016/j.tcs.2020.04.019}
  {\path{doi:10.1016/j.tcs.2020.04.019}}.

\bibitem{MP92}
Zohar Manna and Amir Pnueli.
\newblock {\em {The Temporal Logic of Reactive and Concurrent Systems -
  Specification}}.
\newblock Springer-Verlag, 1992.
\newblock \href {https://doi.org/10.1007/978-1-4612-0931-7}
  {\path{doi:10.1007/978-1-4612-0931-7}}.

\bibitem{McLean96}
John~D. McLean.
\newblock A {G}eneral {T}heory of {C}omposition for a {C}lass of
  "{P}ossibilistic'' {P}roperties.
\newblock {\em {IEEE} Trans. Software Eng.}, 22(1):53--67, 1996.
\newblock \href {https://doi.org/10.1109/32.481534}
  {\path{doi:10.1109/32.481534}}.

\bibitem{MiyanoH84}
S.~Miyano and T.~Hayashi.
\newblock Alternating finite automata on $\omega$-words.
\newblock {\em Theoretical Computer Science}, 32:321--330, 1984.
\newblock \href {https://doi.org/10.1016/0304-3975(84)90049-5}
  {\path{doi:10.1016/0304-3975(84)90049-5}}.

\bibitem{Pnueli77}
Amir Pnueli.
\newblock The {T}emporal {L}ogic of {P}rograms.
\newblock In {\em Proc. 18th FOCS}, pages 46--57. IEEE Computer Society, 1977.
\newblock \href {https://doi.org/10.1109/SFCS.1977.32}
  {\path{doi:10.1109/SFCS.1977.32}}.

\bibitem{Rabe2016}
Markus~N. Rabe.
\newblock {\em A temporal logic approach to information-flow control}.
\newblock PhD thesis, Saarland University, 2016.

\bibitem{SampathSLST95}
Meera Sampath, Raja Sengupta, Stephen Lafortune, Kazin Sinnamohideen, and
  Demosthenis Teneketzis.
\newblock Diagnosability of discrete-event systems.
\newblock {\em {IEEE} Trans. Autom. Control.}, 40(9):1555--1575, 1995.
\newblock \href {https://doi.org/10.1109/9.412626}
  {\path{doi:10.1109/9.412626}}.

\bibitem{SistlaVW87}
A.~Prasad Sistla, Moshe~Y. Vardi, and Pierre Wolper.
\newblock The {C}omplementation {P}roblem for {B}{\"u}chi {A}utomata with
  {A}pplications to {T}emporal {L}ogic.
\newblock {\em Theoretical Computer Science}, 49:217--237, 1987.
\newblock \href {https://doi.org/10.1016/0304-3975(87)90008-9}
  {\path{doi:10.1016/0304-3975(87)90008-9}}.

\bibitem{MeydenS99}
Ron van~der Meyden and Nikolay~V. Shilov.
\newblock Model checking knowledge and time in systems with perfect recall
  (extended abstract).
\newblock In {\em Proc. 19th FSTTCS}, LNCS 1738, pages 432--445. Springer,
  1999.
\newblock \href {https://doi.org/10.1007/3-540-46691-6\_35}
  {\path{doi:10.1007/3-540-46691-6\_35}}.

\bibitem{Var88}
Moshe~Y. Vardi.
\newblock A temporal fixpoint calculus.
\newblock In {\em Proc. 15th POPL}, pages 250--259. ACM, 1988.

\bibitem{VardiW94}
Moshe~Y. Vardi and Pierre Wolper.
\newblock Reasoning about infinite computations.
\newblock {\em Inf. Comput.}, 115(1):1--37, 1994.
\newblock \href {https://doi.org/10.1006/inco.1994.1092}
  {\path{doi:10.1006/inco.1994.1092}}.

\bibitem{VirtemaHFK021}
Jonni Virtema, Jana Hofmann, Bernd Finkbeiner, Juha Kontinen, and Fan Yang.
\newblock Linear-{T}ime {T}emporal {L}ogic with {T}eam {S}emantics:
  {E}xpressivity and {C}omplexity.
\newblock In {\em Proc. 41st {IARCS} {FSTTCS}}, LIPIcs 213, pages 52:1--52:17.
  Schloss Dagstuhl - Leibniz-Zentrum f{\"{u}}r Informatik, 2021.
\newblock \href {https://doi.org/10.4230/LIPIcs.FSTTCS.2021.52}
  {\path{doi:10.4230/LIPIcs.FSTTCS.2021.52}}.

\bibitem{ZdancewicM03}
Steve Zdancewic and Andrew~C. Myers.
\newblock Observational {D}eterminism for {C}oncurrent {P}rogram {S}ecurity.
\newblock In {\em Proc. 16th {IEEE} CSFW-16}, pages 29--43. {IEEE} Computer
  Society, 2003.
\newblock \href {https://doi.org/10.1109/CSFW.2003.1212703}
  {\path{doi:10.1109/CSFW.2003.1212703}}.

\bibitem{Zielonka98}
W.~Zielonka.
\newblock Infinite games on finitely coloured graphs with applications to
  automata on infinite trees.
\newblock {\em Theoretical Computer Science}, 200(1-2):135--183, 1998.
\newblock \href {https://doi.org/10.1016/S0304-3975(98)00009-7}
  {\path{doi:10.1016/S0304-3975(98)00009-7}}.

\end{thebibliography}
% ...
  \newpage
\appendix




\section{Proofs from Section~\ref{sec:NovelLogic}}\label{app:NovelLogic}

\subsection{Equally-expressive  extension of simple $\GHLTL$}\label{app:extensionSimpleGHLTL}

In this section, we show that the extension of simple $\GHLTL$ as defined at the end of Subsection~\ref{subSec:MeaningfulNovelFragments}
is not more expressive than simple $\GHLTL$. In the following,  a \emph{singleton-context} formula is a
quantifier-free $\GHLTL$ formula
of the form $\ctx{\vartrace}\xi$, where $\xi$  only uses singleton contexts $\ctx{\vartraceAux}$ and temporal modalities with subscript $\emptyset$.
A singleton-context formula is \emph{simple} if it is of the form $\ctx{\vartrace}\Rel{\psi}{\vartrace}$ for some $\PLTL$ formula $\psi$.
Thus, the considered extension of simple $\GHLTL$ is obtained by replacing simple singleton-context sub-formulas with arbitrary singleton-context sub-formulas.
Given two quantifier-free $\GHLTL$ formulas $\varphi$ and $\varphi'$, we say that $\varphi$ and $\varphi'$ are \emph{equivalent} if
(i) $\varphi$ and $\varphi'$ use the same trace variables, and (ii) for each  trace assignment $\TracesMap$ whose domain contains the variables of $\varphi$,
$(\TracesMap,\Var)\models \varphi$ \emph{iff} $(\TracesMap,\Var)\models \varphi'$.
In order to show  that the considered extension of simple $\GHLTL$ has the same expressiveness as simple $\GHLTL$, it suffices to show the following result.

\begin{proposition}\label{prop:ExtesionofSimpleGHLTL} Given a singleton-context  formula $\ctx{\vartrace}\xi$, one can construct a Boolean combination $\varphi$ of simple singleton-context  formulas and relativized propositions such that $\varphi$ and $\ctx{\vartrace}\xi$ are equivalent.
\end{proposition}
\begin{proof}
For the proof, we need  additional definitions. Let $\Lambda$ be a finite set whose elements are simple singleton-context formulas or relativized atomic propositions.
A \emph{guess} for $\Lambda$  is a mapping $\HMap:\Lambda \mapsto \{\top,\neg\top\}$ assigning to each formula in $\Lambda$ a Boolean value ($\top$ for \emph{true}, and $\neg\top$ for \emph{false}). We denote by $G_\Lambda$ the finite set of guesses for $\Lambda$.

Fix a singleton-context  formula $\ctx{\vartrace}\xi$ which is not simple.  We prove Proposition~\ref{prop:ExtesionofSimpleGHLTL} by induction on  the nesting depth of the context modalities in $\xi$.

For the base case, $\xi$ does not contain context modalities. Let $\Lambda$ be the set of relativized atomic propositions $\Rel{p}{\vartraceAux}$ occurring in $\xi$ such that
$\vartraceAux\neq \vartrace$. Then, the formula $\varphi$ satisfying Proposition~\ref{prop:ExtesionofSimpleGHLTL} is given by
\[
\varphi\DefinedAs\bigvee_{\HMap\in G_\Lambda}(\xi(\HMap)\wedge \bigwedge_{\theta\in \{\theta\in \Lambda\mid \HMap(\theta)=\top\}} \theta \wedge \bigwedge_{\theta\in \{\theta\in \Lambda\mid \HMap(\theta)=\neg\top\}}\neg\theta)
\]
where $\xi(\HMap)$ is obtained from $\xi$ by replacing all occurrences of the sub-formulas $\theta\in\Lambda$ in $\xi$ with $\HMap(\theta)$. Note that
$\varphi$ is a Boolean combination of simple singleton-context formulas and relativized propositions. Correctness of the construction follows from the fact that
by the semantics of $\GHLTL$, the position of the pointed trace assigned to a variable $\vartraceAux$ distinct from $\vartrace$ remains unchanged during the valuation of
the formula $\xi$ within the singleton context $\ctx{\vartrace}$.

For the induction step, assume that $\xi$ contains singleton context modalities.
By the induction hypothesis, one can construct a formula $\xi'$ which is a Boolean combination of simple singleton-context formulas and relativized propositions
such that $\ctx{\vartrace}\xi$ and $\ctx{\vartrace}\xi'$ are equivalent. Let $\Lambda$ be the set consisting of the relativized atomic propositions $\Rel{p}{\vartraceAux}$ occurring in $\xi'$ with $\vartraceAux\neq \vartrace$ which are not in scope of a context modality, and the sub-formulas of $\xi'$ of the form $\ctx{\vartraceAux}\theta$ such that $\vartraceAux\neq \vartrace$ (note that by hypothesis $\ctx{\vartraceAux}\theta$ is a simple singleton-context formula).
Then, the formula $\varphi$ satisfying Proposition~\ref{prop:ExtesionofSimpleGHLTL} is given by
\[
\varphi\DefinedAs\bigvee_{\HMap\in G_\Lambda}(\xi'(\HMap)\wedge \bigwedge_{\theta\in \{\theta\in \Lambda\mid \HMap(\theta)=\top\}} \theta \wedge \bigwedge_{\theta\in \{\theta\in \Lambda\mid \HMap(\theta)=\neg\top\}}\neg\theta)
\]
where $\xi'(\HMap)$ is obtained from $\xi'$ by replacing all occurrences of the sub-formulas $\theta\in\Lambda$ in $\xi'$ with $\HMap(\theta)$, and by removing the singleton context modality $\ctx{\vartrace}$. Correctness
of the construction directly follows from the induction hypothesis and the semantics of $\GHLTL$.
\end{proof}

\subsection{Proof of Theorem~\ref{theorem:EmbeddingKnowledge}}\label{app:EmbeddingKnowledge}

\theoEmbeddingKnowledge*
\begin{proof}
We focus on the construction of the $\GHLTL$ sentence $\varphi$ capturing the $\KLTL$ formula
$\psi$ under the asynchronous semantics w.r.t.~the given observation map $\Obs$. The construction of the
 $\SGHLTL{\emptyset}$ sentence $\varphi_\emptyset$ capturing
$\psi$ under the synchronous semantics w.r.t.~$\Obs$ is similar. We inductively define a mapping $\TMap_\Obs$ assigning to each pair $(\phi,\vartrace)$ consisting of a $\KLTL$ formula $\phi$ and a trace variable $\vartrace\in\Var$, a $\GHLTL$ formula $\TMap_\Obs(\phi,\vartrace)$. Intuitively,
$\vartrace$ is associated to the current evaluated pointed trace. The mapping $\TMap_\Obs$ is homomorphic w.r.t.~the Boolean connectives and the $\LTL$ temporal modalities
(i.e., $\TMap(\Op\phi,\vartrace)\DefinedAs\Op \TMap(\phi,\vartrace)$ for each $\Op\in\{\neg,\Next\}$ and $\TMap(\phi_1\,\Op\,\phi_2,\vartrace)\DefinedAs\TMap(\phi_1,\vartrace)\,\Op\, \TMap(\phi_2,\vartrace)$ for each $\Op\in\{\vee,\Until\}$) and is defined as follows when the given $\KLTL$ formula is an atomic proposition or its root operator is a knowledge modality:
\begin{itemize}
\item $\TMap(p,\vartrace)\DefinedAs \Rel{p}{\vartrace}$;
\item $\TMap(\know_a\phi,\vartrace)\DefinedAs \forall^\Pt\vartraceAux.\, (\theta_{\Obs}(a,\vartrace,\vartraceAux) \rightarrow \TMap_\Obs(\phi,\vartraceAux)) $ where
$\vartraceAux\neq \vartrace$ and\\
$
\theta_{\Obs}(a,\vartrace,\vartraceAux)\DefinedAs \bigwedge_{p\in \Obs(a)}\Historically_{\Obs(a)}(\Rel{p}{\vartrace} \leftrightarrow \Rel{p}{\vartraceAux})\,\wedge\,
\Once_{\Obs(a)}(\ctx{\vartrace}\neg \Yesterday\top \wedge \ctx{\vartraceAux}\neg \Yesterday\top)
$
\end{itemize}

\noindent Intuitively, $\theta_{\Obs}(a,\vartrace,\vartraceAux)$ asserts that for the pointed traces $(\trace,i)$ and $(\trace',i')$ bound to the trace variables
$\vartrace$ and $\vartraceAux$, respectively, $\trace[0,i]$ and $\trace[0,i']$ are asynchronously equivalent for agent $a$ w.r.t.~the given observation map $\Obs$.
Note that $\TMap_{\Obs}(\phi,\vartrace)$ is a simple $\GHLTL$ formula if $\phi$ is in the one-agent fragment of $\KLTL$.
Let $\Lang$ be a set of traces, $\vartrace\in\Var$, $(\trace,i)$ be a pointed trace over $\Lang$, and $\TracesMap$ be a trace assignment over $\Lang$ such that
$\vartrace\in\Dom(\TracesMap)$ and $\TracesMap(\vartrace)= (\trace,i)$. By a straightforward induction on the structure of a $\KLTL$ formula $\varphi$, one can show that
$(\trace,i)\models_{\Lang,\Obs}\phi$ if and only if $(\TracesMap,\Var) \models_\Lang\TMap_{\Obs}(\phi,\vartrace)$. Hence, the desired $\GHLTL$ sentence $\varphi$
is given by $\forall \vartrace.\,\TMap_{\Obs}(\psi,\vartrace)$ and we are done.
\end{proof}

\section{Proofs from Section~\ref{sec:modelCheckingFragment}}\label{app:modelCheckingFragment}

\subsection{Automata for QPTL and Upper bounds of Theorem~\ref{theorem:QPTLsatisfiability}}\label{app:QPTLsatisfiability}

\subsubsection{Automata for QPTL}\label{app:QPTLAutomata}

In this section, we recall the classes of automata exploited in~\cite{BozzelliMP15} for solving satisfiability of full $\QPTL$ with past.
In the  standard automata-theoretic approach for $\QPTL$ formulas $\varphi$ in prenex  form~\cite{SistlaVW87}, first, one converts the quantifier-free part $\psi$ of $\varphi$ into an equivalent nondeterministic B\"{u}chi  automaton ($\NBA$) accepting the set   $\Lang(\psi)$ of traces satisfying $\psi$. Then, by using
the closure of  $\NBA$ definable languages under projection and complementation, one obtains an  $\NBA$ accepting $\Lang(\varphi)$.
This approach would not work for arbitrary $\QPTL$ formulas $\varphi$, where quantifiers can occur in the scope of both past and future temporal modalities. In this case, for a   subformula $\varphi'$ of $\varphi$, we need to keep track of the full set $\Lang_\wp(\varphi')$ of pointed traces  satisfying $\varphi$, and not simply $\Lang(\varphi')$. Thus, the automata-theoretic approach proposed in~\cite{BozzelliMP15}
is based on
a compositional translation of $\QPTL$ formulas into a simple two-way extension of  $\NBA$, called
 \emph{simple two-way  $\NBA$} ($\SNBA$, for short) accepting languages of \emph{pointed traces}.
Moreover, each step of the translation into  $\SNBA$ is based on an intermediate formalism consisting in  a two-way extension  of the class of (one-way) \emph{hesitant alternating automata} ($\HAA$, for short) over infinite words introduced in~\cite{KupfermanVW00}.
We now recall these two classes of automata. Moreover, for the class of two-way $\HAA$, we consider additional requirements which may be fulfilled by some states of the automaton. These extra requirements allow to obtain
a more fine-grained complexity analysis in the translation of two-way $\HAA$ into equivalent    $\SNBA$, and they are important for obtaining an asymptotically optimal automata-theoretic approach for $\QPTL$ in terms of the strong alternation depth of a $\QPTL$ formula.\vspace{0.2cm}

\noindent \textbf{$\SNBA$~\cite{BozzelliMP15}.} An  $\SNBA$ over traces on $\AP$
is a tuple $\Au=\mktuple{Q,Q_0,\TransA,\FStates_-,\FStates_+}$, where  $Q$ is a finite set of states, $Q_0\subseteq Q$ is a set of initial states, $\TransA:Q\times \{\rightarrow,\leftarrow\}\times 2^{\AP}\rightarrow 2^Q$ is a transition function, and $\FStates_-$ and $\FStates_+$ are sets  of accepting states.
Intuitively, the symbols $\rightarrow$
and $\leftarrow$ are used to denote forward and backward  moves.
 A run of $\Au$
over a pointed trace $(\trace,i)$ is a pair $r=(r_\leftarrow,r_\rightarrow)$ such that $r_\rightarrow=q_i,q_{i+1}\ldots$ is an infinite sequence of states, $r_\leftarrow=q'_i,q'_{i-1}\ldots q'_0 q'_{-1}$ is a finite sequence of states, and: (i) $q_i=q'_i\in Q_0$;
(ii) for each  $h\geq i$, $q_{h+1}\in \TransA(q_h,\rightarrow,\trace(h))$; and (iii)
 for each $h\in [0,i]$, $q'_{h-1}\in \TransA(q'_{h},\leftarrow,\trace(h))$.

Thus, starting from the initial position $i$ in the input pointed trace $(\trace,i)$, the automaton splits in two copies: the first one moves forwardly along the suffix of $\trace$ starting from position $i$ and the second one moves backwardly along the prefix $\trace(0)\ldots \trace(i)$.
The run $r=(r_\leftarrow,r_\rightarrow)$ is \emph{accepting} if $q'_{-1}\in \FStates_-$ and $r_\rightarrow$ visits infinitely often some state in $\FStates_+$. A pointed trace $(\trace,i)$ is accepted by $\Au$ if there is an accepting run of $\Au$ over $(\trace,i)$. We denote by $\Lang_\wp(\Au)$ the set of pointed traces accepted by $\Au$ and by $\Lang(\Au)$ the set of traces $\trace$ such that $(\trace,0)\in\Lang_\wp(\Au)$.\vspace{0.2cm}

 \noindent \textbf{Two-way $\HAA$~\cite{BozzelliMP15}.}
Now, we recall the class of two-way $\HAA$.
For a set $X$, $\PosBool(X)$ denotes the set of positive Boolean
formulas over $X$ built from elements in $X$ using $\vee$ and $\wedge$   (we also allow the formulas $\true$ and $\false$).
 For a formula
$\theta\in\PosBool(X)$, a \emph{model} $Y$ of $\theta$ is a subset $Y$
of $X$ which satisfies $\theta$. The model $Y$ of $\theta$ is \emph{minimal}
if no strict subset of $Y$ satisfies $\theta$.

A two-way $\HAA$ $\Au$  over traces on $\AP$
 is a tuple $\Au=\mktuple{ Q,q_0,\TransA, \FStates_-,\FamStratum}$, where $Q$ is a
finite set of states, $q_0\in Q$ is the initial state,  $\TransA:
Q\times 2^{\AP} \rightarrow \PosBool(\{\rightarrow,\leftarrow\}\times Q)$ is a transition
function, $\FStates_-\subseteq Q$ is the \emph{backward acceptance condition}, and $\FamStratum$ is a \emph{strata family} encoding a particular kind of parity acceptance condition and imposing some syntactical constraints on the transition function $\TransA$.
Before defining  $\FamStratum$,
 we give the notion of run which is independent of $\FamStratum$ and $\FStates_-$. We restrict ourselves to
  \emph{memoryless} runs, in which the behavior of the automaton
  depends only on the current input position and current state. Since
  later we will deal only with parity acceptance conditions,
  memoryless runs are sufficient (see e.g.~\cite{Zielonka98}).
Formally, given a pointed trace $(\trace,i)$
and a state $q'\in Q$, an \emph{$(i,q')$-run of $\Au$ over $\trace$} is a
directed graph $\mktuple{V,\Trans,v_0}$ with set of vertices $V\subseteq
(\nat\cup\{-1\})\times Q$ and initial vertex $v_0=(i,q')$. Intuitively, a vertex $(j,q)$ with $j\geq 0$ describes a copy of the automaton which is in state $q$ and reads the $j^{th}$ input position. Additionally, we require that the set of edges $\Trans$ is consistent with the transition function $\TransA$. Formally,
 for every vertex $v=(j,q)\in V$ such that $j\geq 0$, there is a \emph{minimal} model
$\{(\dir_1,q_1),\ldots,(\dir_n,q_n)\}$ of $\TransA(q,\trace(j))$ such that the set of successors of $v=(j,q)$ is $\{(j_1,q_1),\ldots,(j_n,q_n)\}$, where for all $k\in [1,n]$, $j_k=j+1$ if $\dir_k= \rightarrow$, and $j_k=j-1$ otherwise.

\noindent An infinite path $\pi$ of a run is \emph{eventually strictly-forward} whenever $\pi$ has
a suffix of the form $(i,q_1),(i+1,q_2),\ldots$ for some $i\geq 0$.

Now, we formally define $\FamStratum$ and give the semantic notion of acceptance.
 $\FamStratum$ is a \emph{strata family} of the
form
$\FamStratum=\{\mktuple{\rho_1,Q_1,\FStates_1},\ldots,\mktuple{\rho_k,Q_k,\FStates_k}\}$,
where $Q_1,\ldots,Q_k$ is a partition of the set of states $Q$ of $\Au$, and for all $i\in [1,k]$,
$\rho_i\in \{-,\tran,\Bu,\Co\}$ and $\FStates_i\subseteq Q_i$, such that
$\FStates_i=\emptyset$ whenever $\rho_i\in\{\tran,-\}$. A stratum
$\mktuple{\rho_i,Q_i,\FStates_i}$ is called a \emph{negative} stratum if $\rho_i=-$, a
\emph{transient} stratum if $\rho_i=\tran$, a B\"{u}chi 
stratum
(with B\"{u}chi acceptance condition $\FStates_i$) if $\rho_i=\Bu$, and a
coB\"{u}chi 
stratum (with coB\"{u}chi acceptance condition
$\FStates_i$) if $\rho_i=\Co$.  Additionally, 
there is a
partial order $\leq$ on the sets $Q_1,\ldots,Q_k$ such that the
following holds:
\begin{itemize}
\item \emph{Partial order requirement.} Moves from states in $Q_i$ lead to states in
  components $Q_j$ such that
  $Q_j\leq Q_i$;  additionally, if $Q_i$ belongs to a transient stratum, there are no
  moves from $Q_i$ leading to $Q_i$.
\item  \emph{Eventually syntactical requirement.} For all moves $(\dir,q')$ from states  $q\in Q_i$ such that $q'\in Q_i$ as well, the following holds:
 $\dir$ is $\leftarrow$ if the
    stratum of $Q_i$   is negative, and $\dir$ is
    $\rightarrow$ otherwise.

  \item Each component $Q_i$ satisfies one of the following two requirements:
    \begin{itemize}
  \item \emph{Existential requirement}: for all  states $q_i\in Q_i$ and input symbols $a$, if    $\TransA(q,a)$ is rewritten in disjunctive normal form, then each conjunction contains at most one occurrence of a state in $Q_i$.
 \item  \emph{Universal requirement}: for all  states $q_i\in Q_i$ and input symbols $a$, if    $\TransA(q,a)$ is rewritten in conjunctive normal form, then each disjunction contains at most one occurrence of a state in $Q_i$.
 \end{itemize}
\item Components of B\"{u}chi strata satisfy the existential requirement, while components of coB\"{u}chi strata satisfy the universal requirement.
\end{itemize}
The partial order requirement ensures
that every infinite path $\pi$ of a run gets trapped in the component
$Q_i$ of some  stratum.  The
\emph{eventually syntactical requirement}  ensures that $Q_i$
belongs to a B\"{u}chi or coB\"{u}chi stratum and that $\pi$ is
eventually strictly-forward. Moreover, the existential requirement for a component $Q_i$ establishes that
from each state $q\in Q_i$, at most one copy of the automaton is sent to the next input
symbol in component $Q_i$  (all the other copies move to states in strata with order lower than $Q_i$). Finally, the
 universal requirement corresponds to the dual of the existential requirement.
Note that transient components trivially satisfy both the existential and the universal requirement. Is is worth noting that the existential and universal requirements are not considered in~\cite{BozzelliMP15}. On the other hand, they are crucial in the automata-theoretic approach for $\QPTL$. 

Now we define when a run is accepting.  Let $\pi$ be an infinite path
 of  a  run,  $\mktuple{\rho_i,Q_i,\FStates_i}$ be
the B\"{u}chi or coB\"{u}chi stratum in which $\pi$ gets trapped, and
 $\Inf(\pi)$ be the states from $Q$ that occur infinitely many
times in $\pi$. The path $\pi$ is \emph{accepting} whenever
$\Inf(\pi)\cap \FStates_i\neq \emptyset$ if $\rho_i=\Bu$ and $\Inf(\pi)\cap
\FStates_i= \emptyset$ otherwise (i.e. $\pi$ satisfies the corresponding
B\"{u}chi or coB\"{u}chi requirement).
 A run is
\emph{accepting} if: (i) all its infinite paths are accepting and (ii) for each
vertex $(-1,q)$ reachable from the initial vertex, it holds that $q\in \FStates_-$ (recall that $\FStates_-$ is the backward acceptance condition of $\Au$).
 The
\emph{$\omega$-pointed language} $\Lang_{p}(\Au)$ of $\Au$ is the set
of pointed traces $(\trace,i)$  such that there is an
accepting $(i,q_0)$-run of $\Au$ on $\trace$.

The \emph{dual automaton} $\dual{\Au}$ of a two-way $\HAA$ $\Au=\mktuple{Q,q_0,\TransA,F_-,\FamStratum}$ is defined
as $\dual{\Au}=\mktuple{Q,q_0,\dual{\TransA},Q\setminus \FStates_-,\dual{\FamStratum}}$, where
$\dual{\TransA}(q,a)$ is the dual formula of $\TransA(q,a)$ (obtained from $\TransA(q,a)$ by
 switching $\vee$ and $\wedge$, and switching $\true$ and $\false$),
and $\dual{\FamStratum}$ is obtained from $\FamStratum$ by converting
a B\"{u}chi stratum $\mktuple{\Bu,Q_i,\FStates_i}$ into the coB\"{u}chi stratum
$\mktuple{\Co,Q_i,\FStates_i}$ and a coB\"{u}chi stratum $\mktuple{\Co,Q_i,\FStates_i}$ into
the B\"{u}chi stratum $\mktuple{\Bu,Q_i,\FStates_i}$.  By construction the dual automaton $\dual{\Au}$ of $\Au$ is
still a two-way $\HAA$. Following standard arguments (see
e.g.~\cite{Zielonka98}),  the following holds.

\begin{proposition}[\cite{BozzelliMP15}]\label{prop:Dualization}
  The dual automaton $\dual{\Au}$ of a two-way $\HAA$ $\Au$ is a two-way $\HAA$ accepting the complement of
 $\Lang_\wp(\Au)$.
\end{proposition}

Note that  $\SNBA$ correspond to a subclass of two-way $\HAA$.

\begin{proposition}[\cite{BozzelliMP15}]\label{remark:FromSNWAtoHAA} An  $\SNBA$ $\Au$ can be  converted ``on the fly'' in linear time into a two-way $\HAA$ accepting $\Lang_\wp(\Au)$.
\end{proposition}

 \noindent \textbf{Additional requirements on components of two-way $\HAA$.} Let   $\Au=\mktuple{Q,q_0,\TransA,F_-,\FamStratum}$ be a two-way \HAA.
 A component $Q_i$ of $\Au$ is \emph{globally nondeterministic} if the following  inductively holds:
 \begin{itemize}
   \item $Q_i$ satisfies the existential requirement and $Q_i$ is not a coB\"{u}chi component;
   \item for each component $Q_k$ distinct from $Q_i$ such that there are moves from $Q_k$ leading to $Q_i$ (hence, $Q_k>Q_i$), (i) $Q_k$ is \emph{globally nondeterministic}, and (ii) for each $q\in Q_k$ and input symbol $a$, if
       $\TransA(q,a)$ is rewritten in disjunctive normal form, then each conjunction contains at most one occurrence of a state in $Q_i$.
     \end{itemize}
     The previous requirement ensures that in each run $r$ and for each input position $\ell$, there is at most one copy of the automaton that is in some state
of $Q_i$ reading position $\ell$. A state of $\Au$ is \emph{globally nondeterministic} if it belongs to some globally nondeterministic component of $\Au$.\vspace{0.2cm}

 \noindent \textbf{From two-way $\HAA$ to $\SNBA$~\cite{BozzelliMP15}.} Two-way $\HAA$ can be translated in singly exponential time
 into equivalent $\SNBA$~\cite{BozzelliMP15}. The translation in~\cite{BozzelliMP15}  is a  generalization of the Miyano-Hayashi
construction~\cite{MiyanoH84} (for converting by an exponential blowup a B\"{u}chi alternating automaton into an equivalent
 $\NBA$)  and exploits the so  called \emph{ranking} construction~\cite{KupfermanV01} for converting in quadratic time
 a coB\"{u}chi alternating automaton  into an  equivalent B\"{u}chi alternating automaton. The translation can be easily adapted for handling
 globally nondeterministic components in two-way $\HAA$. Thus, we obtain the following result.

  \begin{proposition}[\cite{BozzelliMP15}]\label{prop:FromHAAtoSNBA}
  Given a two-way $\HAA$ $\Au$,  one can construct ``on the fly'' and in singly
  exponential time a B\"{u}chi $\SNBA$  accepting $\Lang_\wp(\Au)$ with
  $2^{O(n\cdot \log(n+ k))}$ states, where $k$ is the number of states
  in the globally nondeterministic components of $\Au$,  and $n$ is the number of remaining states of $\Au$.
\end{proposition}

\subsubsection{Upper bounds of Theorem~\ref{theorem:QPTLsatisfiability}}

In this section, in order to provide the upper bounds of Theorem~\ref{theorem:QPTLsatisfiability}, we first illustrate  a compositional translation of a $\QPTL$ formula   into an equivalent $\SNBA$ which improves the one given in~\cite{BozzelliMP15}. There, occurrences of temporal modalities immediately preceding propositional quantification always count as additional alternations, and so they lead to an additional exponential blowup in the compositional translation.
 The relevant case in the translation is when the outermost operator of the currently processed sub-formula $\varphi$ (assumed to be in negation normal form)  is a temporal modality (the other cases easily follow from the  closure of $\SNBA$-definable pointed languages under union, intersection, and projection). This case is handled by first building a two-way $\HAA$ $\Au$ accepting $\Lang_\wp(\varphi)$ and then by applying Proposition~\ref{prop:FromHAAtoSNBA}.  The construction of  $\Au$ is obtained by a generalization of the standard linear-time translation of $\LTL$ formulas into B\"{u}chi alternating automata which exploits the (inductively built)  $\SNBA$ associated with the maximal universal and existential sub-formulas of $\varphi$.
 Formally, we establish the following result, where
for a $\QPTL$    formula $\varphi$, we say that $\varphi$  \emph{is of universal-type} if there is a universal sub-formula
$\psi$ of the negation normal form of $\varphi$ such that
$\sad( \psi)=\sad(\varphi)$; otherwise, we say that $\varphi$ \emph{is of existential-type}.

\begin{theorem}\label{theorem:FromQPTLtoSNBA}
Let $\varphi$ be a  $\QPTL$ formula of existential-type (resp., of universal-type)  and $h=\sad(\varphi)$. Then, one can construct ``on the fly'' an  $\SNBA$ $\Au_\varphi$ accepting $\Lang_\wp(\varphi)$   in time $\Tower(h+1,O(|\varphi|))$ (resp., $\Tower(h+2,O(|\varphi|))$).
\end{theorem}
\begin{proof} We assume without loss of generality that $\varphi$ is in negation normal form.
The proof is given by induction on the structure of $\varphi$.  The base case  is trivial. For the induction step, we distinguish four cases depending on the type of root operator of $\varphi$ (either positive  boolean connective, or existential quantifier, or universal quantifier, or temporal modality).\vspace{0.2cm}

 \noindent \textbf{Case 1:}  $\varphi$ is of the form $\varphi=\varphi_1\wedge\varphi_2$ or $\varphi=\varphi_1\vee \varphi_2$. Assume that
  $\varphi=\varphi_1\wedge\varphi_2$ (the other case being similar).
  We use the fact that like B\"{u}chi nondeterministic automata, $\SNBA$ are trivially and efficiently  closed under intersection. In particular, given
two   $\SNBA$ $\Au_1$ and $\Au_2$, one can construct ``on the fly'' and in time $O(|\Au_1||\Au_2|)$ an $\SNBA$
accepting the language $\Lang_\wp(\Au_1)\cap \Lang_\wp(\Au_2)$.
  Since $\sad(\varphi)=\max(\sad(\varphi_1),\sad(\varphi_2))$,  the result easily follows from the induction hypothesis. \vspace{0.2cm}

\noindent \textbf{Case 2:}  $\varphi$ is an existential formula of the form $\varphi=\exists p.\,\varphi'$.
Hence, in particular, $\varphi$ is  of existential-type.
Let $h=\sad(\varphi)$ and $h'=\sad(\varphi')$.
We observe that like B\"{u}chi nondeterministic automata, $\SNBA$ are  efficiently  closed under projection. In particular, given
an  $\SNBA$ $\Au$ over traces and $p\in \AP$, one can construct ``on the fly'' and in linear time an  $\SNBA$
accepting the pointed language $\{(\trace,i)  \mid \text{ there is } (\trace',i)\in\Lang_\wp(\Au) \text{ such that }\trace'=_{\AP\setminus \{p\}}\trace\}$.
Thus, by applying the induction hypothesis, it follows that one can construct ``on the fly'' an  $\SNBA$ accepting $\Lang_\wp(\varphi)$ of size
$\Tower(h'+1,O(|\varphi'|))$ if $\varphi'$ is of existential-type, and of size $\Tower(h'+2,O(|\varphi'|))$ otherwise.
Since $h'= h$ if $\varphi'$ is of existential-type, and $h'=h-1$ otherwise (\ie $\varphi'$ is of universal-type), the result follows.\vspace{0.2cm}

\noindent \textbf{Case 3:}  $\varphi$ is a  universal formula of the form $\varphi=\forall p.\,\varphi'$.
Hence, in particular, $\varphi$ is of universal-type.
Let $h=\sad(\varphi)$ and $\dual{\varphi}\,'$ be the negation normal form of $\neg\varphi'$.
We have that $\sad( \exists p.\, \widetilde{\varphi}\,')=h$ and $\Lang_\wp(\exists p.\, \widetilde{\varphi}\,')=\Lang_\wp(\neg\varphi)$. Hence, by Case~2, one can construct ``on the fly'' an  $\SNBA$ $\Au_{\neg\varphi}$
  of size $\Tower(h+1,O(|\varphi|))$ accepting $\Lang_\wp(\neg\varphi)$. By Propositions~\ref{prop:Dualization}, \ref{remark:FromSNWAtoHAA}  and~\ref{prop:FromHAAtoSNBA},  it follows that one
  can construct ``on the fly'' an  $\SNBA$ $\Au_{\varphi}$
  of size $\Tower(h+2,O(|\varphi|))$ accepting $\Lang_\wp(\varphi)$. Hence, the result follows.\vspace{0.2cm}

\noindent \textbf{Case 4:}  the root operator of $\varphi$ is a temporal modality.
Let $h=\sad(\varphi)$ and
 $P$ be the set of existential and universal sub-formulas of $\varphi$  which do not occur in the scope of a quantifier. If $P=\emptyset$, then
$\varphi$ is a $\PLTL$ formula and $h=0$. In this case, by a straightforward adaptation of the standard translation of $\LTL$ into B\"{u}chi word automata~\cite{Var88},  one can construct a  $\SNBA$  of size  $2^{O(|\varphi|)}$ accepting $\Lang_\wp(\varphi)$. Hence, the result follows.

Assume now that $P\neq \emptyset$. Then, $\varphi$ can be viewed as a $\PLTL$ formula in negation normal form, written $\PLTL(\varphi)$, over the set of atomic propositions given by $P\cup\AP$.\vspace{0.1cm}

\noindent \textbf{Subcase 4.1:} we first assume that for each $\psi\in P$, $\sad(\psi)<\sad(\varphi)$. Hence, for all $\psi\in P$,
$\sad(\psi)\leq h-1$ and $h\geq 1$. Note that
$\varphi$ is of existential-type.  Moreover, for each universal formula $\forall p.\, \xi\in P$, we have that
$\sad(\forall p.\,\xi)=\sad(\neg\exists p.\,\widetilde{\xi})$ and
$\Lang_\wp(\forall p.\,\xi)=\Lang_\wp(\neg\exists p.\,\widetilde{\xi})$, where $\widetilde{\xi}$ is the negation normal form
of $\xi$. Thus, since each existential formula is of existential-type,
 by applying the induction hypothesis, Proposition~\ref{remark:FromSNWAtoHAA} and the complementation result for two-way $\HAA$ (see Proposition~\ref{prop:Dualization}), it follows that for each $\psi\in P$, one can construct ``on the fly'' in time at most $\Tower(h,O(|\psi|))$, a two-way $\HAA$ $\Au_\psi$ accepting $\Lang_\wp(\psi)$. Then, by an easy generalization of the standard linear-time translation of $\LTL$ formulas into B\"{u}chi alternating word automata and by using the two-way $\HAA$ $\Au_\psi$ with $\psi\in P$, we show that one can construct ``on the fly'',
 in time $\Tower(h,O(|\varphi|))$, a two-way $\HAA$ $\Au_\varphi$ accepting $\Lang_\wp(\varphi)$.
 Intuitively, given an input pointed trace, each copy of  $\Au_\varphi$ keeps track of the current subformula of $\PLTL(\varphi)$ which needs to be evaluated. The evaluation simulates the semantics of $\PLTL$ 
 by using universal and existential branching, but when the current subformula $\psi$ is in $P$, then the current copy of $\Au_\varphi$ activates a copy of $\Au_\psi$ in the initial state.

 Formally, for each $\psi\in P$, let $\Au_\psi=\mktuple{Q_\psi,q_\psi,\TransA_\psi, \FStates_\psi^{-}, \FamStratum_\psi}$. Without loss of generality, we assume that the state sets of the two-way $\Au_\psi$ are pairwise distinct. Then, $\Au_\varphi=\mktuple{Q,q_0,\TransA,\FStates_-,\FamStratum}$, where:
 \begin{itemize}
   \item $Q = \displaystyle{\bigcup_{\psi\in P}Q_\psi}\cup \Sub(\varphi)$, where $\Sub(\varphi)$ is the set of subformulas of $\PLTL(\varphi)$;
   \item $q_0=\varphi$;
   \item The transition function $\TransA$ is defined as follows: $\TransA(q,a) = \TransA_\psi(q,a)$ if $q\in Q_\psi$ for some $\psi\in P$.
   If instead  $ q\in \Sub(\varphi)$, then $\TransA(q,a)$ is inductively defined as follows:
   \begin{itemize}
     \item $\TransA(p,a) = \true$ if $p\in a$, and $\TransA(p,a) = \false$ otherwise (for all $p\in\AP\cap \Sub(\varphi)$);
     \item $\TransA(\neg p,a) = \false$ if $p\in a$, and $\TransA(\neg p,a) = \true$ otherwise (for all $p\in\AP\cap \Sub(\varphi)$);
     \item $\TransA(\phi_1\wedge \phi_2,a)=\TransA(\phi_1,a)\wedge\TransA(\phi_2,a)$ and $\TransA(\phi_1\vee \phi_2,a)=\TransA(\phi_1,a)\vee\TransA(\phi_2,a)$;
     \item $\TransA(\Next\phi,a)=(\rightarrow,\phi)$ and $\TransA(\Yesterday\phi,a)=(\leftarrow,\phi)$;
      \item $\TransA(\phi_1\Until \phi_2,a)=\TransA(\phi_2,a)\vee (\TransA(\phi_1,a)\wedge (\rightarrow,\phi_1\Until \phi_2))$;
       \item $\TransA(\phi_1\Since \phi_2,a)=\TransA(\phi_2,a)\vee (\TransA(\phi_1,a)\wedge (\leftarrow,\phi_1\Since \phi_2))$;
       \item $\TransA(\phi_1\Release \phi_2,a)=\TransA(\phi_2,a)\wedge (\TransA(\phi_1,a)\vee (\rightarrow,\phi_1\Release \phi_2))$;
       \item $\TransA(\phi_1\PastRelease \phi_2,a)=\TransA(\phi_2,a)\wedge (\TransA(\phi_1,a)\vee (\leftarrow,\phi_1\PastRelease \phi_2))$;
      \item for each $\psi\in P$, $\TransA(\psi,a)=\TransA_\psi(q_\psi, \sigma)$.
   \end{itemize}
   \item $\FStates_-=\displaystyle{\bigcup_{\psi\in P}}\FStates_\psi^{-}$
   \item $\FamStratum = \displaystyle{\bigcup_{\psi\in P}\FamStratum_\psi} \cup \bigcup_{\phi\in \Sub(\varphi)}\{\mathcal{S}_\phi$\}, where for each
   $\phi\in \Sub(\varphi)$, $\mathcal{S}_\phi$ is defined as follows:
   \begin{itemize}
     \item if $\phi$ has as root a past temporal modality, then $\mathcal{S}_\phi$ is the negative stratum $(-,\{\phi\},\emptyset)$;
     \item if $\phi$ has as root the (future) until modality, then $\mathcal{S}_\phi$ is the B\"{u}chi stratum $(\Bu,\{\phi\},\emptyset)$;
          \item if $\phi$ has as root the (future) release modality, then $\mathcal{S}_\phi$ is the coB\"{u}chi stratum  $(\Co,\{\phi\},\emptyset)$;
     \item otherwise, $\mathcal{S}_\phi$ is the transient stratum given by $(\tran,\{\phi\},\emptyset)$.
   \end{itemize}
 \end{itemize}
 Finally, since $h\geq 1$ and the size of  the two-way $\HAA$ $\Au_\varphi$ is $\Tower(h,O(|\varphi|))$,
by applying Proposition~\ref{prop:FromHAAtoSNBA}, one can construct ``on the fly'' an   $\SNBA$ accepting $\Lang_\wp(\varphi)$ of size
$\Tower(h+1,O(|\varphi|))$. Hence, the result follows.\vspace{0.2cm}

\noindent \textbf{Subcase 4.2:} for some $\psi\in P$, $\sad(\psi)= \sad(\varphi)$, and either $\varphi=\Next\varphi_1$ or $\varphi=\Yesterday\varphi_2$.
This case easily follows from the induction hypothesis and the fact that given an $\SNBA$ $\Au$ one can easily construct in linear time
two $\SNBA$ $\Au_\Next$ and $\Au_\Yesterday$ such that $\Lang_\wp(\Au_\Next)=\{(\trace,i)\mid (\trace,i+1)\in \Lang_\wp(\Au)\}$
and $\Lang_\wp(\Au_\Yesterday)=\{(\trace,i)\mid i>0 \text{ and } (\trace,i-1)\in \Lang_\wp(\Au)\}$.\vspace{0.2cm}

\noindent \textbf{Subcase 4.3:} for some $\psi\in P$, $\sad(\psi)= \sad(\varphi)=h$, and either $\varphi=\varphi_1\Until\varphi_2$ or $\varphi=\varphi_1\Since\varphi_2$.
Let $P_1$ be the set of universal and existential sub-formulas of $\varphi_1$. Since $\psi$ is a universal or existential sub-formula of $\varphi$ and
$\sad(\psi)= \sad(\varphi)$, by definition of the strong alternation depth, it follows that (i) $\varphi$ and $\varphi_2$ are of existential-type, (ii) $\sad(\varphi_2)=\sad(\varphi)$, and (iii) for each $\psi_1\in P_1$, $\sad(\psi_1)<h$. We focus on the case where $\varphi= \varphi_1\Until\varphi_2$ (the case where $\varphi=\varphi_1\Since\varphi_2$ is similar).  
By the induction hypothesis, one can construct an $\SNBA$ $\Au_2$ in time $\Tower(h+1,O(|\varphi|))$ accepting $\Lang_\wp(\varphi_2)$
 We distinguish two cases:
 \begin{itemize}
   \item $h=0$: hence, $P_1=\emptyset$ and $\varphi_1$ is a $\PLTL$ formula.
   By an easy adaptation of the standard translation of $\LTL$ into B\"{u}chi word automata~\cite{Var88} and by exploiting the $\SNBA$ $\Au_2$ for the formula $\varphi_2$,  one can construct an  $\SNBA$ $\Au_\varphi$  of size  $2^{O(|\varphi|)}$ accepting $\Lang_\wp(\varphi)$.
     Intuitively, given an input pointed trace  $(\trace,i)$, $\Au_\varphi$ guesses a position $j\geq i$ and checks that $(\trace,j)\in\Lang_\wp(\Au)$ and for all $k\in [i,j)$, $(\trace,k)\models \varphi_1$ as follows.
      Initially, $\Au_\varphi$ keeps track of both the guessed set $\Lambda_0$ of sub-formulas of $\varphi_1$ holding at the current position $i$, and the guessed state $q$ of $\Au_2$ which represents the state where the backward copy of $\Au_2$ would be on reading the
      $i^{th}$ position of $\trace$ in some guessed accepting run of $\Au_2$ over $(\trace,j)$. If $j=i$, then $q$ needs to be some initial state of $\Au_2$, and $\Au_2$ simply simulates the behavior of $\Au_2$ on $(\trace,i)$ and propagates the guesses about the sub-formulas of $\varphi_1$ in accordance to the semantics of $\PLTL$. Otherwise, $\Au_\varphi$ splits in two copies: the backward copy simulates the backward copy of $\Au_2$ and \emph{deterministically} checks that the initial guessed set $\Gamma_0$ of sub-formulas of $\varphi_1$ contains $\varphi_1$ and is consistent in the interval of positions $[0,i]$, while the forward copy of $\Au_\varphi$ behaves as follows. In the first step, the forward copy of $\Au_\varphi$ moves to the same state $(q,\Gamma_0)$, and after this step, such a copy starts to simulate  in forward-mode the backward copy of $\Au_2$ until, possibly, a `switch' occurs at the guessed position $j$, where the forward copy of $\Au_\varphi$ simulates in a unique step from the current state some initial split of $\Au_2$ in the backward and forward copy. In the phase before the switch, the current guessed set of sub-formulas always contains $\varphi_1$.
      After such a switch (if any), the forward copy of $\Au_\varphi$ simply simulates the forward copy of $\Au_2$ and propagates the guesses about the sub-formulas of $\varphi_1$ in accordance to the semantics of $\PLTL$. We use two flags to distinguish the different phases of the simulation (in particular, the initial phase and the switch phase). The acceptance condition is a  generalized B\"{u}chi condition which can be converted into a B\"{u}chi condition in a standard way.
   \item $h\geq 1$: in this case, we construct a two-way $\HAA$ $\Au_\varphi$ accepting $\Lang_\wp(\varphi)$ as done for the subcase~4.1 but we replace
   the set $P$ with the set $P_1\cup \{\varphi_2\}$. Note that for each $\psi\in P_1$, being $\sad(\psi)<h$, the two-way $\HAA$ associated with formula $\psi$ has size
   at most $\Tower(h,O(|\psi|))$. Moreover, let $\Au_{\varphi_2}$ be the two-way $\HAA$ associated with the $\SNBA$ $\Au_2$ accepting $\Lang_\wp(\varphi_2)$.
   $\Au_{\varphi_2}$ has size at most $\Tower(h+1,O(|\varphi_2|))$ and since it is the two-way $\HAA$ associated to an $\SNBA$, it has just two strata: a negative \emph{existential} stratum, say $\mathcal{S}_-$, and a B\"{u}chi stratum (a B\"{u}chi stratum is always existential), say $\mathcal{S}_{\Bu}$. Thus, applying the construction illustrated
   for the subcase~4.1, we have that the upper stratum $\mathcal{S}_{\varphi}$ of the two-way $\HAA$ $\Au_\varphi$ associated with the formula $\varphi=\varphi_1\Until\varphi_2$  is a
   B\"{u}chi stratum consisting of the single state $\varphi$ which is the initial state. Additionally, $\mathcal{S}_{\varphi}$ is the unique stratum from which it is possible to move to the strata of $\Au_{\varphi_2}$, and for all input symbols $a$, if $\TransA_\varphi(\varphi,a)$ is rewritten in disjunctive normal
   form, then each disjunction contains at most one state from $\mathcal{S}_{-}$ and at most one state from $\mathcal{S}_{\Bu}$. In other terms, the strata
   $\mathcal{S}_{-}$ and $\mathcal{S}_{\Bu}$ are globally nondeterministic in $\Au_\varphi$. Thus, $\Au_\varphi$ has at most  $\Tower(h+1,O(|\varphi_2|))$
   globally nondeterministic states while the number of remaining states is at most $\Tower(h,O(|\varphi_1|))$. By applying Proposition~\ref{prop:FromHAAtoSNBA}, one can construct ``on the fly'' an   $\SNBA$ accepting $\Lang_\wp(\varphi)$ of size
$\Tower(h+1,O(|\varphi|))$. Hence, the result follows.
 \end{itemize}

\noindent \textbf{Subcase 4.4:} for some $\psi\in P$, $\sad(\psi)= \sad(\varphi)$, and either $\varphi=\varphi_1\Release\varphi_2$ or $\varphi=\varphi_1\PastRelease\varphi_2$.  Since $\psi$ is a universal or existential sub-formula of $\varphi_1$ or $\varphi_2$, by definition
of $\sad(\varphi)$, $\psi$ must be a universal formula. Hence, $\varphi$ is of universal-type.
      Let $\dual{\varphi}\,'$ be the negation normal form of $\neg\varphi'$.
We have that $\sad(\dual{\varphi}\,')=h$ and $\Lang_\wp(\dual{\varphi}\,')=\Lang_\wp(\neg\varphi)$. Hence, by the subcase 4.3, one can construct ``on the fly'' an  $\SNBA$ $\Au_{\neg\varphi}$
  of size $\Tower(h+1,O(|\varphi|))$ accepting $\Lang_\wp(\neg\varphi)$. By Propositions~\ref{prop:Dualization}, \ref{remark:FromSNWAtoHAA}  and~\ref{prop:FromHAAtoSNBA}, it follows that one
  can construct ``on the fly'' an  $\SNBA$ $\Au_{\varphi}$
  of size $\Tower(h+2,O(|\varphi|))$ accepting $\Lang_\wp(\varphi)$. Hence, the result follows. 
\vspace{0.1cm}

\noindent This concludes the proof of Theorem~\ref{theorem:FromQPTLtoSNBA}.
\end{proof}

By exploiting Theorem~\ref{theorem:FromQPTLtoSNBA}, we can provide the upper bounds of Theorem~\ref{theorem:QPTLsatisfiability}.

\begin{theorem}\label{theorem:UpperBoundsQPTLsatisfiability}
For all $h\geq 0$, satisfiability of $\QPTL$ sentences $\varphi$ with strong alternation depth at most $h$  is in  $h$-\EXPSPACE.
\end{theorem}
\begin{proof}
We  observe that a $\QPTL$ sentence is satisfiable iff it is valid. Thus, since a $\QPTL$ sentence in negation normal form is a positive Boolean combination
of universal and existential sentences, it suffices to show the result for existential and universal $\QPTL$ sentences $\varphi$:
\begin{itemize}
  \item $\varphi=\exists p.\,\varphi'$: hence, $\varphi$ is of existential-type. By Theorem~\ref{theorem:FromQPTLtoSNBA}, one can construct ``on the fly'' an  $\SNBA$ $\Au_\varphi$ accepting $\Lang_\wp(\varphi)$   in time $\Tower(h+1,O(|\varphi|))$. We observe that an $\SNBA$ $\Au$ can be trivially converted into a B\"{u}chi nondeterminstic automaton accepting the set of traces $\trace$ such that $(\trace,0)\in \Lang_\wp(\Au)$. Thus, since  checking non-emptiness for B\"{u}chi nondeterministic automata is in \NLOGSPACE, the result follows.
  \item $\varphi=\forall p.\,\varphi'$: since $\varphi$ is a   sentence, we have that $\varphi$ is satisfiable iff  $\exists p.\,\neg\varphi$ is unsatisfiable.
  Thus, since $\sad(\varphi)=\sad(\exists p.\,\neg\varphi)$, this case reduces to the previous case.\qedhere
\end{itemize}
\end{proof}

\subsection{Detailed proof of Theorem~\ref{theorem:fromEmptyGammaFragmentToQPTL}}\label{app:fromEmptyGammaFragmentToQPTL}

For a $\QPTL$ formula $\varphi$ and $\AP'\subseteq \AP$ with $\AP' =\{p_1,\ldots, p_n\}$, we write $\exists \AP'.\varphi$ to mean
$\exists p_1.\ldots \exists p_n.\, \varphi$.
Given a fair Kripke structure $(\KS,\FStates)$, a \emph{$(\KS,\FStates)$-assignment} is a partial mapping $\TracesMap$ over $\Var$ assigning to each trace variable $\vartrace$ in its domain
$\Dom(\vartrace)$ a pair $(\Path,i)$ consisting of a $\FStates$-fair  path of $\KS$ and a position $i\geq 0$.

\theoFromEmptyGammaFragmentToQPTL*
 \begin{proof}
Let $\KS= \mktuple{\States,\States_0,\Trans,\Lab}$. The high-level description of the  reduction of model checking $(\KS,\FStates)$ against $\varphi$
  to $\QPTL$ satisfiability has been given in Section~\ref{sec:modelCheckingFragment}.  Here, we  provide the details of the reduction. We consider a new finite set $\AP'$ of atomic propositions defined as follows:
\[
\begin{array}{l}
\AP':= \bigcup_{\vartrace\in \Var}\AP_{\vartrace} \cup \States_{\vartrace} \cup \{\pad_{\overleftarrow{\vartrace}},\pad_{\overrightarrow{\vartrace}}\} \\
 \AP_{\vartrace} := \{p_{\overleftarrow{\vartrace}},p_{\overrightarrow{\vartrace}}\mid p\in\AP\} \text{ and }  \States_{\vartrace} :=\{\state_{\overleftarrow{\vartrace}},\state_{\overrightarrow{\vartrace}} \mid \state\in \States\}
\end{array}
\]
Let $\AP_{\overleftarrow{\vartrace}}:= \{p_{\overleftarrow{\vartrace}}\mid p\in\AP\}$, $\AP_{\overrightarrow{\vartrace}}:= \{p_{\overrightarrow{\vartrace}}\mid p\in\AP\}$,
$\States_{\overleftarrow{\vartrace}} :=\{\state_{\overleftarrow{\vartrace}} \mid \state\in \States\}$, and
$\States_{\overrightarrow{\vartrace}} :=\{\state_{\overrightarrow{\vartrace}} \mid \state\in \States\}$.
Thus, we associate to each variable $\vartrace\in \Var$ and atomic proposition $p\in \AP$, two fresh atomic propositions $p_{\overleftarrow{\vartrace}}$ and
$p_{\overrightarrow{\vartrace}}$, and to each variable $\vartrace\in \Var$ and state $\state$ of $\KS$, two fresh atomic proposition $\state_{\overleftarrow{\vartrace}}$ and $\state_{\overrightarrow{\vartrace}}$. Moreover, for each $\vartrace\in \Var$, $\pad_{\overleftarrow{\vartrace}}$ and $\pad_{\overrightarrow{\vartrace}}$ are exploited as padding propositions.\vspace{0.2cm}

\noindent \emph{Encoding of paths and pointed paths of $\KS$.} For each $\vartrace\in \Var$ and  an infinite word  $\Path=\state_0,\state_1,\ldots$ over $\States$, we denote by $\overrightarrow{\trace}(\vartrace,\Path)$   the trace over
$\AP_{\overrightarrow{\vartrace}}\cup \States_{\overrightarrow{\vartrace}}$ defined as follows for all $i\geq 0$:
\[
\overrightarrow{\trace}(\vartrace,\Path)(i):=\{(\state_i)_{\overrightarrow{\vartrace}}\}\cup \{p_{\overrightarrow{\vartrace}}\mid p\in \Lab(\state_i)\}
\]
For each finite word $\Path$ over $\States$, let $\overleftarrow{\trace}(\vartrace,\Path)$ be the trace over $\AP_{\overleftarrow{\vartrace}}\cup \States_{\overleftarrow{\vartrace}}$ given by $\{\pad_{\overleftarrow{\vartrace}}\}\cdot w\cdot \{\pad_{\overleftarrow{\vartrace}}\}^{\omega}$, where $|w|=|\Path|$ and for each $0\leq  i<|\Path|$,
$w(i)=\{(\state_i)_{\overleftarrow{\vartrace}}\}\cup \{p_{\overleftarrow{\vartrace}}\mid p\in \Lab(\state_i)\}$.
Now, we define the forward and backward encodings of a path $\Path$ of $\KS$. For each $k\in\nat$, the \emph{forward $\vartrace$-encoding of $\Path$ with offset $k$}  is the trace over $\AP_{\vartrace}\cup \States_{\vartrace}\cup\{\pad_{\overleftarrow{\vartrace}},\pad_{\overrightarrow{\vartrace}}\}$, denoted by $\overrightarrow{\trace}(\vartrace,\Path,k)$, defined as follows:
\begin{itemize}
  \item the projection of $\overrightarrow{\trace}(\vartrace,\Path,k)$ over $\AP_{\overrightarrow{\vartrace}}\cup \States_{\overrightarrow{\vartrace}}\cup\{\pad_{\overrightarrow{\vartrace}}\}$ is $\{\pad_{\overrightarrow{\vartrace}}\}^{k}\cdot \overrightarrow{\trace}(\vartrace,\Path)$;
  \item the projection of $\overrightarrow{\trace}(\vartrace,\Path,k)$ over $\AP_{\overleftarrow{\vartrace}}\cup \States_{\overleftarrow{\vartrace}}\cup\{\pad_{\overleftarrow{\vartrace}}\}$ is $\{\pad_{\overleftarrow{\vartrace}}\}^{\omega}$.
\end{itemize}

For each $k>0$,   the \emph{backward $\vartrace$-encoding of $\Path$ with offset $k>0$}, is the trace over $\AP_{\vartrace}\cup \States_{\vartrace}\cup\{\pad_{\overleftarrow{\vartrace}},\pad_{\overrightarrow{\vartrace}}\}$, denoted by $\overleftarrow{\trace}(\vartrace,\Path,k)$, defined as follows, where $\Path_{\leq k}^{R}$ denotes the reverse of the prefix $\Path(0)\ldots \Path(k-1)$ of $\Path$ until position $k-1$:
\begin{itemize}
  \item the projection of $\overrightarrow{\trace}(\vartrace,\Path,k)$ over $\AP_{\overrightarrow{\vartrace}}\cup \States_{\overrightarrow{\vartrace}}\cup\{\pad_{\overrightarrow{\vartrace}}\}$ is $ \overrightarrow{\trace}(\vartrace,\Path^{k})$;
  \item the projection of $\overrightarrow{\trace}(\vartrace,\Path,k)$ over $\AP_{\overleftarrow{\vartrace}}\cup \States_{\overleftarrow{\vartrace}}\cup\{\pad_{\overleftarrow{\vartrace}}\}$ is $ \overleftarrow{\trace}(\vartrace,\Path_{\leq k}^{R})$.
\end{itemize}



\noindent \emph{Claim 1.}
  For all $\vartrace \in\Var$, one can construct in linear time two $\PLTL$ formulas $\theta(\vartrace,\rightarrow)$ and $\theta(\vartrace,\leftarrow)$ such that
  for all  pointed traces $(\trace,i)$ over $\AP'$, we have:
 \begin{itemize}
   \item $(\trace,i)\models\theta(\vartrace,\rightarrow)$ \emph{iff} there is a  $\FStates$-fair path $\Path$ of $\KS$ such that the projection of $\trace$ over $\States_{\vartrace}\cup \AP_{\vartrace}\cup\{\pad_{\overleftarrow{\vartrace}},\pad_{\overrightarrow{\vartrace}}\}$ is the forward $\vartrace$-encoding of $\Path$ for some offset $k\geq 0$.
   \item    $(\trace,i)\models\theta(\vartrace,\leftarrow)$ \emph{iff} there is a  $\FStates$-fair path $\Path$ of $\KS$ such that the projection of $\trace$ over $\States_{\vartrace}\cup \AP_{\vartrace}\cup\{\pad_{\overleftarrow{\vartrace}},\pad_{\overrightarrow{\vartrace}}\}$ is the backward $\vartrace$-encoding of $\Path$ for some offset $k> 0$.
 \end{itemize}

We give the details on the construction of the $\PLTL$ formula $\theta(\vartrace,\leftarrow)$ in Claim~1 (the construction of $\theta(\vartrace,\rightarrow)$ is similar).
For each state $\state\in\States$,
  $\Trans(\state)$ denotes the set of successors of $\state$ in $\KS$, while $\Trans^{-1}(\state)$ denotes the set of predecessors of $\state$ in $\KS$.
  The  $\PLTL$ formula $\theta(\vartrace,\leftarrow)$ is defined as follows:
\begin{eqnarray*}
\theta(\vartrace,\overleftarrow{})&:=&  \Once \bigl ( (\neg \Yesterday \top) \,\wedge\, \pad_{\overleftarrow{\vartrace}} \,\wedge\, \Always \neg \pad_{\overrightarrow{\vartrace}}  \wedge \bigvee_{\state\in\States}[\xi(\vartrace,\state,\rightarrow)\wedge \bigvee_{\state'\in\Trans^{-1}(\state)} \Next\xi(\vartrace,\state',\leftarrow)]   \bigr)\\
 \xi(\vartrace,\state,\rightarrow) &:= &   \state_{\overrightarrow{\vartrace}} \,\wedge\, \bigvee_{\state'\in \FStates}\Always\Future \,\state'_{\overrightarrow{\vartrace}} \,\,\wedge\,\, \bigwedge_{\state'\in \States} \Always\Bigl(\state'_{\overrightarrow{\vartrace}} \rightarrow \\
 & & \Bigl[\bigwedge_{p\in \Lab(\state')}p_{\overrightarrow{\vartrace}} \wedge
   \bigwedge_{p\in \AP\setminus \Lab(\state')}\neg p_{\overrightarrow{\vartrace}} \wedge \bigwedge_{\state''\in \States\setminus\{\state'\}}\neg  \state''_{\overrightarrow{\vartrace}} \wedge \bigvee_{\state''\in \Trans(\state')}\Next\, \state''_{\overrightarrow{\vartrace}} \Bigr] \Bigr)\\
 \xi(\vartrace,\state,\leftarrow) &:= &   \state_{\overleftarrow{\vartrace}} \,\wedge\,  \Future\Always\pad_{\overleftarrow{\vartrace}} \,\wedge\, \bigvee_{\state'\in \States_0}\Future (\state'_{\overleftarrow{\vartrace}}\wedge
 \Next \pad_{\overleftarrow{\vartrace}}) \,\,\wedge\,\,  \Historically\Always [\pad_{\overleftarrow{\vartrace}} \,\rightarrow\, \bigwedge_{q\in \AP\cup \States}\neg q_{\overleftarrow{\vartrace}}]  \,\, \wedge\,\, \\
  & & \bigwedge_{\state'\in \States} \Always\Bigl( \state'_{\overleftarrow{\vartrace}}
   \longrightarrow \Bigl[\bigwedge_{p\in \Lab(\state')}p_{\overleftarrow{\vartrace}} \,\,\wedge
   \bigwedge_{p\in \AP\setminus \Lab(\state')}\neg p_{\overleftarrow{\vartrace}} \,\wedge \, \neg \pad_{\overleftarrow{\vartrace}} \,\wedge  \\
 & &  \bigwedge_{\state''\in \States\setminus\{\state'\}}\neg  \state''_{\overleftarrow{\vartrace}} \,\wedge \, \Next \bigl(  \pad_{\overleftarrow{\vartrace}}\vee  \bigvee_{\state''\in \Trans^{-1}(\state')}   \state''_{\overleftarrow{\vartrace}}\bigr) \Bigr] \Bigr)
\end{eqnarray*}

 We now extend the notions of backward and forward encodings of paths of $\KS$ to pointed paths $(\Path,i)$  of $\KS$. For each $k\in\nat$, the \emph{forward $\vartrace$-encoding of $(\Path,i)$ with offset $k$} is the pointed trace
given by $(\overrightarrow{\trace}(\vartrace,\Path,k),i+k)$. We say that  the position $i+k$  is in \emph{forward mode} in the encoding.  For each $k>0$,    the \emph{backward $\vartrace$-encoding of $(\Path,i)$ with offset $k$} is the pointed trace
given by $(\overleftarrow{\trace}(\vartrace,\Path,k),j)$, where $j= i-k$ if $i\geq k$, and $j=k-i$ otherwise. In the first case, we say that position $j$ is in \emph{forward mode} in the encoding, and the second case, we say that position $j$ is in \emph{backward mode} in the encoding.
Intuitively, when $(\trace,\ell)$ is a backward or forward encoding of a pointed path $(\Path,i)$   of $\KS$, then  $\ell$ represents the position in the encoding associated to position $i$ of $\Path$.\vspace{0.2cm}

\noindent \emph{Encoding of $(\KS,\FStates)$-path assignments.} Let $\TracesMap$ be a $(\KS,\FStates)$-path assignment. A pointed trace $(\trace,i)$ over $\AP'$ is a
\emph{forward encoding of $\TracesMap$} if for each $\vartrace\in\Var$, the following holds, where $\trace_{\vartrace}$ denotes the projection of
$\trace$ over $\States_{\vartrace}\cup \AP_{\vartrace}\cup\{\pad_{\overleftarrow{\vartrace}},\pad_{\overrightarrow{\vartrace}}\}$:
\begin{itemize}
  \item if $\vartrace\notin \Dom(\TracesMap)$, then $\trace_{\vartrace}$ is $\emptyset^{\omega}$;
  \item if $\vartrace\in \Dom(\TracesMap)$, then $(\trace_{\vartrace},i)$ is a forward or backward encoding of
  $\TracesMap(\vartrace)$ where $i$ is in forward mode in the encoding.
\end{itemize}
When $\Dom(\TracesMap)\neq\emptyset$, we also consider the notion of a \emph{backward encoding} $(\trace,i)$ of $\TracesMap$ which is defined as a forward encoding but we require that for each $\vartrace\in\Dom(\TracesMap)$,
$(\trace_{\vartrace},i)$ is a  backward encoding of $\TracesMap(\vartrace)$ where $i$ is in backward mode in the encoding. Note that in this case, we have that
$i\geq 1$ holds.

Let $\TracesMap$ be a $(\KS,\FStates)$-path assignment and   $\TracesMap'$ be the trace assignment over $\Lang(\KS,\FStates)$ obtained from $\TracesMap$
in the obvious way. For each $\ell\geq 0$, we write $\SUCC^{\ell}(\TracesMap)$ to mean $\SUCC^{\ell}_{(\emptyset,\Var)}(\TracesMap')$ and $\PRED^{\,\ell}(\TracesMap)$
to mean $\PRED^{\,\ell}_{(\emptyset,\Var)}(\TracesMap')$.
By construction, we easily obtain the following result, where $\Halt_\rightarrow \DefinedAs \bigvee_{\vartrace\in \Var} \pad_{\overrightarrow{\vartrace}}$ and
$\Halt_\leftarrow \DefinedAs \bigvee_{\vartrace\in \Var} \pad_{\overleftarrow{\vartrace}}$.\vspace{0.2cm}

\noindent\emph{Claim 2.} Let $\TracesMap$ be a $(\KS,\FStates)$-path assignment and $\ell\geq 0$.
\begin{itemize}
  \item If $(\trace,i)$ is a forward encoding of $\TracesMap$, then:
  \begin{itemize}
  \item  $(\trace,i+\ell)$ is a forward encoding  of $\SUCC^{\ell}(\TracesMap)$;
  \item if $\ell\leq i$, then $\PRED^{\,\ell}(\TracesMap)\neq \Undef$  iff  $(\trace,i-\ell)\models \neg \Halt_{\rightarrow}$. Moreover, if
   $\PRED^{\,\ell}(\TracesMap)\neq \Undef$, then $(\trace,i-\ell)$ is a forward encoding of $\PRED^{\,\ell}(\TracesMap)$;
  \item if $\ell> i$, then $\PRED^{\,\ell}(\TracesMap)\neq \Undef$  iff    $(\trace,\ell-i)\models \neg \Halt_{\leftarrow}$. Moreover, if
   $\PRED^{\,\ell}(\TracesMap)\neq \Undef$, then $(\trace,\ell-i)$ is a backward encoding of $\PRED^{\,\ell}(\TracesMap)$.
\end{itemize}
  \item If $(\trace,i)$ is a backward encoding of $\TracesMap$, then:
  \begin{itemize}
  \item  If $\ell<i$, $(\trace,i-\ell)$ is a backward encoding  of $\SUCC^{\ell}(\TracesMap)$;
  \item  if $\ell\geq i$, $(\trace,\ell-i)$ is a forward encoding  of $\SUCC^{\ell}(\TracesMap)$;
  \item  $\PRED^{\,\ell}(\TracesMap)\neq \Undef$  iff   $(\trace,i+\ell)\models \neg \Halt_{\leftarrow}$. Moreover, if
   $\PRED^{\,\ell}(\TracesMap)\neq \Undef$, then $(\trace,i+\ell)$ is a backward encoding of $\PRED^{\,\ell}(\TracesMap)$.
\end{itemize}
\end{itemize}

\noindent \emph{Reduction to $\QPTL$ satisfiability.}
 We  define by structural induction a mapping $\TMap:\{\leftarrow,\rightarrow\}\times   \SGHLTL{\emptyset} \rightarrow \QPTL $ associating to each pair $(\dir,\phi)$ consisting of
   a direction $\dir\in \{\leftarrow,\rightarrow\}$ and a $\SGHLTL{\emptyset}$ formula $\phi$ over $\AP$ and $\Var$ a $\QPTL$ formula $\TMap(\dir,\phi)$ over $\AP'$.
   Define $\AP'_{\vartrace}:= \AP_{\vartrace}\cup \States_{\vartrace}\cup \{\pad_{\overleftarrow{\vartrace}},\pad_{\overrightarrow{\vartrace}}\}$.
\begin{itemize}
 \item $\TMap(\dir,\top)=\top$;
  \item $\TMap(\leftarrow,\Rel{p}{\vartrace}) = p_{\overleftarrow{\vartrace}}$ for all $p\in \AP$;
   \item $\TMap(\rightarrow,\Rel{p}{\vartrace}) = p_{\overrightarrow{\vartrace}}$ for all $p\in \AP$;
   \item $\TMap(\dir,\ctx{x}\Rel{\psi}{\vartrace}) = \TMap_{\vartrace}(\dir, \Rel{\psi}{\vartrace})$, where
   $\TMap_{\vartrace}(\dir, \Rel{\psi}{\vartrace})$ is obtained from $\TMap(\dir, \Rel{\psi}{\vartrace})$ by replacing each occurrence of $\Halt_{\leftarrow}$
   (resp., $\Halt_{\rightarrow}$) with $\pad_{\overleftarrow{\vartrace}}$
   (resp., $\pad_{\overrightarrow{\vartrace}}$);
  \item $\TMap(\dir,\neg \phi)= \neg \TMap(\dir,\phi)$;
   \item $\TMap(\dir,\phi_1\vee \phi_2)= \TMap(\dir,\phi_1)\vee \TMap(\dir, \phi_2)$;
   \item $\TMap(\leftarrow,\Next \phi)= \Yesterday (\TMap(\leftarrow,\phi)\wedge \Yesterday\top) \,\vee\, \Yesterday (\TMap(\rightarrow,\phi)\wedge \neg\Yesterday\top)$;
   \item $\TMap(\rightarrow,\Next \phi)= \Next\,\TMap(\rightarrow,\phi)$;
  \item $\TMap(\leftarrow,\Yesterday \phi)=\Next (\TMap(\leftarrow,\ \phi)\wedge\neg\Halt_{\leftarrow})$;
  \item $\TMap(\rightarrow,\Yesterday \phi)=\Yesterday (\TMap(\rightarrow,\ \phi)\wedge\neg\Halt_{\rightarrow})\vee [\neg\Yesterday\top \wedge \Next (\TMap(\leftarrow,\ \phi)\wedge\neg\Halt_{\leftarrow})]$;
  \item $\TMap(\leftarrow,\phi_1  \Until  \phi_2)= \TMap(\leftarrow,\phi_1)  \Since  (\Yesterday\top \wedge \TMap(\leftarrow,\phi_2)) \,\, \vee $\\
  \phantom{$\TMap(\leftarrow,\phi_1  \Until  \phi_2)=$}
  $[\Historically(\Yesterday\top  \rightarrow \TMap(\leftarrow,\phi_1))\wedge \Once(\neg \Yesterday\top \wedge \TMap(\rightarrow,\phi_1) \Until  \TMap(\rightarrow,\phi_2)) ]$;
   \item  $\TMap(\rightarrow,\phi_1  \Until \phi_2)= \TMap(\rightarrow,\phi_1)  \Until \TMap(\rightarrow,\phi_2)$;
    \item  $\TMap(\leftarrow,\phi_1 \Since  \phi_2)= \TMap(\leftarrow,\phi_1)   \Until (\TMap(\leftarrow,\phi_2)\wedge \neg\Halt_{\leftarrow})$;
    \item $\TMap(\rightarrow,\phi_1  \Since  \phi_2)= \TMap(\rightarrow,\phi_1)  \Since   (\TMap(\rightarrow,\phi_2)\wedge \neg\Halt_{\rightarrow})) \,\, \vee $\\
  \phantom{$\TMap(\rightarrow,\phi_1  \Since  \phi_2)=$}
  $[\Historically \TMap(\rightarrow,\phi_1) \wedge \Once\bigl( \neg\Yesterday\top \wedge \Next(\TMap(\leftarrow,\phi_1)  \Until  (\TMap(\leftarrow,\phi_2)\wedge \neg\Halt_{\leftarrow}))\bigr)]$;
    \item  $\TMap(\leftarrow,\exists \vartrace.\,\phi)= \exists \,\AP'_{\vartrace}.\,   (\theta(\vartrace,\leftarrow)\wedge \TMap(\leftarrow,\phi) \wedge \neg\pad_{\overleftarrow{\vartrace}}
    \wedge \Next  \pad_{\overleftarrow{\vartrace}})$;
     \item  $\TMap(\rightarrow,\exists \vartrace.\,\phi)= \exists \,\AP'_{\vartrace}.\,   (\theta(\vartrace,\rightarrow)\wedge \TMap(\rightarrow,\phi) \wedge \neg\pad_{\overrightarrow{\vartrace}}
    \wedge (\Yesterday \top \rightarrow \Yesterday \pad_{\overrightarrow{\vartrace}}))$;
  \item $\TMap(\leftarrow,\exists^{\Pt} \vartrace.\,\phi)= \exists \,\AP'_{\vartrace}.\,   (\theta(\vartrace,\leftarrow)\wedge \TMap(\leftarrow,\phi) \wedge \neg\pad_{\overleftarrow{\vartrace}})$;
  \item $\TMap(\rightarrow,\exists^{\Pt} \vartrace.\,\phi)= \exists \,\AP'_{\vartrace}.\,   (\TMap(\rightarrow,\phi) \wedge [\theta(\vartrace,\leftarrow)\vee
  (\theta(\vartrace,\rightarrow)\wedge \neg\pad_{\overrightarrow{\vartrace}})]) $;
\end{itemize}

\noindent where $\theta(\vartrace,\leftarrow)$ and $\theta(\vartrace,\rightarrow)$ are the $\PLTL$ formulas of Claim~1. By construction  $\TMap(\dir,\phi)$ has size linear in $\phi$ and has the same strong alternation depth as $\phi$.
Moreover, $\TMap(\dir,\phi)$ is a $\QPTL$ sentence if $\phi$ is a $\SGHLTL{\emptyset}$ sentence. Now, we prove that the construction is correct.
Given a $\SGHLTL{\emptyset}$ formula $\phi$ and a $(\KS,\FStates)$-assignment $\TracesMap$ such that $\Dom(\TracesMap)$ consists of all and only the trace variables occurring free in $\phi$, we write
$\TracesMap\models \phi$ to mean that $(\TracesMap',\Var)\models_{\Lang(\KS,\FStates)}\phi $, where  $\TracesMap'$ is the trace assignment over $\Lang(\KS,\FStates)$ obtained from $\TracesMap$
in the obvious way. For a trace $\trace$ over $\AP'$, a variable $\vartrace\in\Var$,   and a trace $\trace_{\vartrace}$ over
$\AP'_{\vartrace}$, we denote by $\trace[\vartrace \mapsto \trace_{\vartrace}]$ the trace  $\trace'$ over $\AP'$ such that
$\trace' =_{\AP'\setminus \AP'_{\vartrace}}\trace$  and the projection of $\trace'$ over $\AP'_{\vartrace}$ is $\trace_{\vartrace}$.
 For a $\SGHLTL{\emptyset}$ sentence $\varphi$,  we show that
$\TracesMap_{\emptyset} \models  \varphi \Leftrightarrow  (\emptyset^{\omega},0) \models \TMap(\rightarrow,\varphi)$ where $\Dom(\Pi_\emptyset)=\emptyset$.
The result directly follows from the following claim. \vspace{0.2cm}

\noindent \emph{Claim 3:} let $\dir\in \{\leftarrow,\rightarrow\}$, $\phi$ be a $\SGHLTL{\emptyset}$ formula, and
$\TracesMap$ be a $(\KS,\FStates)$-assignment such that
 $\Dom(\TracesMap)$ contains all the trace variables occurring free in $\phi$ and $\dir = \rightarrow$ if $\Dom(\TracesMap)=\emptyset$. Then, for all
   encodings $(\trace,i)$ of $\TracesMap$ such that $(\trace,i)$ is a forward encoding if $\dir=\rightarrow$, and a backward encoding otherwise, the following holds:
$
\TracesMap \models  \phi \Leftrightarrow  (\trace,i) \models \TMap(\dir,\phi).
$ \vspace{0.2cm}

\noindent The proof of Claim~3 is by structural induction on $\phi$. When $\phi$ is a trace-relativized atomic proposition, the result directly follows
from the definition of the map $\TMap$ and the notion of an encoding $(\trace,i)$ of $\Pi$.  The cases for the boolean connectives directly follow from the induction hypothesis, while the cases relative to the temporal modalities  easily follow from the induction hypothesis, Claim~2, and the definition of the map $\TMap$. For the other cases, the ones relative to the hyper quantifiers,  we proceed as follows:
\begin{itemize}
 \item $\phi= \exists \vartrace.\, \phi'$: we focus on the case, where $\dir=\leftarrow$ (the other case being similar). By hypothesis, $(\trace,i)$ is a backward encoding of $\TracesMap$ and $\Dom(\TracesMap)\neq\emptyset$. Hence, it holds that $i\geq 1$.
For the implication  $ \TracesMap \models  \phi \Rightarrow  (\trace,i) \models \TMap(\leftarrow,\phi) $, assume that
$\TracesMap \models  \phi$. Hence, there exists a $\FStates$-fair  path $\Path$ of $\KS$ such that
$\TracesMap[\vartrace \mapsto (\Path,0)]\models \phi'$. Since $i\geq 1$, by construction, the trace $\trace_{\vartrace}$ over $\AP'_{\vartrace}$ given by   $(\overleftarrow{\trace}(\vartrace,\Path,i),i)$ is the backward
$\vartrace$-encoding of $(\Path,0)$ where $i$ is in backward mode in the encoding.
Hence, the trace $\trace'$ given by $\trace[\vartrace \mapsto \trace_{\vartrace}]$ is a backward encoding of $\TracesMap[\vartrace \mapsto (\Path,0)]$.
By the induction hypothesis, $(\trace',i) \models \TMap(\leftarrow,\phi')$. Moreover, by construction and Claim~1,
    $(\trace_{\vartrace},i) \models \theta(\vartrace,\leftarrow)$, $\pad_{\overleftarrow{\vartrace}}\notin \trace_{\vartrace}(i)$ and
  $\pad_{\overleftarrow{\vartrace}}\in \trace_{\vartrace}(i+1)$. Hence, by definition of $\TMap(\leftarrow,\exists \vartrace.\, \phi')$,
  we obtain that  $(\trace,i) \models \TMap(\leftarrow,\phi)$.

For the converse implication, assume that $(\trace,i) \models \TMap(\leftarrow,\phi)$.
By definition of $\TMap(\leftarrow,\exists \vartrace.\, \phi')$ and  Claim~1, there exists a trace
$\trace_{\vartrace}$ over $\AP'_{\vartrace}$ such that $\trace_{\vartrace}$ is a
 backward $\vartrace$-encoding $\overleftarrow{\trace}(\vartrace,\Path,k)$ of some  $\FStates$-fair   path $\Path$ of $\KS$ for some offset $k>0$. Moreover,
 $(\trace',i)\models \TMap(\leftarrow,\phi')$, where $\trace'$ is given by $\trace[\vartrace \mapsto \trace_{\vartrace}]$,
  $\pad_{\overleftarrow{\vartrace}}\notin \trace_{\vartrace}(i)$ and
  $\pad_{\overleftarrow{\vartrace}}\notin \trace_{\vartrace}(i+1)$. This means that the offset $k$ is exactly $i$. Hence,
  $(\overleftarrow{\trace}(\vartrace,\Path,i),i)$ is a backward
$\vartrace$-encoding of $(\Path,0)$ where $i$ is in backward mode in the encoding, and $(\trace',i)$ is a backward encoding
of  $\TracesMap[\vartrace \mapsto (\Path,0)]$. Thus, by the induction hypothesis, the result directly follows.
\item $\phi= \exists^{\Pt} \vartrace.\, \phi'$: we focus on the case, where $\dir=\rightarrow$ (the other case being similar). By hypothesis, $(\trace,i)$ is a forward encoding of $\TracesMap$.
For the implication  $ \TracesMap \models  \phi \Rightarrow  (\trace,i) \models \TMap(\rightarrow,\phi) $, assume that
$\TracesMap \models  \phi$. Hence, there exists a $\FStates$-fair  path $\Path$ of $\KS$ and a position $\ell\geq 0$ such that
$\TracesMap[\vartrace \mapsto (\Path,\ell)]\models \phi'$. Let $(\trace_{\vartrace},i)$ be the trace over $\AP'_{\vartrace}$ defined as follows:
\begin{itemize}
  \item if $\ell\leq i$: we set $(\trace_{\vartrace},i)$ to $(\overrightarrow{\trace}(\vartrace,\Path,i-\ell),i)$ which by construction is a forward
$\vartrace$-encoding of $(\Path,\ell)$. By construction and Claim~1, $(\trace_{\vartrace},i)\models \theta(\vartrace,\rightarrow)\wedge \neg\pad_{\overrightarrow{\vartrace}}$;
  \item  if $\ell> i$: we set $(\trace_{\vartrace},i)$ to $(\overleftarrow{\trace}(\vartrace,\Path,\ell-i),i)$ which by construction is a backward
$\vartrace$-encoding of $(\Path,\ell)$ where $i$ is in forward mode in the encoding. By construction and Claim~1, $(\trace_{\vartrace},i)\models \theta(\vartrace,\leftarrow)$.
\end{itemize}
Hence, the trace $\trace'$ given by $\trace[\vartrace \mapsto \trace_{\vartrace}]$ is a forward encoding of $\TracesMap[\vartrace \mapsto (\Path,\ell)]$.
By the induction hypothesis, $(\trace',i) \models \TMap(\leftarrow,\phi')$. Hence, by definition of $\TMap(\rightarrow,\exists^{\Pt} \vartrace.\, \phi')$,
  we obtain that  $(\trace,i) \models \TMap(\rightarrow,\phi)$.

For the converse implication, assume that $(\trace,i) \models \TMap(\rightarrow,\phi)$.
By definition of $\TMap(\rightarrow,\exists \vartrace.\, \phi')$ and  Claim~1, there exists a  $\FStates$-fair  path $\Path$ of $\KS$ and a trace
$\trace_{\vartrace}$ over $\AP'_{\vartrace}$ such that $(\trace',i)\models \TMap(\rightarrow,\phi')$, where $\trace'$ is given by $\trace[\vartrace \mapsto \trace_{\vartrace}]$ and
 one of the following holds:
\begin{itemize}
  \item  $\trace_{\vartrace}$ is the forward
$\vartrace$-encoding     $\overrightarrow{\trace}(\vartrace,\Path,k)$ of  $\Path$ for some offset $k\geq 0$ and $\pad_{\overrightarrow{\vartrace}}\notin \trace_{\vartrace}(i)$. It follows that $k\leq i$ and $(\trace_{\vartrace},i)$ is the
  forward $\vartrace$-encoding of the pointed path $(\Path,i-k)$.
  \item  $\trace_{\vartrace}$ is the backward
$\vartrace$-encoding     $\overleftarrow{\trace}(\vartrace,\Path,k)$ of  $\Path$ for some offset $k> 0$.
Hence, $(\trace_{\vartrace},i)$ is
the backward $\vartrace$-encoding of the pointed path $(\Path,i+k)$ where $i$ is in forward mode in the encoding.
\end{itemize}
Hence, for some $\ell\geq 0$, $(\trace',i)$ is a forward encoding
of $\TracesMap[\vartrace \mapsto (\Path,\ell)]$. Thus, being $(\trace',i)\models \TMap(\rightarrow,\phi')$,  by the induction hypothesis, the result directly follows.
\end{itemize}
This concludes the proof of Claim~3 and  Theorem~\ref{theorem:fromEmptyGammaFragmentToQPTL} too.
 \end{proof}

\subsection{Proof of Theorem~\ref{theorem:fromQPTLtoEmptyGammaFragment}}\label{app:fromQPTLtoEmptyGammaFragment}

\theoFromQPTLtoEmptyGammaFragment*
\begin{proof}Without loss of generality, we only consider \emph{well-named} $\QPTL$ formulas, \ie
$\QPTL$  formulas where each quantifier introduces a different proposition.
Moreover, we can assume that $\AP$ is the set of all and only the propositions occurring in the given $\QPTL$ sentence.
 Let $\AP'=\AP\cup \{\Tag,\ini\}$, where $\Tag$ and $\ini$ are fresh propositions, and fix an ordering $\{p_1,\ldots, p_n\}$ of the propositions in $\AP$.
First, we encode a trace $\trace$ over $\AP$  by a trace $\en(\trace)$ over $\AP'$ defined as follows: $\en(\trace):=\trace_0 \cdot \trace_1 \cdot \ldots$, where
for each $i\geq 0$,  $\trace_i$ (the encoding of the $i^{th}$ symbol of $\trace$) is the finite word over $2^{\AP'}$ of length $n+1$ given by $P_{0} P_1,\ldots P_n$, where
(i) for all $k\in [1,n]$,  $P_k= \{p_k\}$ if $p_k\in \trace(i)$, and $P_k=\emptyset$ otherwise, and (ii) $P_0= \{\Tag\}$ if $i>0$ and $P_0=\{\Tag,\ini\}$ otherwise.
Note that proposition $\ini$ marks the first position of $\en(\trace)$.
 Then, the finite Kripke structure   $\KS_\AP= \mktuple{\States,\States_0,\Trans,\Lab}$ over $\AP'$ has size linear in $|\AP|$ and it is constructed in such a way that its set of traces $\Lang(\KS_\AP)$ is the set of the encodings of the traces over $\AP$.  Formally, $\KS_\AP$ is defined as follows:
\begin{itemize}
  \item $\States=\{p_h,\overline{p}_h\mid h\in\{1,\ldots,n\}\}\cup \{\Tag,\ini\}$ and $\States_0 = \{\ini\}$;
  \item $\Trans$ consists of the edges $(p_k,p_{k+1})$, $(p_k,\overline{p}_{k+1})$,
  $(\overline{p}_k,p_{k+1})$ and $(\overline{p}_k,\overline{p}_{k+1})$ for all $k\in [1,n-1]$, and the edges
  $(\state,p_1)$, $(\state,\overline{p}_1)$, $(p_n,\Tag)$, and $(\overline{p}_n,\Tag)$ where $\state\in \{\Tag,\ini\}$.
  \item $\Lab(\Tag)=\{\Tag\}$, $\Lab(\ini)=\{\ini,\Tag\}$, and  $\Lab(p_k)=\{p_k\}$ and $\Lab(\overline{p}_k)=\emptyset$ for all $k\in [1,n]$.
\end{itemize}

Let $\Lambda$ be the set of pairs $(h,\psi)$ consisting of a natural number $h\in [0,n]$ and
 a well-named $\QPTL$ formula $\psi$ over $\AP$ such that there is no quantifier in $\psi$ binding proposition $p_h$ if $h\neq 0$,
and $\psi$ is a sentence iff $h=0$. We inductively define a mapping assigning to each pair $(h,\psi)\in\Lambda$ a
 singleton-free $\SGHLTL{\emptyset}$ formula $\TMap(h,\psi)$ over $\AP'$ and $\Var=\{\vartrace_1,\ldots,\vartrace_n\}$ (intuitively, if $h\neq 0$, then $p_h$ represents the currently quantified proposition):
\begin{itemize}
 \item $\TMap(h,\top)=\top$;
  \item $\TMap(h,p_i) = \Next^{i}\, \Rel{p_i}{\vartrace_h}$ for all $p_i\in \AP$;
  \item $\TMap(h,\neg \psi)= \neg \TMap(h,\psi)$;
   \item $\TMap(h,\psi_1\vee \psi_2)= \TMap(h,\psi_1)\vee \TMap(h,\psi_2)$;
   \item $\TMap(h,\Next \psi)= \Next^{n+1} \TMap(h,\psi)$;
    \item $\TMap(h,\Yesterday \psi)= \Yesterday^{n+1} \TMap(h,\psi)$;
   \item  $\TMap(h,\psi_1 \Until  \psi_2)= (\Rel{\Tag}{\vartrace_h} \rightarrow \TMap(h,\psi_1)) \Until  (\TMap(h,\psi_2)\wedge \Rel{\Tag}{\vartrace_h})$;
    \item  $\TMap(h,\psi_1 \Since  \psi_2)= (\Rel{\Tag}{\vartrace_h} \rightarrow \TMap(h,\psi_1)) \Since  (\TMap(h,\psi_2)\wedge \Rel{\Tag}{\vartrace_h})$;
   \item  $
\begin{array}{ll}
 \TMap(h,\exists p_k.\psi) =
  \begin{cases}
  \exists^{\Pt}\,\vartrace_k. \,  \Bigl(\TMap(k,\psi)  \wedge \Once \bigl(\Rel{\ini}{\vartrace_h}\wedge \Rel{\ini}{\vartrace_k}\wedge  &\\
  \phantom{\exists^{\Pt}\,\vartrace_k. \,  \Bigl(\TMap(k,\psi)  \wedge \Once \bigl(}
  \Always \displaystyle{\bigwedge_{j\in [1,n]\setminus \{k\}}}(\Rel{p_j}{\vartrace_h} \leftrightarrow \Rel{p_j}{\vartrace_k}\bigr)\Bigr)
    & \text{if } h \neq 0\\
   \exists\,\vartrace_k. \, \TMap(k,\psi)
      & \text{ otherwise }
  \end{cases}
\end{array}
 $
\end{itemize}

By construction,  $\TMap(h,\psi)$ has size linear in $\psi$ and has the same strong alternation depth as $\psi$.
Moreover, $\TMap(h,\psi)$ is a $\SGHLTL{\emptyset}$ sentence \emph{iff} $\psi$ is a $\QPTL$ sentence. Now we show that the construction is correct.
A  \emph{$\Lang(\KS_{\AP})$-assignment} is a mapping  $\TracesMap:\Var \rightarrow \Lang(\KS_{\AP})$.
For all $i\geq 0$, $\Lang(\KS_{\AP})$-assignments $\TracesMap$, and $\SGHLTL{\emptyset}$ formulas $\varphi$ over $\AP'$, we write
$(\TracesMap,i)\models \varphi$ to mean that $(\TracesMap_i,\Var)\models_{\Lang(\KS_{\AP})} \varphi$, where $\TracesMap_i$ is the (pointed)
trace assignment over $\Lang(\KS_{\AP})$ assigning to each trace variable $\vartrace\in\Var$, the pointed trace $(\TracesMap(\vartrace),i)$.
Note that for each $\QPTL$ sentence $\psi$ over $\AP$, $\psi$ is satisfiable if for all traces $\trace$, $(\trace,0)\models \psi$.
Thus, correctness of the construction directly follows from the following claim,  where
for each $i\geq 0$,   $\wp(i):= i\cdot(n+1)$.  Intuitively, $\wp(i)$ is the $\Tag$-position associated with the $\AP'$-encoding of the position $i$ of a trace over $ \AP$.\vspace{0.2cm}

\noindent \emph{Claim.} Let $(h,\psi)\in\Lambda$. Then, for all pointed traces $(\trace,i)$ over $\AP$ and $\Lang(\KS_{\AP})$-assignments $\TracesMap$
  such that $\TracesMap(\vartrace_h) = \en(\trace)$ if $h\neq 0$, and $i=0$ if $h=0$, the following holds:
\[
(\trace,i) \models \psi \Leftrightarrow (\TracesMap,\wp(i)) \models  \TMap(h,\psi)
\]
 The claim is proved  by structural induction  on $\psi$. The cases for the boolean connectives  easily follow from the induction hypothesis. For the other cases, we proceed as follows:
\begin{itemize}
   \item $\psi = p_j$ for some $p_j\in \AP$: in this case $h\neq 0$ (recall that $(h,\psi)\in\Lambda$). Then, we have that
  $(\trace,i) \models p_j$  $\Leftrightarrow$ $p_j\in \trace(i)$ $\Leftrightarrow$  $p_j\in \en(\trace)(\wp(i)+j)$ $\Leftrightarrow$
  $p_j\in  \TracesMap(\vartrace_h) (\wp(i)+j)$ $\Leftrightarrow$ $(\TracesMap, \wp(i)) \models \Next^{j}\, \Rel{p_j}{x_h} $
  $\Leftrightarrow$ $(\TracesMap, \wp(i)) \models  \TMap(h,p_j)$.
 Hence, the result follows.
  \item $\psi =\Next \psi'$: hence, $h\neq 0$. We have that
  $(\trace,i) \models \Next \psi'$  $\Leftrightarrow$ $(\trace,i+1)\models \psi'$ $\Leftrightarrow$ (by the induction hypothesis)
   $(\TracesMap, \wp(i+1)) \models \TMap(h,\psi')$ $\Leftrightarrow$ (since $\wp(i+1)= \wp(i)+n+1$)
   $(\TracesMap, \wp(i)) \models  \Next^{n+1} \TMap(h,\psi')$ $\Leftrightarrow$ $(\TracesMap,\wp(i))\models \TMap(h,\Next\psi')$.
 Hence, the result follows.
 \item $\psi = \Yesterday \psi'$: similar to the previous case.
 \item $\psi =\psi_1\Until \psi_2$: hence, $h\neq 0$. We have that
  $(\trace,i) \models \psi_1 \Until \psi_2$  $\Leftrightarrow$
  there is $j\geq i$ such that
 $(\trace,j)\models \psi_2$ and $(\trace,\ell) \models\psi_1$ for all $i\leq \ell<j$ $\Leftrightarrow$
 (by the induction hypothesis)
   there is $j\geq i$ such that
 $(\TracesMap, \wp(j)) \models\TMap(h,\psi_2)$ and $(\TracesMap, \wp(\ell)) \models \TMap(h,\psi_1)$ for all $i\leq \ell<j$ $\Leftrightarrow$
 there is $j'\geq \wp(i)$ such that $(\TracesMap,j') \models \TMap(h,\psi_2)$ and $\Tag\in \TracesMap(x_h)(j')$, and for all
 $\wp(i)\leq \ell'<j'$ such that $\Tag\in  \TracesMap(x_h)(\ell')$, $(\TracesMap,\ell') \models \TMap(h,\psi_1)$
 $\Leftrightarrow$ $(\TracesMap,\wp(i)) \models \TMap(h,\psi_1\Until\psi_2)$.
  Hence, the result follows.
 \item $\psi =\psi_1\Since \psi_2$: similar to the previous case.
  \item $\psi= \exists p_k.\, \psi'$ and $h\neq 0$: by hypothesis, $k\neq h$.
For the implication, $ (\TracesMap,\wp(i)) \models \TMap(h,\psi) \Rightarrow (\trace,i) \models \psi $, assume that
$(\TracesMap,\wp(i)) \models \TMap(h,\psi)$. By definition of $\TMap(h,\exists p_k.\, \psi')$ and since $\ini$ marks the first position
of the encoding of a trace over $\AP$,  it easily follows that there exists  a pointed  trace over $\AP$ of the form $(\trace',i)$ such that $\trace'=_{\AP\setminus\{p_k\}}\trace$
and $(\TracesMap[x_k \leftarrow \en(\trace')], \wp(i))\models  \TMap(k,\psi')$. Since $(k,\psi')\in\Lambda$, by the induction hypothesis,
it follows that  $(\trace',i) \models \psi'$.
  Thus,   being $\trace'=_{\AP\setminus\{p_k\}}\trace$, we obtain that  $(\trace,i) \models \psi$.

The converse implication  $(\trace,i) \models  \psi \Rightarrow  (\Pi,  \wp(i))\models \TMap(h,\psi)$ is similar, and we omit the details here.
\item  $\psi= \exists p_k.\, \psi'$ and $h= 0$: hence, $\psi$ and $\TMap(0,\psi)$ are sentences.
We have that
  $(\trace,0) \models \exists p_k.\, \psi'$  $\Leftrightarrow$ ($\psi$ is a $\QPTL$ sentence) for some trace $\trace'$ over $\AP$, $(\trace',0)\models \psi'$ $\Leftrightarrow$ (by the induction hypothesis)
   $(\TracesMap[\vartrace_k \rightarrow \en(\trace')], \wp(0)) \models \TMap(k,\psi')$ $\Leftrightarrow$ (by definition of $\TMap(0,\exists p_k.\, \psi')$)
   $(\TracesMap, \wp(0)) \models    \TMap(0,\exists p_k.\, \psi')$.
\end{itemize}
This concludes the proof of Theorem~\ref{theorem:fromQPTLtoEmptyGammaFragment}.
  \end{proof}

\subsection{Proof of Proposition~\ref{prop:FormSimpleToProposition}}\label{app:FormSimpleToProposition}

In order to prove Proposition~\ref{prop:FormSimpleToProposition}, we need some
preliminary results.
Recall that a Nondeterministic B\"{u}chi Automaton over words ($\NBA$ for short) is a tuple
$\Au =\mktuple{\Sigma,Q,Q_0,\Delta,\Acc}$, where $\Sigma$ is a finite
alphabet, $Q$ is a finite set of states, $Q_0\subseteq Q$ is the set
of initial states, $\Delta\subseteq Q\times \Sigma \times Q$ is the
transition relation, and $\Acc\subseteq Q$ is the set of
\emph{accepting} states. Given a infinite word $w$ over $\Sigma$, a
run of $\Au$ over $w$ is an infinite sequence of states
$q_0,q_1,\ldots$ such that $q_0\in Q_0$ and for all $i\geq 0$,
$(q_i,w(i),q_{i+1})\in \Delta$.
The run is accepting if for infinitely many $i$, $q_i\in\Acc$.  The
language $\Lang(\Au)$ accepted by $\Au$ consists of the infinite words
$w$ over $\Sigma$ such that there is an accepting run over $w$.

Fix a non-empty set $\Gamma$ of $\PLTL$ formulas over $\AP$.
The closure $\cl(\Gamma)$ of $\Gamma$ is the set of $\PLTL$ formulas
consisting of the formulas $\top$, $\Yesterday\top$, the  sub-formulas of the formulas $\theta\in\Gamma$, and
the negations of such formulas (we identify $\neg\neg\theta$ with $\theta$).
Note that $\Gamma\subseteq \cl(\Gamma)$.
Without loss of generality, we can assume that $\AP\subseteq \Gamma$.
Precisely, $\AP$ can be taken as the set of propositions occurring in
the given simple $\GHLTL$ sentence and $\cl(\Gamma)$ contains all the
propositions in $\AP$ and their negations.
For each formula $\theta\in \cl(\Gamma)\setminus \AP$, we introduce a
fresh atomic proposition not in $\AP$, denoted by
$\at(\theta)$. Moreover, for allowing a uniform notation, for each
$p\in \AP$, we write $\at(p)$ to mean $p$ itself.
Let $\AP_\Gamma$ be the set $\AP$ extended with these new
propositions.
By a straightforward adaptation of the well-known translation of
$\PLTL$ formulas into equivalent $\NBA$~\cite{VardiW94}, we
obtain the following result, where for a trace $\trace_\Gamma$ over
$\AP_\Gamma$, $(\trace_\Gamma)_{\AP}$ denotes the projection of $\trace_\Gamma$ over $\AP$.

\begin{proposition}
  \label{prop:NFAforLTL}
  Given a finite set $\Gamma$ of $\PLTL$ formulas over $\AP$, one can
  construct in single exponential time an $\NBA$ $\Au_\Gamma$
  over $2^{\AP_\Gamma}$ with $2^{O(|\AP_\Gamma|)}$ states satisfying
  the following:
  \begin{enumerate}
  \item let $\trace_\Gamma\in\Lang(\Au_\Gamma)$: then for all $i\geq 0$ and
    $\theta\in \cl(\Gamma)$, $\at(\theta)\in \trace_\Gamma(i)$ \emph{iff}
    $((\trace_\Gamma)_{\AP},i)\models \theta$.
  \item for each trace $\trace$ over $\AP$,
    there exists $\trace_\Gamma\in \Lang(\Au_\Gamma)$ such that $\trace= (\trace_\Gamma)_{\AP}$.
\end{enumerate}
\end{proposition}

\begin{proof} Here, we construct a \emph{generalized} $\NBA$
  $\Au_\Gamma$ satisfying Properties~(1) and~(2) of
  Proposition~\ref{prop:NFAforLTL}, which can be converted in linear
  time into an equivalent $\NBA$. Recall that a generalized
  $\NBA$ is defined as an $\NBA$ but the acceptance condition
  is given by a family $\Family=\{\Acc_1,\ldots,\Acc_k\}$ of sets
  of accepting states. In this case, a run is accepting if for each
  accepting component $\Acc_i\in \Family$, the run visits
  infinitely often states in $\Acc_i$.

The generalized $\NBA$ $\Au_\Gamma = \mktuple{2^{\AP_\Gamma},Q,Q_0,\Delta,\Family}$  is defined as follows.
$Q$ is the set of \emph{atoms} of $\Gamma$ consisting of the maximal  propositionally consistent subsets $A$ of $\cl(\Gamma)$.
Formally, an atom $A$ of $\Gamma$ is a subset of $\cl(\Gamma)$   satisfying the following:
 \begin{itemize}
 \item $\top\in A$
  \item for each $\theta\in\cl(\Gamma)$, $\theta\in A$ iff $\neg\theta\notin A$;
    \item for each $\theta_1\vee \theta_2\in \cl(\Gamma)$, $\theta_1\vee\theta_2\in A$ iff $\{\theta_1,\theta_2\}\cap A\neq \emptyset$.
\end{itemize}

The set $Q_0$ of initial states consists of the atoms $A$ of $\Gamma$ such that $\neg\Yesterday\top\in A$.
For an atom $A$, $\at(A)$ denotes the subset of propositions in $\AP_\Gamma$ associated with the formulas in $A$, i.e.
$\at(A)\DefinedAs\{\at(\theta)\mid \theta\in A\}$.
The transition relation $\Delta$ captures the semantics of the next and previous modalities  and the local fixpoint characterization of the until and since  modalities.
Formally, $\Delta$ consists of the transitions of the form $(A,\at(A),A')$ such that:
 \begin{itemize}
 \item for each $\Next\theta\in\cl(\Gamma)$, $\Next\theta\in A$ iff $\theta\in A'$;
 \item for each $\Yesterday\theta\in\cl(\Gamma)$, $\Yesterday\theta\in A'$ iff $\theta\in A$;
  \item for each $\theta_1\Until\theta_2\in\cl(\Gamma)$, $\theta_1\Until\theta_2\in A$ iff either $\theta_2\in A$, or $\theta_1\in A$  and $\theta_1\Until\theta_2\in A'$;
  \item for each $\theta_1\Since\theta_2\in\cl(\Gamma)$, $\theta_1\Since\theta_2\in A'$ iff either $\theta_2\in A'$, or $\theta_1\in A'$  and $\theta_1\Since\theta_2\in A$.
\end{itemize}
Finally, the generalized B\"{u}chi acceptance condition is used for ensuring the fulfillment of the liveness requirements $\theta_2$ in the until sub-formulas $\theta_1\Until\theta_2$ in $\Gamma$. Formally, for each $\theta_1\Until\theta_2\in\cl(\Gamma)$, $\Family$ has a component
consisting of the atoms $A$ such that either $\neg(\theta_1\Until\theta_2)\in A$ or $\theta_2\in A$.

Let $\trace_\Gamma\in\Lang(\Au_\Gamma)$. By construction, there is an accepting infinite sequence of atoms $\rho=A_0 A_1\ldots$  such that for all $i\geq 0$, $\trace_\Gamma(i)=\at(A_i)$. Let $\trace$ be the projection of $\trace_\Gamma$ over $\AP$ (note that $A_i\cap\AP = \trace(i)$ for all $i\geq 0$).
By standard arguments (see~\cite{VardiW94}), the following holds: for all $i\geq 0$ and $\theta\in \cl(\Gamma)$, $\theta\in A_i$ (hence, $\at(\theta)\in \trace_\Gamma(i)$) if and only if $(\trace,i)\models \theta$.
Hence, Property~(1) of Proposition~\ref{prop:NFAforLTL} follows.

For Property~(2), let $\trace$ be a trace over $\AP$ and let $\rho= A_0 A_1\ldots$ be the infinite sequence of atoms defined
as follows for all $i\geq 0$: $A_i=\{\theta\in \cl(\Gamma)\mid (\trace,i)\models \theta\}$. By construction and the semantics of $\PLTL$,
$\rho$ is an accepting run of $\Au_\Gamma$ over the word $\trace_\Gamma= at(A_0) at(A_1)\ldots$. Moreover, $\trace$ coincides with the projection of $\trace_\Gamma$ over $\AP$. Hence, the result follows.
\end{proof}

Let $\KS= \mktuple{\States,\States_0,\Trans,\Lab}$ be a finite Kripke structure over $\AP$ and $\FStates\subseteq \States$.
Next, we consider the synchronous product of the fair Kripke structure
$(\KS , \FStates)$ with the $\NBA$
$\Au_\Gamma = \mktuple{2^{\AP_\Gamma},Q,Q_0,\Delta,\Acc}$ over
$2^{\AP_\Gamma}$ of Proposition~\ref{prop:NFAforLTL} associated with
$\Gamma$.
More specifically, we construct a Kripke structure $\KS_\Gamma$ over
$\AP_\Gamma$ and a subset $\FStates_\Gamma$ of $\KS_\Gamma$-states such that
$\Lang(\KS_\Gamma,F_\Gamma)$ is the set of traces
$\trace_\Gamma\in\Lang(\Au_\Gamma)$ whose projections over $\AP$ are in
$\Lang(\KS,\FStates)$.
Formally, the \emph{$\Gamma$-extension of $(\KS,\FStates)$} is the fair
Kripke structure $(\KS_\Gamma,\FStates_\Gamma)$ where
$\KS_\Gamma=\mktuple{\States_\Gamma,\States_{0,\Gamma},\Trans_\Gamma,\Lab_\Gamma}$ and
$\FStates_\Gamma$ are defined as follows:
\begin{itemize}
\item $\States_\Gamma$ is the set of tuples
  $(\state,B,q,\ell)\in \States\times 2^{\AP_\Gamma}\times Q\times \{1,2\}$ such
  that $\Lab(s)= B\cap \AP$;
\item
  $\States_{0,\Gamma}= \States_\Gamma \cap (\States_0\times 2^{\AP_\Gamma}\times
  Q_0\times \{1\})$;
\item $\Trans_\Gamma$ consists of the following transitions:
  \begin{itemize}
  \item $((\state,B,q,1),(\state',B',q',\ell))$ such that $(\state,\state')\in \Trans$,
    $(q,B,q')\in \Delta$, and $\ell=2$ if $\state\in \FStates$ and $\ell=1$
    otherwise;
  \item $((\state,B,q,2),(\state',B',q',\ell))$ such that $(\state,\state')\in \Trans$,
    $(q,B,q')\in \Delta$, and $\ell=1$ if $q\in \Acc$ and $\ell=2$
    otherwise.
  \end{itemize}
\item for each $(\state,B,q,\ell)\in \States_\Gamma$, $\Lab_\Gamma((s,B,q,\ell))=B$;
\item $\FStates_\Gamma=\{(\state,B,q,2)\in \States_\Gamma\mid q\in \Acc\}$.
\end{itemize}
 
\noindent By construction and Proposition~\ref{prop:NFAforLTL}(2), we
easily obtain the following result.

\begin{proposition}
  \label{prop:SynchronousProduct}
  For each trace $\trace_\Gamma$ over $\AP_\Gamma$,
  $\trace_\Gamma\in \Lang(\KS_\Gamma,\FStates_\Gamma)$ if and only if
  $\trace_\Gamma\in\Lang(\Au_\Gamma)$ and $(\trace_\Gamma)_{\AP}\in \Lang(\KS,\FStates)$. Moreover,
  for each $\trace\in \Lang(\KS,\FStates)$, there exists
  $\trace_\Gamma\in\Lang(\KS_\Gamma,\FStates_\Gamma)$ such that $(\trace_\Gamma)_{\AP} =\trace$.
\end{proposition}

\noindent We can now provide a proof of Proposition~\ref{prop:FormSimpleToProposition}.

\propFormSimpleToProposition*
\begin{proof}
Let $\varphi$ be a $\GHLTL$ sentence over $\AP$ and $(\KS,\FStates)$
be a fair finite Kripke  structure $(\KS,\FStates)$ over $\AP$.
Then, there is a finite set $\Gamma_0$ of $\PLTL$ formulas  such that
$\varphi$ is in the fragment $\SGHLTL{\Gamma_0}$.
Let $\Gamma$ be the set of $\PLTL$ formulas consisting of the formulas
in $\Gamma_0$ and the $\PLTL$ formulas $\psi$ such that
$\ctx{x}\Rel{\psi}{\vartrace}$ is a sub-formula of $\varphi$ for some variable $\vartrace$.
Define $\AP'$, $\Gamma'\subseteq \AP'$, $(\KS',\FStates')$, and $\varphi'$ as follows:
\begin{itemize}
 \item $\AP'\DefinedAs \AP_\Gamma$ (recall that $\AP_\Gamma=\AP\cup \{\at(\theta)\mid \theta\in \Gamma\}$)
 and $\Gamma'$ is the $\AP_\Gamma$-counterpart of $\Gamma_0$, \ie
 $\Gamma'\DefinedAs \{\at(\theta)\mid \theta\in \Gamma_0\}$;
 \item $(\KS',\FStates')\DefinedAs (\KS_\Gamma,\FStates_\Gamma)$, where
   $(\KS_\Gamma,\FStates_\Gamma)$ is the $\Gamma$-extension of $(\KS,\FStates)$;
 \item $\varphi'\DefinedAs \TMap(\varphi)$  where the mapping $\TMap$ replaces (i) each sub-formula
    $\ctx{x}\Rel{\psi}{\vartrace}$ of $\varphi$ with its propositional $\vartrace$-version
    $\Rel{\at(\psi)}{\vartrace}$, and (ii) each $\Gamma_0$-relativized temporal modality  with its
    $\Gamma'$-relativized version. Note that $\TMap(\varphi)$ is a singleton-free $\SGHLTL{\Gamma'}$ formula
    where $\Gamma'$ is propositional (in particular, $\Gamma'\subseteq \AP'$) and has the same strong alternation
    depth as $\varphi$.
 \end{itemize}

It remains to show that the construction is correct, \ie $\Lang(\KS_\Gamma,\FStates_\Gamma)\models \TMap(\varphi)$ \emph{iff}
  $\Lang(\KS,\FStates)\models \varphi$. Let $\Lambda$ be the set of formulas $\phi$ in the
fragment  $\SGHLTL{\Gamma_0}$ such that for each sub-formula $\ctx{x}\Rel{\psi}{\vartrace}$ of $\phi$,
it holds that $\psi\in\Gamma$. Moreover, for a trace assignment $\TracesMap_\Gamma$ over
 $\Lang(\KS_\Gamma,\FStates_\Gamma)$, the \emph{$\AP$-projection of
 $\TracesMap_\Gamma$}, written
 $(\TracesMap_\Gamma)_{\AP}$, is the trace assignment with domain
 $\Dom(\TracesMap)$ obtained from
 $\TracesMap$ by replacing each pointed trace  $\TracesMap(\vartrace)$, where $\vartrace\in\Dom(\TracesMap)$ and
 $\TracesMap(\vartrace)$ is of the form $(\trace_\Gamma,i)$,
 with $((\trace_\Gamma)_\AP,i)$.  Note that by Proposition~\ref{prop:SynchronousProduct},
 $(\TracesMap_\Gamma)_{\AP}$ is a trace assignment over $\Lang(\KS,\FStates)$.
Correctness of the construction directly follows from the following claim.\vspace{0.2cm}

\noindent \emph{Claim.} Let $\phi\in\Lambda$ and $\TracesMap_\Gamma$ be a trace assignment over $\Lang(\KS_\Gamma,\FStates_\Gamma)$.
Then:
\[
  (\TracesMap_\Gamma,\Var) \models_{\Lang(\KS_\Gamma,\FStates_\Gamma)}\TMap(\phi)$ iff $((\TracesMap_\Gamma)_{\AP},\Var) \models_{\Lang(\KS,\FStates)} \phi
\]
The claim is proved by  structural induction on $\phi\in\Lambda$. The cases where the root modality of $\phi$ is a Boolean connective  directly follow from the induction hypothesis. For the other cases, we proceed as follows.
\begin{itemize}
  \item $\phi= \ctx{x}\Rel{\psi}{\vartrace}$, where $\psi\in\Gamma$. Hence, $\TMap(\phi)=  \Rel{\at(\psi)}{\vartrace}$. Let $\TracesMap_\Gamma(\vartrace)=(\trace_\Gamma,i)$.
  We have that $(\TracesMap_\Gamma)_\AP(\vartrace)=((\trace_\Gamma)_\AP,i)$. By Propositions~\ref{prop:NFAforLTL} and~\ref{prop:SynchronousProduct},
  $\at(\psi)\in \trace_\Gamma(i)$ \emph{iff}
    $((\trace_\Gamma)_{\AP},i)\models \psi$. Hence, the result follows.
  \item The root modality of $\phi$ is a $\Gamma_0$-relativized temporal modality. Assume  that
  $\phi = \phi_1 \Until_{\Gamma_0} \phi_2$ (the other cases being similar). Then, $\TMap(\phi)= \TMap(\phi_1) \Until_{\Gamma'} \TMap(\phi_2)$ (recall that
  $\Gamma'= \{\at(\theta)\mid \theta\in \Gamma_0\}$ and $\Gamma_0\subseteq \Gamma$).  By Propositions~\ref{prop:NFAforLTL} and~\ref{prop:SynchronousProduct}, for each pointed trace $(\trace_\Gamma,i)$ over $\Lang(\KS_\Gamma,\FStates_\Gamma)$ and $\theta\in \Gamma_0$, it holds that
  $\at(\theta)\in \trace_\Gamma(i)$ \emph{iff}
    $((\trace_\Gamma)_{\AP},i)\models \theta$. By construction, it follows that $\SUCC_{(\Gamma_0,\Var)}((\TracesMap_\Gamma)_{\AP})$
    is the $\AP$-projection of $\SUCC_{(\Gamma',\Var)}(\TracesMap_\Gamma)$. Hence, by the semantics of $\GHLTL$ and the induction hypothesis, the result follows.
  \item $\phi = \exists\vartrace.\,\phi'$. Hence, $\TMap(\phi)= \exists\vartrace.\,\TMap(\phi')$. By Proposition~\ref{prop:SynchronousProduct}, for each trace $\trace\in\Lang(\KS,\FStates)$,
   there exists
  $\trace_\Gamma\in\Lang(\KS_\Gamma,\FStates_\Gamma)$ such that $(\trace_\Gamma)_{\AP} =\trace$. Hence, the result easily follows from the induction hypothesis.
  \item $\phi = \exists^{\Pt}\vartrace.\,\phi'$: this case is similar to the previous one.
\end{itemize}
\end{proof}

\subsection{Proof of Proposition~\ref{prop:ExtendedStutterTrace}}\label{app:ExtendedStutterTrace}

\propExtendedStutterTrace*
\begin{proof}
Let $\KS=\mktuple{\States,\States_0,\Trans,\Lab}$.
Intuitively, the Kripke structure $\KS_\Gamma$ is obtained from
 $\KS$ by adding edges which are \emph{summaries} of finite paths $\Path$ of $\KS$ where \emph{either} the propositional valuation in $\Gamma$
 changes only at the final state of $\Path$, or the propositional valuation in $\Gamma$ does not change along $\Path$.
 Formally, let $\TransR_\Gamma(\KS)$ and $\TransR_\Gamma(\KS,\FStates)$ be the
 sets of state pairs in $\KS$ defined as follows:
 \begin{itemize}
 \item $\TransR_\Gamma(\KS)$ consists of the pairs $(\state,\state')\in \States\times \States$
   such that $\Lab(\state)\cap \Gamma\neq \Lab(\state')\cap \Gamma$ and there is a
   finite path of $\KS$ of the form $\state\cdot \rho \cdot \state'$ such that
   $\Lab(\state)\cap \Gamma= \Lab(\rho(i))\cap \Gamma$ for all $0\leq i<|\rho|$.
 \item $\TransR_\Gamma(\KS,\FStates)$ is defined similarly but, additionally, we
   require that the finite path $\state\cdot \rho \cdot \state'$ visits some
     state in $\FStates$.
 \end{itemize}
 Intuitively, $\TransR_\Gamma(\KS)$ keeps track  of the initial and final states of the finite paths
 of $\KS$ with length at least $2$ where the propositional valuation in $\Gamma$ changes only at the final state of the finite path.
 Additionally, $\TransR_\Gamma(\KS,\FStates)$ considers only those finite paths which visit  some state in $\FStates$.
 Moreover, let $\TransR^{\pad}_\Gamma(\KS)$ and $\TransR^{\pad}_\Gamma(\KS,\FStates)$ be the
 sets of state pairs in $\KS$ defined as follows:
 \begin{itemize}
\item $\TransR^{\pad}_\Gamma(\KS)$ consists of the pairs $(\state,\state')\in \States\times \States$
   such that $\Lab(\state)\cap \Gamma = \Lab(\state')\cap \Gamma$ and there is a
   finite path of $\KS$ of the form $\state\cdot \rho \cdot \state'$ such that
   $\Lab(\state)\cap \Gamma= \Lab(\rho(i))\cap \Gamma$ for all $0\leq i<|\rho|$.
 \item $\Trans^{\pad}_\Gamma(\KS,\FStates)$ is defined similarly but, additionally, we
   require that the finite path $\state\cdot \rho \cdot \state'$ visits some
   accepting state in $\FStates$.
 \end{itemize}
Thus, $\TransR^{\pad}_\Gamma(\KS)$ keeps track  of the initial and final states of the finite paths
 of $\KS$ with length at least $2$ where the propositional valuation in $\Gamma$ does not change.
 Additionally, $\TransR^{\pad}_\Gamma(\KS,\FStates)$ considers only those finite paths which visit  some state in $\FStates$.
 The finite sets $\TransR_\Gamma(\KS)$, $\TransR_\Gamma(\KS,\FStates)$, $\TransR^{\pad}_\Gamma(\KS)$, $\TransR^{\pad}_\Gamma(\KS,\FStates)$  can be easily
 computed in polynomial time by standard closure algorithms. By
 exploiting these finite sets, we define
 the finite Kripke structure $\KS_\Gamma=\mktuple{\States_\Gamma,\States_{\Gamma,0},\Trans_\Gamma,\Lab_\Gamma}$ and the set $\FStates_\Gamma\subseteq \States_\Gamma$ as follows.
 \begin{itemize}
\item $\States_\Gamma$ is given by $\States\times 2^{\{\acc,\pad\}}$ and
  $\States_{\Gamma,0}$ is the set of states of the form $(\state,\emptyset)$ for some $\state\in \States_0$.
\item $\Trans_\Gamma$ consists of the edges $((\state,T),(\state',T'))$ such
  that one of the following holds:
  \begin{itemize}
  \item$(\state,\state')\in \Trans \cup \TransR_\Gamma(\KS)$, $\pad\notin T'$, and ($\acc\in T'$
    iff $\state'\in \FStates$);
  \item $(\state,\state')\in \TransR_\Gamma(\KS,\FStates)$, $\pad\notin T'$, and $\acc\in T'$;
  \item $(\state,\state')\in  \TransR^{\pad}_\Gamma(\KS)$, $\pad\notin T$, $\pad\in T'$, and ($\acc\in T'$
    iff $\state'\in \FStates$);
  \item $(\state,\state')\in \TransR^{\pad}_\Gamma(\KS,\FStates)$, $\pad\notin T$, $\pad\in T'$, and $\acc\in T'$.
   \end{itemize}
 \item $\Lab_\Gamma{(\state,T)}=\Lab(\state)\cup\{\pad\}$ if $\pad\in T$, and $\Lab_\Gamma{(\state,T)}=\Lab(\state)$ otherwise;
 \item $\FStates_\Gamma$ is the set of $\KS_\Gamma$-states $(\state,T)$ such that $\acc\in T$.
 \end{itemize} 

Intuitively, proposition $\#$ marks only the $\KS$-states which are targets of pairs in $\TransR^{\pad}_\Gamma(\KS)\cup \TransR^{\pad}_\Gamma(\KS,\FStates)$,
while the flag $\acc$ marks either the states in $\FStates$ which are targets of $\KS$-edges, or the $\KS$-states which are targets
of pairs in $\TransR_\Gamma(\KS,\FStates)\cup \TransR^{\pad}_\Gamma(\KS,\FStates)$. We say that
a trace $\trace$ over $\AP\cup\{\pad\}$ is \emph{well-formed} if one of the following conditions holds:
  \begin{itemize}
  \item \emph{or} $\trace$ is a $\Gamma$-stutter free trace over $\AP$ (\ie $\stfr_{\Gamma}(\trace)=\trace$);
  \item \emph{or} there is a position $i\geq 0$  such that $i+1$ is the unique position  where $\pad$ holds and
    $\stfr_{\Gamma}(\trace)= \trace(0)\ldots \trace(i)\cdot\trace^{i+2}$.
   \end{itemize}

By construction, it easily follows that  $\stfr_{\Gamma}^{\pad}(\Lang(\KS,\FStates))$ is the set of \emph{well-formed} traces in $\Lang(\KS_\Gamma,\FStates_\Gamma)$.
Then, the $\LTL$ formula $\theta_\Gamma$ captures the well-formed requirement and is defined as follows.
\[
\begin{array}{l}
\hspace{-0.2cm}\neg\pad \wedge \Always(\pad \rightarrow \Next\Always \neg\pad)\wedge \Always\Bigl(\displaystyle{\bigvee_{p\in \Gamma}} \bigl((p \leftrightarrow \neg \Next p)
\wedge \neg\pad \wedge \Next\neg\pad\bigr) \, \vee\,
\displaystyle{\bigwedge_{p\in \Gamma}} \Always \bigl((p \leftrightarrow \Next p)\wedge \neg\pad \bigr) \,\vee \\
\bigl(\Next\pad\wedge  \displaystyle{\bigwedge_{p\in \Gamma}}(p \leftrightarrow \Next p) \wedge  \displaystyle{\bigvee_{p\in \Gamma}} (p \leftrightarrow \neg \Next^2 \neg p)  \bigr)\Bigr)
\end{array}
\]
\noindent This concludes the proof of Proposition~\ref{prop:ExtendedStutterTrace}.
\end{proof} 

\end{document}